\documentclass[
paper=letter,%
numbers=noendperiod,%
captions=nooneline,%
DIV=10%
]{scrartcl}

\usepackage[T1]{fontenc}%
\usepackage{lmodern}%
\usepackage[american]{babel}%
\usepackage{microtype}%

\usepackage[
hyperref,%
table%
]{xcolor}%

\usepackage{scrlayer-scrpage}%

\usepackage{calc}%
\usepackage{rotating}%
\usepackage{amsmath}%
\usepackage{mathtools}%
\usepackage{amssymb}%
\usepackage{amsthm}%
\usepackage{thmtools}%
\usepackage{etoolbox}%
\usepackage{xparse}%
\usepackage{bm}%
\usepackage{bbm}%
\usepackage{enumitem}%
\usepackage[calc]{datetime2}%
\usepackage{graphicx}%
\usepackage{grffile}%
\usepackage{tikz}%
\usetikzlibrary{arrows.meta}
\usepackage{wrapfig}%
\usepackage{tabularx}%
\usepackage{siunitx}%
\usepackage{booktabs}%
\usepackage{multirow}%
\usepackage{fancyvrb}%
\usepackage{textcomp}%
\usepackage{listings}%
\usepackage{csquotes}%
\usepackage{pgfplots} \pgfplotsset{width=6cm,compat=1.9}%
\usepackage{dsfont} %
\usepackage[
style=authoryear,%
dashed=false,%
hyperref=true,%
useprefix=true,%
maxnames=2,%
uniquelist=minyear,%
maxbibnames=6,%
uniquename=false,%
seconds=true,%
date=iso,%
urldate=iso%
]{biblatex}%
\usepackage[
hypertexnames=false,%
setpagesize=false,%
pdfborder={0 0 0},%
pdfstartview=Fit,%
bookmarksopen=true,%
bookmarksnumbered=true%
]{hyperref}%
\usepackage[hyphens]{xurl}
\hypersetup{breaklinks=true}
\setkomafont{pageheadfoot}{\normalfont\normalcolor\sffamily}%
\setkomafont{pagenumber}{\normalfont\normalcolor\sffamily}%
\automark{section}%
\setkomafont{captionlabel}{\normalfont\normalcolor\sffamily\bfseries}%

\newcommand*{\mysquare}{\rule[0.18em]{0.36em}{0.36em}}
\newcommand*{\mytriangle}{\raisebox{0.12em}{\resizebox{0.48em}{0.48em}{$\blacktriangleright$}}}
\newcommand*{\mybar}{\rule[0.32em]{0.62em}{0.08em}}
\newcommand*{\mydot}{\raisebox{0.14em}{\resizebox{0.44em}{!}{$\bullet$}}}
\setlist{%
align=left,%
labelindent=0mm, %
leftmargin=!,%
itemindent=0mm, %
listparindent=\parindent,%
parsep=0mm,%
topsep=1mm,%
itemsep=1mm%
}
\setlist[itemize,1]{label={\mysquare}, labelwidth=\widthof{\mysquare\ }}%
\setlist[itemize,2]{label={\mytriangle}, labelwidth=\widthof{\mytriangle\ }}%
\setlist[itemize,3]{label={\mybar}, labelwidth=\widthof{\mybar\ }}%
\setlist[itemize,4]{label={\mydot}, labelwidth=\widthof{\mydot\ }}%
\setlist[enumerate,1]{label=\arabic*), labelwidth=\widthof{9)}}%
\setlist[enumerate,2]{label=\arabic{enumi}.\arabic*), labelwidth=\widthof{9.9)}}%
\setlist[enumerate,3]{label=\arabic{enumi}.\arabic{enumii}.\arabic*), labelwidth=\widthof{9.9.9)}}%
\setlist[enumerate,4]{label=\arabic{enumi}.\arabic{enumii}.\arabic{enumiii}.\arabic*), labelwidth=\widthof{9.9.9.9)}}%

\allowdisplaybreaks%
\newcommand*{\abstractnoindent}{}%
\let\abstractnoindent\abstract
\renewcommand*{\abstract}{\let\quotation\quote\let\endquotation\endquote
\abstractnoindent}
\deffootnote[1em]{1em}{1em}{\textsuperscript{\thefootnotemark}}%
\definecolor{blue}{RGB}{58, 95, 205}%
\definecolor{red}{RGB}{205, 41, 144}%
\definecolor{orange}{RGB}{238, 118, 0}%
\definecolor{chocolate}{RGB}{205, 102, 29}%

\lstdefinestyle{input}{
backgroundcolor=\color{black!12},%
commentstyle=\itshape\color{black!50},%
keywordstyle=\bfseries\color{black},%
stringstyle=\color{black}%
}

\lstdefinestyle{output}{
backgroundcolor=\color{black!6}%
}

\lstdefinestyle{codestyle}{
language={},%
keywords={},%
otherkeywords={}%
}

\lstnewenvironment{codeinput}[1][]{%
\lstset{style=input, style=codestyle}
#1%
}{\vspace{-0.25\baselineskip}}%

\lstnewenvironment{codeoutput}[1][]{%
\lstset{style=output, style=codestyle}
#1%
}{\vspace{-0.25\baselineskip}}%

\newcommand*{\code}{\lstinline[
basicstyle=\upshape\ttfamily,
style=codestyle,
literate={~}{{$\sim$}}1
]}

\expandafter\let\csname Sinput\endcsname\relax
\expandafter\let\csname endSinput\endcsname\relax
\expandafter\let\csname Soutput\endcsname\relax
\expandafter\let\csname endSoutput\endcsname\relax

\lstdefinestyle{Rstyle}{
language=R,%
keywords={},%
otherkeywords={}%
}

\lstnewenvironment{Sinput}[1][]{%
\lstset{style=input, style=Rstyle}
#1%
}{\vspace{-0.25\baselineskip}}%

\lstnewenvironment{Soutput}[1][]{%
\lstset{style=output, style=Rstyle}
#1%
}{\vspace{-0.25\baselineskip}}%

\lstdefinestyle{Cstyle}{
language=C,%
keywords={},%
otherkeywords={}%
}

\lstnewenvironment{Cinput}[1][]{%
\lstset{style=input, style=Cstyle}
#1%
}{\vspace{-0.25\baselineskip}}%

\lstnewenvironment{Coutput}[1][]{%
\lstset{style=output, style=Cstyle}
#1%
}{\vspace{-0.25\baselineskip}}%

\lstdefinestyle{Bashstyle}{
language=bash,%
keywords={},%
otherkeywords={}%
}

\lstnewenvironment{Bashinput}[1][]{%
\lstset{style=input, style=Bashstyle}
#1%
}{\vspace{-0.25\baselineskip}}%

\lstnewenvironment{Bashoutput}[1][]{%
\lstset{style=output, style=Bashstyle}
#1%
}{\vspace{-0.25\baselineskip}}%

\lstdefinestyle{LaTeXstyle}{
language=[LaTeX]TeX,%
texcs={},%
keywords={},%
otherkeywords={}%
}

\lstnewenvironment{LaTeXinput}[1][]{%
\lstset{style=input, style=LaTeXstyle}
#1%
}{\vspace{-0.25\baselineskip}}%

\lstnewenvironment{LaTeXoutput}[1][]{%
\lstset{style=output, style=LaTeXstyle}
#1%
}{\vspace{-0.25\baselineskip}}%

\DeclareNameAlias{sortname}{family-given}%
\DefineBibliographyExtras{american}{\DeclareQuotePunctuation{}}%
\renewbibmacro*{volume+number+eid}{%
\setunit*{\addcomma\space}%
\printfield{volume}%
\printfield{number}}
\DeclareFieldFormat*{number}{(#1)}
\DeclareFieldFormat*{title}{#1}%
\DeclareFieldFormat{doi}{%
\ifhyperref
{\href{http://dx.doi.org/#1}{\nolinkurl{doi:#1}}}%
{\nolinkurl{doi:#1}}}%
\renewbibmacro*{in:}{}%
\DeclareFieldFormat{isbn}{ISBN #1}%
\DeclareFieldFormat{pages}{#1}%
\DeclareFieldFormat{url}{\url{#1}}%
\DeclareFieldFormat{urldate}{\mkbibparens{#1}}%
\renewcommand*{\cite}[2][]{\textcite[#1]{#2}}%

\newif\ifstarttheorem
\declaretheoremstyle[%
spaceabove=0.5em,
spacebelow=0.5em,
headfont=\sffamily\bfseries\global\starttheoremtrue,
notefont=\sffamily\bfseries,
notebraces={(}{)},
headpunct={},
bodyfont=\normalfont,
postheadspace=\newline%
]{myMainStyle}
\declaretheorem[style=myMainStyle, numberwithin=section]{definition}%
\declaretheorem[style=myMainStyle, sibling=definition]{proposition}
\declaretheorem[style=myMainStyle, sibling=definition]{lemma}
\declaretheorem[style=myMainStyle, sibling=definition]{theorem}

\declaretheorem[style=myMainStyle, sibling=definition]{remark}
\declaretheorem[style=myMainStyle, sibling=definition]{example}

\makeatletter
\preto\itemize{%
\if@inlabel
\ifstarttheorem
\mbox{}\par\nobreak\vskip\glueexpr-\parskip-\baselineskip+0.25em\relax\hrule\@height\z@
\fi%
\fi%
\global\starttheoremfalse%
\def\tempa{proof}%
\ifx\tempa\mycurrenvir
\ifstarttheorem
\mbox{}\par\nobreak\vskip\glueexpr-\parskip-\baselineskip+0.25em\relax\hrule\@height\z@
\fi%
\fi%
\global\starttheoremfalse%
}
\preto\enditemize{\global\starttheoremfalse}
\makeatother

\makeatletter
\preto\enumerate{%
\if@inlabel
\ifstarttheorem
\mbox{}\par\nobreak\vskip\glueexpr-\parskip-\baselineskip+0.25em\relax\hrule\@height\z@
\fi%
\fi%
\global\starttheoremfalse%
\def\tempa{proof}%
\ifx\tempa\mycurrenvir
\ifstarttheorem
\mbox{}\par\nobreak\vskip\glueexpr-\parskip-\baselineskip+0.25em\relax\hrule\@height\z@
\fi%
\fi%
\global\starttheoremfalse%
}
\preto\endenumerate{\global\starttheoremfalse}
\makeatother

\makeatletter
\NewDocumentCommand{\tmb}{O{0.1mm} O{0.1mm} O{0.88} m m m}{%
\mathrel{%
	\vbox{\offinterlineskip\m@th
		\ialign{%
			\hfil##\hfil\cr
			$\scriptscriptstyle\text{\scalebox{#3}{#4}}\mathstrut$\cr%
			\noalign{\vspace{#1}}%
			\vtop{%
				\ialign{%
					\hfil##\hfil\cr
					$#5$\cr\noalign{\vspace{#2}}%
					$\scriptscriptstyle\text{\scalebox{#3}{#6}}\mathstrut$\cr%
				}%
			}\cr
		}%
	}%
}%
}
\makeatother

\makeatletter
\NewDocumentCommand{\tmbc}{O{0.1mm} O{0.1mm} O{0.88} m m m}{%
\mathrel{%
	\vbox{\offinterlineskip\m@th
		\ialign{%
			\hfil##\hfil\cr
			$\scriptscriptstyle\mathclap{\text{\scalebox{#3}{#4}}}\mathstrut$\cr%
			\noalign{\vspace{#1}}%
			\vtop{%
				\ialign{%
					\hfil##\hfil\cr
					$#5$\cr\noalign{\vspace{#2}}%
					$\scriptscriptstyle\mathclap{\text{\scalebox{#3}{#6}}}\mathstrut$\cr%
				}%
			}\cr
		}%
	}%
}%
}
\makeatother

\usetikzlibrary{shapes.geometric}

\newcommand*{\isim}{\tmb{\tiny{ind.}}{\sim}{}}
\newcommand*{\deq}{\tmb{\tiny{d}}{=}{}}

\newcommand*{\IN}{\mathbb{N}}

\newcommand*{\IR}{\mathbb{R}}

\newcommand*{\U}{\operatorname{U}}
\newcommand*{\B}{\operatorname{B}}

\newcommand*{\N}{\operatorname{N}}
\renewcommand*{\t}{\textit{t}}

\newcommand*{\I}{\mathbbm{1}}
\newcommand*{\rd}{\mathrm{d}}
\renewcommand*{\mod}{\operatorname{mod}}
\newcommand*{\ran}{\operatorname{ran}}

\newcommand*{\supp}{\operatorname{supp}}

\renewcommand*{\P}{\mathbb{P}}
\newcommand*{\E}{\mathbb{E}}

\newcommand*{\ARMA}{\operatorname{ARMA}}
\newcommand*{\GARCH}{\operatorname{GARCH}}

\newcommand*{\W}{\mathcal{W}}
\newcommand*{\Wg}{\W_{\text{g}}}

\newcommand*{\R}{\textsf{R}}

\begin{document}
\thispagestyle{plain}
\begin{center}
  \sffamily
  {\bfseries\LARGE W-transforms: Uniformity-preserving transformations and induced dependence structures\par}
  \bigskip\smallskip
  {\Large Marius Hofert\footnote{Department of Statistics and Actuarial Science, The University of
      Hong Kong,
      \href{mailto:mhofert@hku.hk}{\nolinkurl{mhofert@hku.hk}}.},
    Zhiyuan Pang\footnote{Department of Statistics and Actuarial Science, The University of
      Hong Kong,
      \href{mailto:markus.pang@connect.hku.hk}{\nolinkurl{markus.pang@connect.hku.hk}}.}
    \par
    \bigskip
    \today\par}
\end{center}
\par\smallskip
\begin{abstract}
  W-transforms are introduced as uniformity-preserving univariate
  transformations on the unit interval induced by distribution functions and
  piecewise strictly monotone functions, and their properties are
  investigated. When applied componentwise to random vectors with standard
  uniform univariate margins, W-transforms naturally serve as copula-to-copula
  transformations. Properties of the resulting W-transformed copulas, including
  their analytical form, density, measures of concordance, tail dependence and
  symmetries, are derived. A flexible parametric family of W-transforms is
  proposed as a special case to further enhance tractability. Illustrative
  examples highlight the introduced concepts, and improved dependence modelling
  is demonstrated in terms of a real-life dataset.
\end{abstract}
\minisec{Keywords}
Copula-to-copula transformation, Danube data, invariance principle, piecewise
strictly monotone functions, tractable copula construction, v-transform.

\section{Introduction}
Transformations $\mathcal{T}:[0,1]\to[0,1]$ are \emph{uniformity-preserving} if
$\mathcal{T}(U)\sim\U(0,1)$ for $U\sim\U(0,1)$. Such transformations were considered, for example,
by \cite{porubskysalatstrauch1988}, who showed that $\mathcal{T}$ is
uniformity-preserving if and only if $\E\bigl(h(\mathcal{T}(U))\bigr)=\E(h(U))$ for all
Riemann integrable $h:[0,1]\to\IR$. \cite{strauchporubsky1993} considered the multivariate
case and showed that for $U_1,\dots,U_d\mathrel{\smash{\isim}}\U(0,1)$, $\mathcal{T}_1,\dots,\mathcal{T}_d:[0,1]\to[0,1]$
are jointly uniformity-preserving, that is $(T_1(U_1), \dots, T_d(U_d)) \sim \U(0,1)^d$,
if and only if $\E\bigl(h(\mathcal{T}_1(U_1),\dots,\mathcal{T}_d(U_d))\bigr)=\E(h(U_1,\dots,U_d))$
for all Riemann integrable $h:[0,1]^d\to\IR$. The equivalent result for bounded and continuous $h$
is a direct consequence of the Portmanteau lemma; see \cite[Lemma~2.2]{vandervaart2000}.

Uniformity-preserving transformations also naturally appear in the context of
copula-to-copula transformations, such as the transformations of
\cite{rosenblatt1952} (or its inverse), \cite{khoudraji1995}, \cite{morillas2005},
\cite{liebscher2008}, \cite{durantesarkocisempi2009}, %
\cite[Section~2.7]{hofertkojadinovicmaechleryan2018}, %
and others, with the goal to construct new, tailor-made dependence models from
given ones. More recently, uniformity-preserving transformations $\mathcal{T}$
of the form $\mathcal{T}(u)=F_{T(X)}\bigl(T(F_X^{-1}(u))\bigr)$, $u\in[0,1]$ for
$X\sim F_X$ with quantile function $F_X^{-1}(u)=\inf\{x\in\IR:F_X(x)\ge u\}$,
$u\in[0,1]$, and transformations $T:\IR\to\IR$ were considered in
\cite{mcneil2021} under the name of ``v-transforms'' in the context of modelling
volatile time series and under specific assumptions on both $F_X$ and $T$
(detailed later). \cite{quessy2024general} considered piecewise monotone
transformations as uniformity-preserving transformations and applied them to the
components of copulas for the purpose of multivariate analysis, nonmonotone
regression, and modelling spatial dependence.

While previous work successfully modelled exchangeable dependence, real-life data often
exhibit non-exchangeability. For example, the Danube dataset, see left-hand side
of Figure~\ref{fig:danube}, from the \R\ package \texttt{lcopula} of
\cite{belzile2023} (659 pseudo-observations of base flows measured at Sch\"arding
in Austria and Nagymaros in Hungary) violates exchangeability since measurements
from Sch\"arding (upstream on the Inn River) show systematically larger
base flows than at Nagymaros (downstream on the Danube).
\begin{figure}[htbp]
  \centering
  \includegraphics[width=0.48\textwidth]{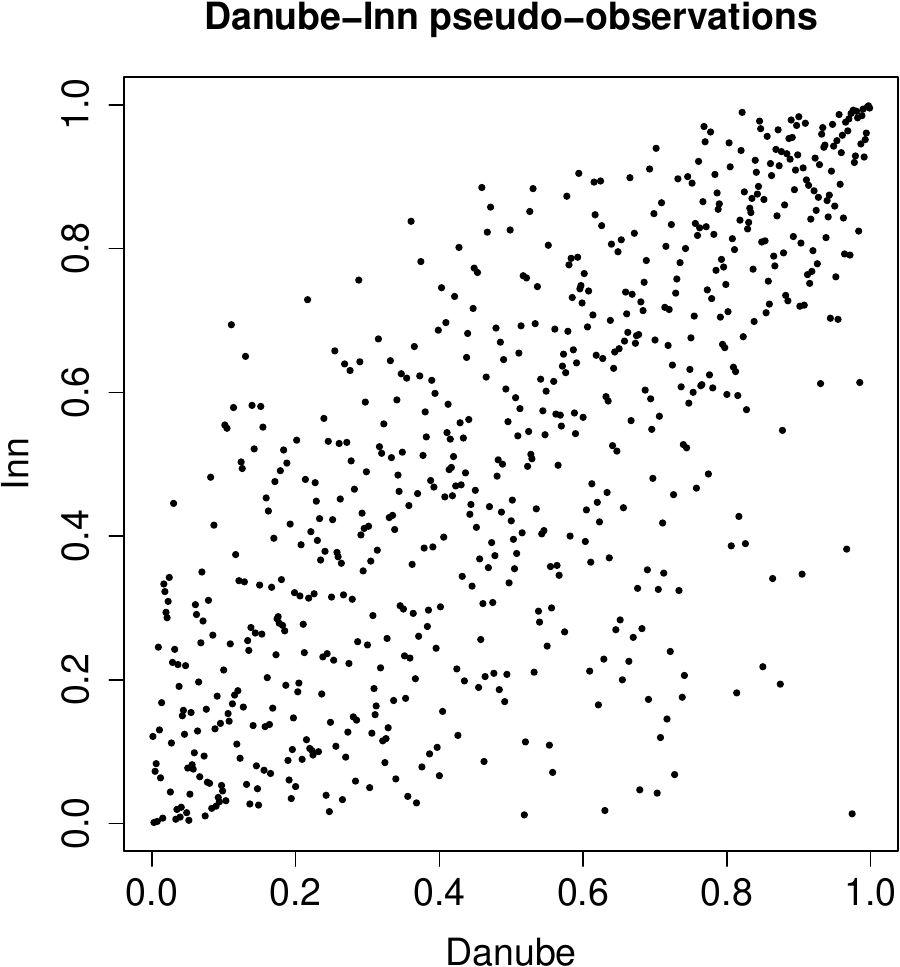}
  \hfill
  \includegraphics[width=0.48\textwidth]{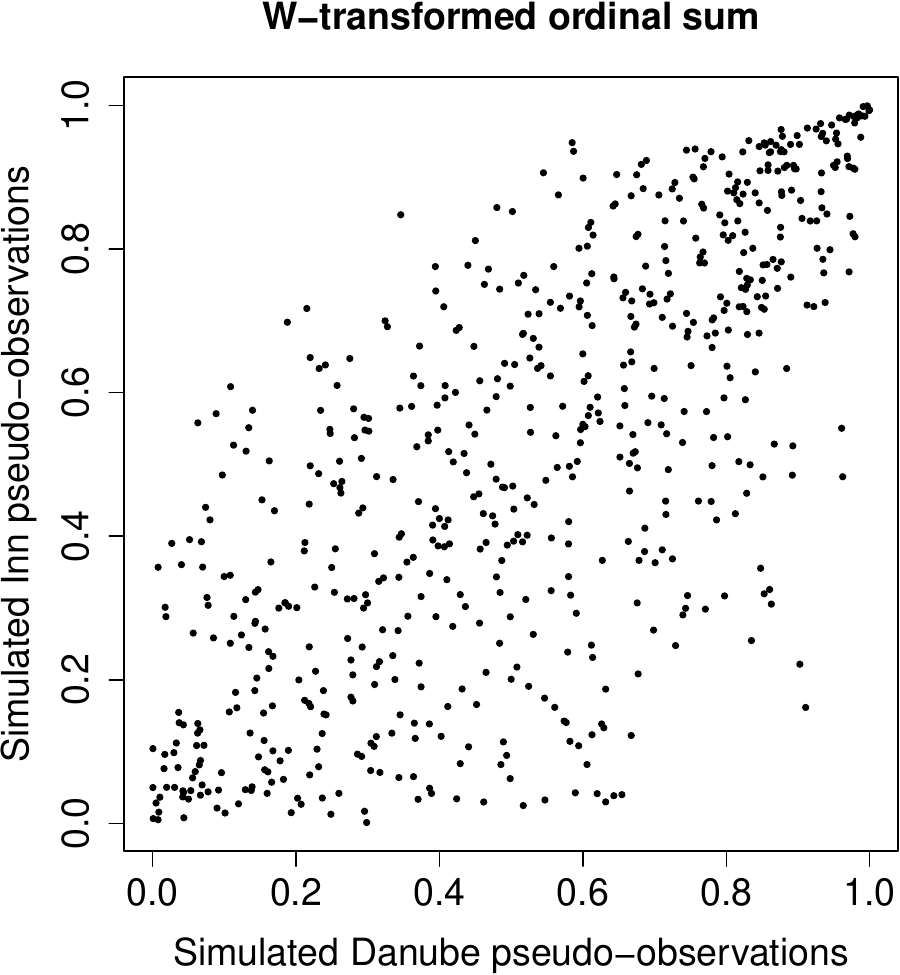}
  \caption{659 pseudo-observations of the Danube dataset of \cite{belzile2023} (left)
    and generated sample of the same size of a model constructed based on W-transforms
    introduced later (right).}
  \label{fig:danube}
\end{figure}
This asymmetry, driven by upstream-downstream dynamics, is critical for
accurately characterising joint base-flow behaviour, which symmetric copula
models fail to capture. While devices like those proposed in
\cite{khoudraji1995}, \cite{liebscher2008}, and \cite{freesvaldez1998} partially
address non-exchangeability, their flexibility remains limited. The model we
propose in this work based on W-transforms is more flexible and produces samples
that closely resemble the Danube dataset (right-hand side of
Figure~\ref{fig:danube}).

We conclude this introduction by positioning W-transforms relative to vine copulas,
introduced and detailed in \cite{bedford2002vines}, \cite{aas2009pair}, currently the most widely
used framework for high-dimensional non-exchangeable dependence modelling. Vine copulas
achieve flexibility through a hierarchical pairwise construction but typically invoke
the simplifying assumption that conditional copulas do not depend on the conditioning
values to remain tractable in higher dimensions. As shown by \cite{stoeber2013simplified}, this
assumption makes widely used copula families such as the Archimedean families (except
Clayton) and mixtures of normals (except student-t copula) outside the range of vine
copulas. W-transforms, as we shall see in Section~\ref{sec:cop} and Section~\ref{sec:application},
operate on a base copula $C$, and therefore any valid copula can serve as the dependence
skeleton. Moreover, the W-transformed ordinal sums constructed via \eqref{eq:ordinal:sum},
\eqref{eq:os:homo:in:W:cop}, and \eqref{eq:os:W:cop} provide closed-form density
and tail dependence (Proposition~\ref{prop:os:tail:dependence}) without recursive conditioning.
We view W-transformed copulas as complementary to vine copulas that is particularly
useful when asymmetry has a physical explanation (such as the Danube data).

The paper is organised as follows. In Section~\ref{sec:gen:W:transform} we
introduce the notion of ``W-transforms''. %
Their properties are thoroughly investigated in Sections~\ref{sec:cont:W:trafo}
and~\ref{sec:general:W:trafo}. In Section~\ref{sec:cop}, we then focus on
``W-transformed copulas'', that is the copulas implied by marginally applying
``W-transforms''. In Section~\ref{sec:application}, we demonstrate how
``W-transforms'' can generate flexible tail dependencies and
non-exchangeability, with applications to the aforementioned Danube
dataset. Conclusions and proofs are provided in Section~\ref{sec:concl} and the
appendix, respectively.

\section{The notion of a W-transform}\label{sec:gen:W:transform}
Let $\bar{\IR}=\IR\cup\{-\infty,\infty\}$ and
$\bar{\IN}=\IN\cup\{\infty\}$. %
For $K\in\bar{\IN}$, \emph{change points} $(t_k)_{k=0}^K$ are points satisfying
$-\infty\le t_0<t_1<\dots<t_k<\cdots\le\infty$. The $K$ intervals $(t_{k-1},t_k)$,
$k=1,\dots,K$, are referred to as %
\emph{pieces}.
Let $D\coloneqq [t_0,\infty)$ if $t_K=\infty$ and $[t_0,t_K]$ otherwise.
We call $T:D\to\IR$ \emph{piecewise continuous and strictly monotone (pcsm)}
with $K$ pieces and change points $(t_k)_{k=0}^K$, if the restriction
$T|_{(t_{k-1},t_k)}$ is continuous and strictly monotone for all
$k\in\{1,\dots,K\}$, where we interpret $T(-\infty)$ as
$\lim_{x\to-\infty} T(x)$ and $T(\infty)$ as $\lim_{x\to\infty} T(x)$. The case
$K=\infty$ is included to allow for countably infinitely many pieces.

We can now introduce the notion of a W-transform as follows.
\begin{definition}[W-transforms]\label{def:W:transform}
  For pcsm $T:D\to\IR$ and $X$ following a \emph{base distribution} $F_X$, let
  $\supp(F_X)\coloneqq\{x\in\IR: F_X(x)-F_X(x-h)>0\ \forall\,h>0\}$ be the support of
  $F_X$ with $\inf\supp(F_X)=t_0$ and $\sup\supp(F_X)=t_K$. %
  For $X\sim F_X$ and $K\in\bar{\IN}$, let $(t_k)_{k=0}^K$ be change points of $T$.
  The \emph{W-transform} $\W:[0,1]\to[0,1]$ of $F_X$ and $T$ is then defined by
  \begin{align}\label{eq:W:transform}
    \W(u)=\begin{cases}
      \lim_{u\to 0+}\W(u),& u=0,\\
      F_{T(X)}\bigl(T(F_X^{-1}(u))\bigr),& u\in(0,1],
    \end{cases}
  \end{align}
  where $T(X)$ follows the \emph{transformed distribution} $F_{T(X)}$. As we
  shall see in Proposition~\ref{prop:W:property}, $\W$ is also pcsm with
  \emph{change points} $\delta_k$, $k\in\{1, \dots, K\}$.
\end{definition}
\begin{remark}[Technical details]
  \begin{enumerate}
  \item $F_X^{-1}(1)=\sup\supp(F_X)=x_{F_X}$ is the \emph{right endpoint} of
    $F_X$.  By assumption $t_K= x_{F_X}$, so $\W(1)$ is
    well-defined. $\W(0)$ is defined as a limit since, otherwise, we
    would need, for all distributions with \emph{left endpoint}
    $\sup\{x\in\IR:F_X(x)=0\}>-\infty$, to be able to define $T(-\infty)$ and
    thus have to choose $t_0=-\infty$ just for the purpose of defining
    $\W(0)$ (but values of $\W(u)$ on a Lebesgue null set will
    not affect uniformity-preservation).

    Note that we do not know the value $\W(0)$ or $\W(1)$ in
    general. For the latter, $\W(1)=F_{T(X)}(T(x_{F_X}))=\P(T(X)\le T(x_{F_X}))$,
    but this can take any value in $[0,1]$ depending on $T$. If $T$ is strictly
    increasing (decreasing), it is $1$ ($0$).

  \item For uniformity-preservation to hold, $T$ cannot be constant $y$ on
    any interval $[s_1,s_2]\subseteq\supp(F_X)$ with $s_1<s_2$ as then
    $F_{T(X)}(z)-F_{T(X)}(z-)=\P(T(X)=z)\ge\P(X\in[s_1,s_2])>0$ %
    so $F_{T(X)}$ jumps in $z$ and thus $\W$ cannot be uniformity-preserving
    since $\W$ does not attain any values in $(F_{T(X)}(z-),F_{T(X)}(z))$.
  \item As we shall see, $\W$ in \eqref{eq:W:transform} is uniformity preserving if $F_X$
    is continuous. If $F_X$ is discontinuous, the W-transform $\W$ is not uniformity-preserving; see
    Example~\ref{eg:non:unif:preserv}. In Section~\ref{sec:general:W:trafo}, we extend the definition
    of W-transforms to discontinuous $F_X$.
  \item A pcsm $T$ allows us to treat fairly general functions while being able
    to identify conditions on $F_X$ (in combination with $T$) that guarantee
    uniformity-preservation. Technically, we allow that
    $\lim_{t\to t_k-}T(t)<T(t_k)<\lim_{t\to t_k+}T(t)$ at all finite $t_k$ as
    the value of $T$ at these at most countably many points is irrelevant for
    the question of uniformity-preservation under continuous $F_X$ due to forming a Lebesgue null
    set. Thus, if $F_X$ is continuous, we can assume without loss of generality
    that $T$ is left-continuous at all finite change points.
  \item The base distribution $F_X$ and the transformation $T$ are auxiliary generators of
    the W-transform $\W$. Different pairs of $(F_X,T)$ mat lead to the same function $\mathcal{W}$ on $[0,1]$
    and, therefore, the pair is not jointly identifiable. For example, one may take $X\sim\U(0,1)$,
    $T_1(x)=|1-2x|$,  $T_2(x)=(1-2x)^2$, $x\in[0,1]$ to obtain $\mathcal{W}_1(u)=\mathcal{W}_2(u)=|1-2u|$,
    $u\in[0,1]$. In applications, one either fixes $F_X$ (e.g, to $\U(0,1)$) and works with a parametric family
    for $T$, or parameterise $\W$ directly through its properties as shown in Proposition~\ref{prop:W:property}.
    The latter has been applied in Section~\ref{subsec:illu:real:data} in the construction of
    \eqref{eq:W:trafo:Inn} where we chose the first piece and recovered the second piece via the property
    Proposition~\ref{prop:W:property}~\ref{prop:W:property:partition:sq}.
  \end{enumerate}
\end{remark}

As the following examples show, the generic form of a W-transform does not
necessarily imply that $\W$ is uniformity-preserving.
\begin{example}[Non-uniformity-preservation of generic W-transforms]\label{eg:non:unif:preserv}
  Let $X\sim\B(1,p)$, $p\in[0,1]$.
  \begin{enumerate}
  \item If $p=0$, then $X=0$ almost surely (a.s.), so that
    $F_X(x)=\I_{[0,\infty)}(x)$, $x\in\IR$, with
    $F_X^{-1}(u)=0$, $u\in(0,1]$. Therefore,
    $\W(u)=F_{T(X)}\bigl(T(F_X^{-1}(u))\bigr)=F_{T(0)}(T(0))=1$,
    $u\in(0,1]$, %
    which is not uniformity-preserving. Similarly for $p=1$, $X=1$ a.s.,
    $F_X(x)=\I_{[1,\infty)}(x)$, $x\in\IR$, with $F_X^{-1}(u)=1$, $u\in(0,1]$,
    and thus $\W(u)=F_{T(1)}(T(1))=1$, $u\in(0,1]$. Note that, in both cases,
    $T$ is only utilised in a single point.
  \item If $p\in(0,1)$, then $F_X(x)=(1-p)\I_{[0,\infty)}(x)+p\I_{[1,\infty)}$,
    $x\in\IR$, with quantile function $F_X^{-1}(u)=\I_{(1-p,1]}(u)$,
    $u\in(0,1]$. With stochastic representation $X\deq F_X^{-1}(U)=\I_{(1-p,1]}(U)$
    for $U\sim\U(0,1)$, we obtain
    $\W(u)=F_{T(X)}\bigl(T(F_X^{-1}(u))\bigr)=\P\bigl(T(\I_{(1-p,1]}(U))\le T(\I_{(1-p,1]}(u))\bigr)$, $u\in(0,1]$.
    Therefore, for $u\in(0, 1]$,
    \begin{align*}
      \W(u)=\begin{cases}
        1-(1-p)\I_{(1-p,1]}(u),& \text{if}\ T\ \text{is strictly decreasing},\\
        1-p\I_{(0,1-p]}(u),& \text{if}\ T\ \text{is strictly increasing},
      \end{cases}
    \end{align*}
    and neither case leads to a uniformity-preserving $\W$.
  \end{enumerate}
\end{example}

\section{W-transforms constructed from continuous random variables}\label{sec:cont:W:trafo}
As already applied, the \emph{quantile transform} $F_X^{-1}(u)$ satisfies
$F_X^{-1}(U)\deq X$ for $U\sim \U(0, 1)$; see, for example, \cite{embrechtshofert2013c}.
In this section we consider W-transforms under continuous base distributions $F_X$,
in which case \emph{probability transform} $F_X(X)$ satisfies $F_X(X)\sim \U(0, 1)$.

\subsection{Uniformity-preservation}
Our first result shows that W-transforms with continuous base distributions $F_X$ are
uniformity-preserving.
\begin{proposition}[Uniformity-preservation under continuous $F_X$]\label{prop:w:unif:preserv}
  Let $X\sim F_X$ for continuous base distribution $F_X$, and let $T:D\to\IR$ be pcsm
  with change points $(t_k)_{k=0}^K$, $K\in\bar{\IN}$. If $U\sim\U(0,1)$, then $\W(U)\sim\U(0,1)$.
\end{proposition}

The following example addresses ``v-transforms'', a special case of
uniformity-preserving W-transforms considered in~\cite{mcneil2021}.
\begin{example}[V-transforms and their use in \cite{mcneil2021}]\label{eg:v:trans:use}
  \begin{enumerate}
  \item \cite{mcneil2021} considered \emph{v-transforms} (denoted by
    $\mathcal{V}$), which are W-transforms of the form
    $\mathcal{T}(u)=F_{T(X)}\bigl(T(F_X^{-1}(u))\bigr)$, $u\in[0,1]$, for
    absolutely continuous $F_X$ with density $f_X$ symmetric about
    $0$ %
    and continuous and differentiable transformations $T:\IR\to[0,\infty)$ that
    are, for change points $t_0=-\infty$, $t_1$ and $t_2=\infty$, strictly decreasing on
    $(-\infty,t_1]$, strictly increasing on $[t_1,\infty)$ and satisfy
    $T(\t_1)=0$. The point $\delta=F_X(t_1)$ is the \emph{fulcrum} of the
    v-transform, and, due to its intended application, $T$ is called
    \emph{volatility proxy transformation}.  For $T(x)=|x|$, one has
    $\mathcal{V}(u)=|2u-1|$, $u\in[0,1]$, which is of v-shape and piecewise
    linear.
  \item\label{eg:v:trans:use:chara:v} \cite[Theorem~1, Proposition~3]{mcneil2021} shows that $\mathcal{V}:[0,1]\to[0,1]$
    is a v-transform if and only if
    \begin{align}
      \mathcal{V}(x)=\begin{cases}
        (1-x)-(1-\delta)G(\frac{x}{\delta}), & x\in[0,\delta],\\
        x-\delta G^{-1}(\frac{1-x}{1-\delta}), & x\in(\delta,1],
      \end{cases}\label{eq:vtrans}
    \end{align}
    for a continuous and strictly increasing distribution function $G$ on
    $[0,1]$, referred to as the \emph{generator} of $\mathcal{V}$. In particular,
    \cite{mcneil2021} considered the two-parameter family of distribution
    functions $G(x) = \exp(-\kappa(-\ln x)^{\xi})$ where
    $\kappa = 2$, $\xi = 0.5$, $\delta = 0.4$ in \eqref{eq:vtrans}; see the
    left-hand side of Figure~\ref{fig:v:trans:relation}. We can see that the
    v-transform has two strictly monotone branches, and the graph resembles the
    letter ``v'', hence the name.
    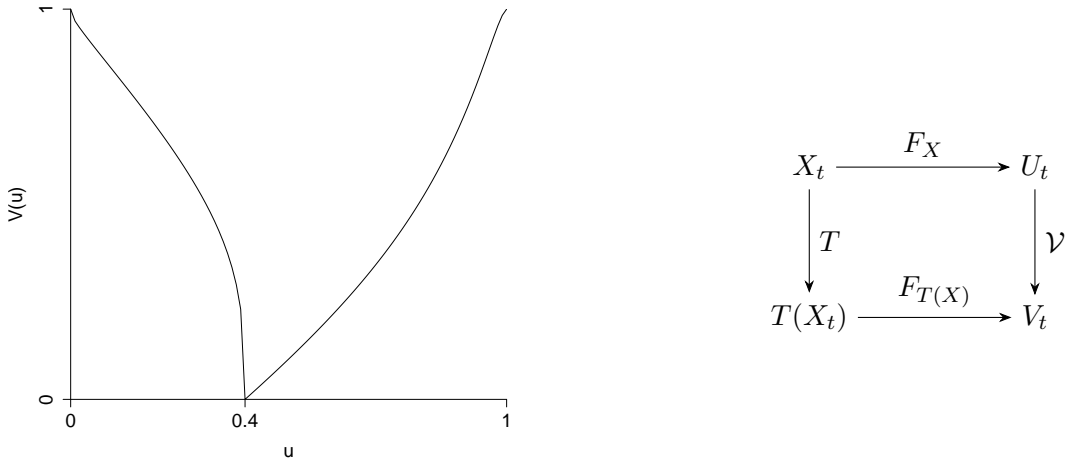
\begin{figure}[htbp]
      \centering
      \begin{minipage}[c]{0.48\textwidth}
        \centering
        \includegraphics[width=\textwidth]{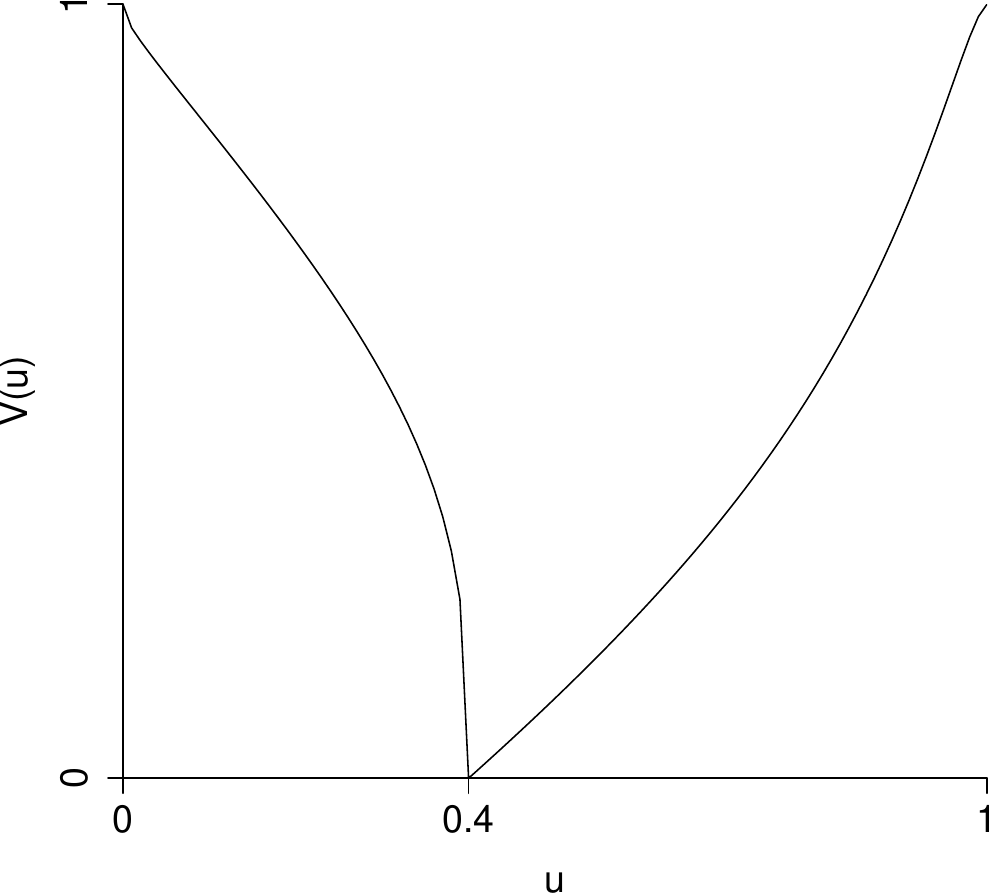}
      \end{minipage}
      \hfill
      \begin{minipage}[c]{0.48\textwidth}
        \centering
        \begin{tikzpicture}[node distance=2cm]
          \node (Xt) at (0, 1) {$X_t$};
          \node (Ut) at (3, 1) {$U_t$};
          \node (TXt) at (0, -1) {$T(X_t)$};
          \node (Vt) at (3, -1) {$V_t$};

          \draw[-Stealth] (Xt) -- (Ut) node[midway, above] {$F_X$};
          \draw[-Stealth] (TXt) -- (Vt) node[midway, above] {$F_{T(X)}$};
          \draw[-Stealth] (Ut) -- (Vt) node[midway, right] {$\mathcal{V}$};
          \draw[-Stealth] (Xt) -- (TXt) node[midway, right] {$T$};
        \end{tikzpicture}
      \end{minipage}
      \caption{A v-transform (left) and the conceptual relationship between $X_t$, $U_t$,
        $T(X_t)$ and $V_t$ (right).}\label{fig:v:trans:relation}
    \end{figure}
  \item As v-transforms have no ordinary inverse,
    \cite{mcneil2021} %
    considered stochastic (that is randomised) inverse v-transforms for the purpose of
    constructing competitive alternatives to $\GARCH$ time series models. Inspired by the
    fact that a $\GARCH(p,q)$ process $(X_t)_{t\in\IN}$, when squared, is an
    $\ARMA(p,q)$ process, the idea is to construct a new, symmetric and strictly
    stationary stochastic process $(X_t)_{t\in\IN}$ with given absolutely continuous
    margin $F_X$ and even density $f_X$ (for example from a Laplace distribution or Student's
    $t$-distribution), such that the \emph{volatility proxy series}
    $(T(X_t))_{t\in\IN}$ for even $T$ (such as $T(x)=x^2$
    or $T(x)=|x|$) is an $\ARMA$ process $(Z_t)_{t\in\IN}$. This can be done as follows:
    \begin{enumerate}[label=(\arabic*), labelwidth=\widthof{(2)}]
    \item Construct the \emph{normalised volatility proxy series},
      that is a causal and invertible $\ARMA$ process $(Z_t)_{t\in\IN}$ with
      standardised innovation distribution, without loss of generality $\N(0, 1)$.
    \item\label{mcneil:vol:model:algo:vol:PIT:process} Construct the \emph{volatility PIT process}
      $(V_t)_{t\in\IN}=(\Phi(Z_t))_{t\in\IN}$.
    \item Construct the \emph{series PIT process}
      $(U_t)_{t\in\IN}$ from $(V_t)_{t\in\IN}$ via
      $U_t=\mathcal{V}^{-1}(V_t)$, where $\mathcal{V}^{-1}$ is the stochastic
      inverse of the v-transform $\mathcal{V}$.
    \item\label{mcneil:vol:model:algo:last:step} Construct $(X_t)_{t\in\IN}$ via $(X_t)_{t\in\IN}=(F_X^{-1}(U_t))_{t\in\IN}$.
    \end{enumerate}
    The volatility PIT process $(V_t)_{t\in\IN}=(\Phi(Z_t))_{t\in\IN}$ in
    Step~\ref{mcneil:vol:model:algo:vol:PIT:process} by construction equals
    $F_{T(X)}(T(X_t))$ derived in
    Step~\ref{mcneil:vol:model:algo:last:step}. V-transforms therefore
    characterise the copula of $(U_t,V_t)$, which, by the invariance principle,
    is also that of $(X_t,T(X_t))$. This conceptual relationship is illustrated
    on the right-hand side of Figure~\ref{fig:v:trans:relation}.
  \end{enumerate}
\end{example}

\subsection{Properties of W-transforms}
In this section, we study properties of W-transforms.
We start with the following, fundamental ones.
\begin{proposition}[Properties of W-transforms]\label{prop:W:property}
  Let $X\sim F_X$ be continuous and $T: D \rightarrow \IR$ be pcsm and
  left-continuous. A W-transform $\W$ defined by
  \eqref{eq:W:transform} then has the following properties:
  \begin{enumerate}
  \item\label{prop:W:property:change:pt}  $\W$ has change points at
    $\delta_0=0<\delta_1<\dots<\delta_k<\cdots<1=\delta_K$ with $\delta_k=F_X(t_k)$,
    $k\in\{1,\dots,K\}$.
  \item\label{prop:W:property:monotone} $\W$ has the same monotonicity in
    $(\delta_{k-1},\delta_k]$ as $T$ has in $(t_{k-1}, t_k]$ for any $k\in\{1,\dots,K\}$.
    If $T$ is continuous everywhere, then so is $\W$.
  \item\label{prop:W:property:partition:sq} \emph{Partition of square property.}
    Consider $v\in [0,1]$. Define the preimage sets restricted
    to $(\delta_{k-1}, \delta_k]$ as $S_k(v)=\{u\in(\delta_{k-1}, \delta_k]:
    \W(u)\leq v)\}$ for all $k\in\{1,\dots,K\}$.
    Then, for the Lebesgue measure $\lambda$,
    \begin{align}\label{eq:partition:sq:property}
      \lambda\biggl(\,\biguplus_{k=1}^{K} S_k(v)\biggr) = v.
    \end{align}
    For each $k\in\{1,\dots,K\}$, consider the bijective piece
    $\W_{|k}\coloneqq\W|_{(\delta_{k-1}, \delta_k]}$. Then $S_k(v)$ is of the following form:
    \begin{enumerate}[label=\roman*), labelwidth=\widthof{iii)}]
    \item If $\inf \W_{|k}>v$, then $S_k(v)=\emptyset$.
    \item If $\sup \W_{|k}\leq v$, then $S_k(v)=(\delta_{k-
        1},\delta_k]$.
    \item Otherwise, if $\W_{|k}$ is increasing then $S_k(v) =
      (\delta_{k-1}, \W^{-1}_{|k}(v)]$, and if $\W_{|k}$ is decreasing then
      $S_k(v)=(\W^{-1}_{|k}(v), \delta_k]$.
    \end{enumerate}
  \end{enumerate}
\end{proposition}
Proposition~\ref{prop:W:property}~\ref{prop:W:property:change:pt} and
\ref{prop:W:property:monotone} ensure that $\W$ is pcsm, and the
partition of square property ensures that $\W$ is uniformity-preserving.

The following result shows that the properties listed in Proposition~\ref{prop:W:property} are
closed under composition of W-transforms; we also apply this result later when considering
``periodic'' W-transforms.
\begin{proposition}[Composition of W-transforms preserves properties of W-transforms]\label{prop:W:compo}
  Let $\W'$ and $\W''$ be W-transforms constructed from continuous base distributions
  with change points $\{\delta_k'\}_{k=0}^{K'}$ and $\{\delta_\ell''\}_{\ell=0}^{K''}$
  where $K', K''\in \bar{\IN}$. Then, the composition $\W=\W'\circ\W''$
  is uniformity-preserving and pcsm.
\end{proposition}

We now present some examples of W-transforms, with one featuring an illustration of the partition
of square property.
\begin{example}[W-transforms]\label{eg:W:trafo}
  \begin{enumerate}
  \item\label{eg:W:trafo:shuffle} \emph{Shuffle of identity.} Let $X\sim \U(0, 1)$ and
    \begin{align*}
      T(x)=\begin{cases}
        -x + 1, & x\in[0, \frac{1}{3}],\\
        x, & x\in(\frac{1}{3}, \frac{2}{3}],\\
        x-\frac{2}{3}, & x\in(\frac{2}{3}, 1],
      \end{cases}
    \end{align*}
    with change points $t_0=0$, $t_1=1/3$, $t_2=2/3$, $t_3=1$. Then the functional form
    of \eqref{eq:W:transform} is $\W(u)=T(u)$. A plot of $\W_1\coloneqq \W$
    is shown in Figure~\ref{fig:example:W:trans} (top-left) and one sees that
    $\W_1$ is a shuffle-and-reorder of strips of the identity on $[0,1]$,
    reminiscent of the construction of shuffle-of-min; see \cite{durante2009}.
    \begin{figure}[htbp]
      \includegraphics[width=0.48\textwidth]{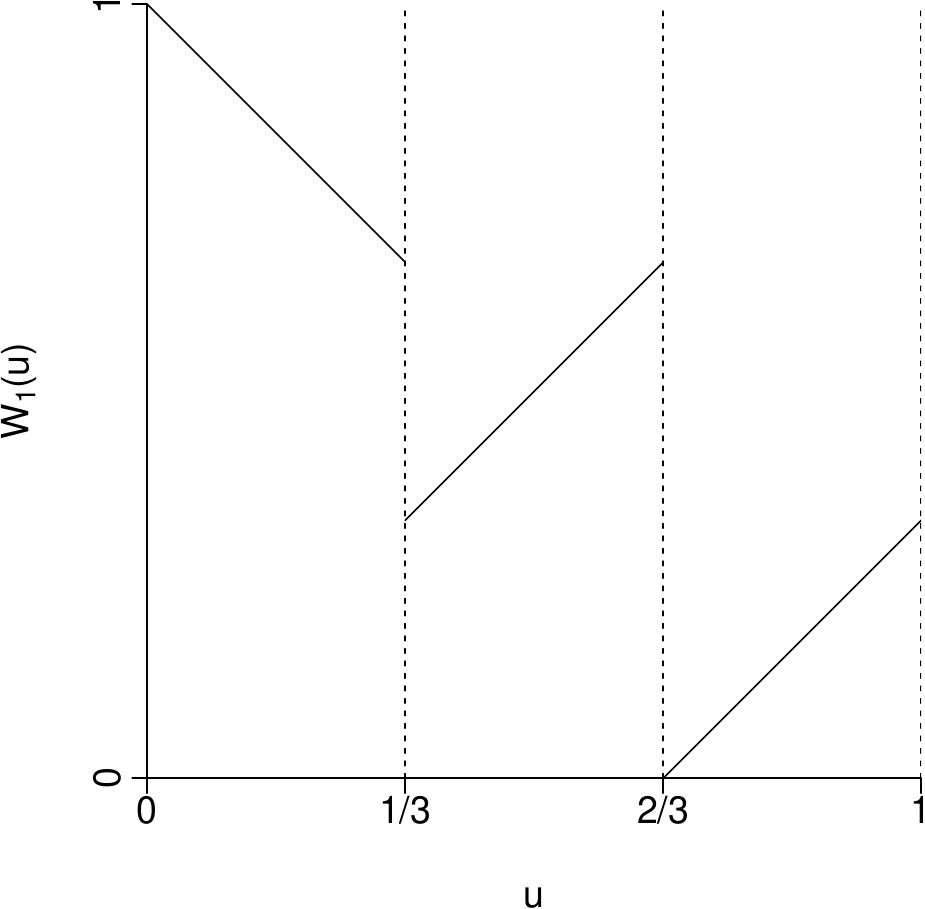}
      \hfill
      \includegraphics[width=0.48\textwidth]{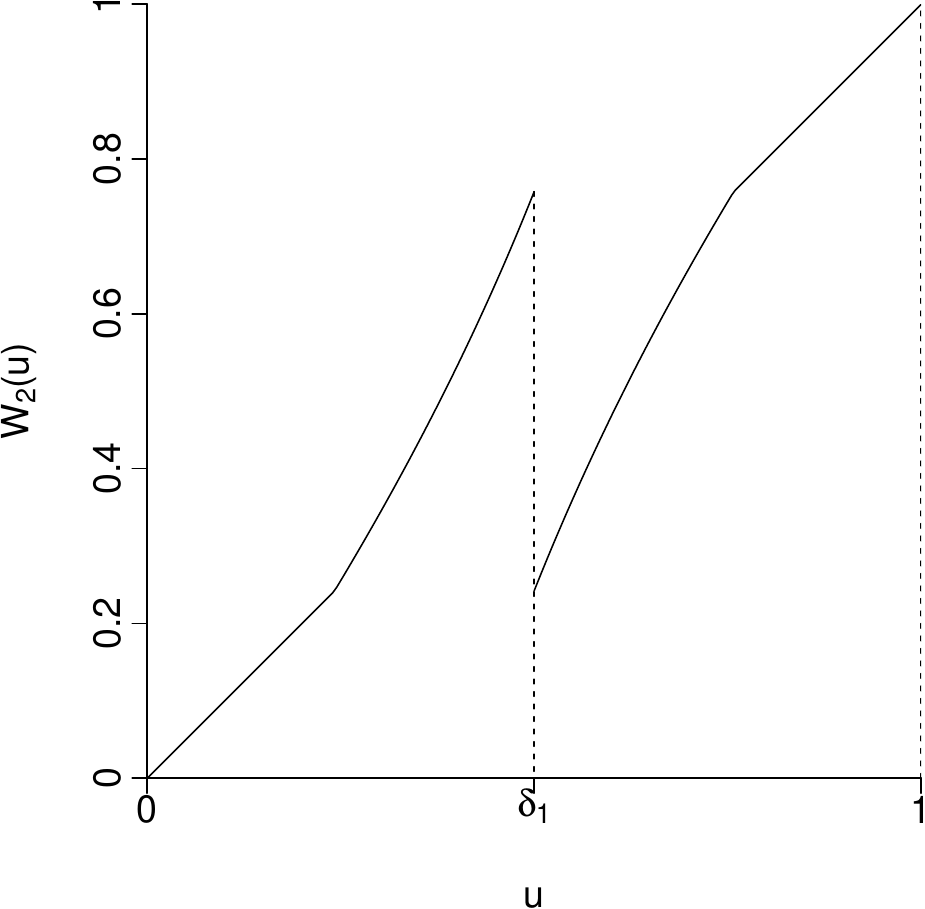}\\[5mm]
      \includegraphics[width=0.48\textwidth]{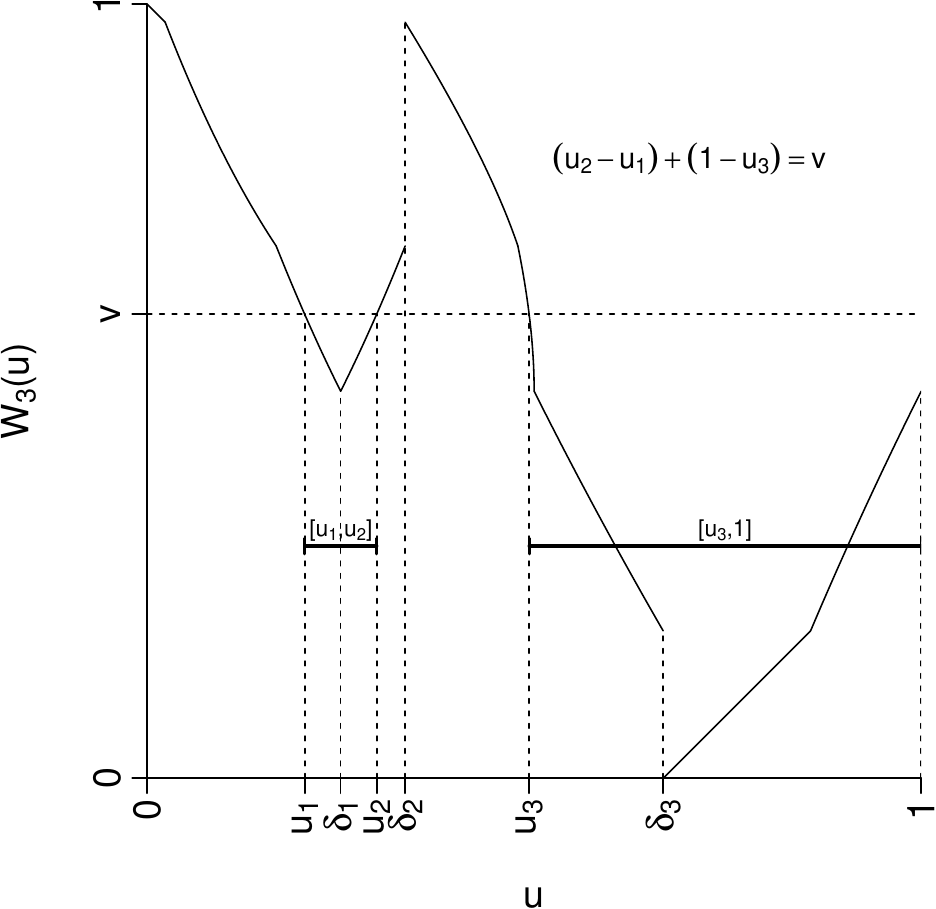}
      \hfill
      \includegraphics[width=0.48\textwidth]{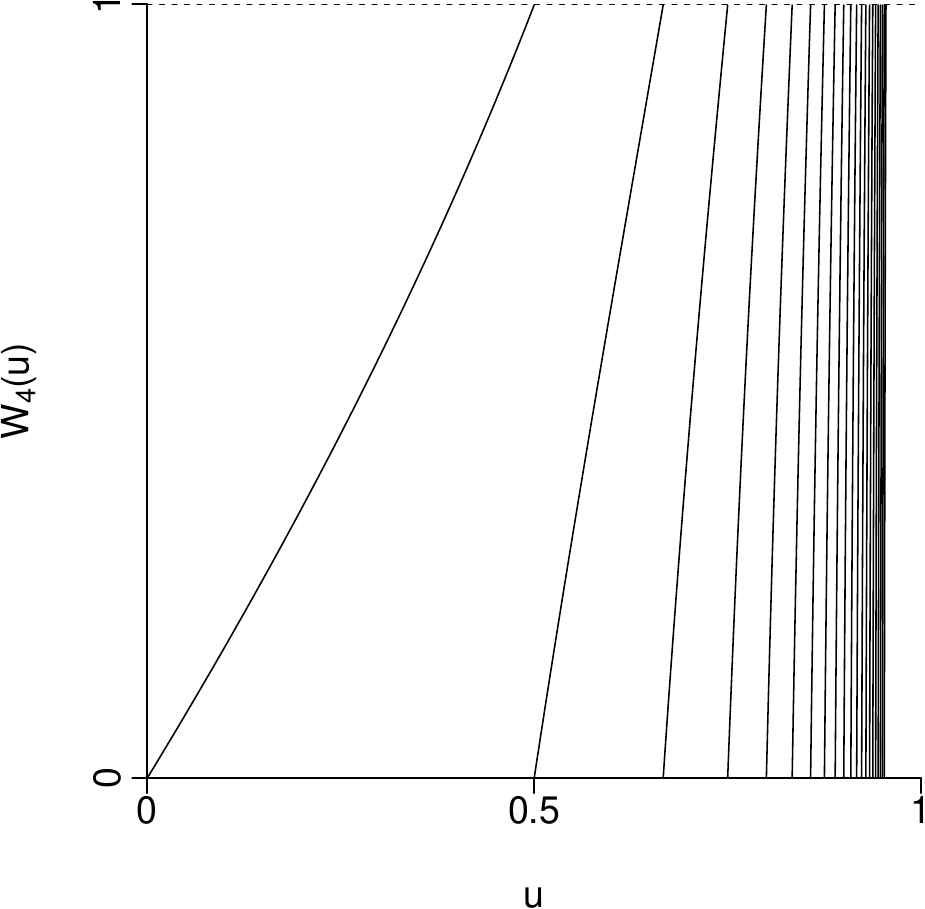}
      \caption{Shuffle of identity (top-left), piecewise increasing W-transform (top-right),
        zig-zagged W-transform with illustration of partition of square property (bottom-left) and
        W-transform with countably many change points (bottom-right).}
      \label{fig:example:W:trans}
    \end{figure}
  \item\label{eg:W:trafo:piece:in} \emph{Piecewise increasing W-transform.} Let $F_X(x) = \begin{cases}
    1 - 0.25^{x}, & x\in[0,0.5), \\
    4^{x-1}, & x\in[0.5,1],
  \end{cases}$
  and $T(x) = \begin{cases}
    x, & x\in[0,0.5], \\
    x-\alpha, & x\in(0.5,1],
  \end{cases}$
  with change points $t_0=0, t_1=0.5, t_2=1$, where $\alpha\in\IR$. If $a<0$, then the functional
  form of \eqref{eq:W:transform} is
  $\W(u) = u,\ u\in[0,1]$, and if $\alpha>0.5$, then $\W(u) = \begin{cases}
    u+\frac{1}{2-2u}-0.5, & u\in[0, 0.5],\\
    u-\frac{1}{2u}+0.5, & u\in(0.5, 1].\\
  \end{cases}$
  Otherwise, if $\alpha\in[0,0.5]$, then
  \begin{align*}
    \W(u) = \begin{cases}
      u, & u\in[0, 1-0.25^{0.5-\alpha}], \\[10pt]
      u+\dfrac{4^{\alpha-1}}{1-u}-0.5, & u\in(1-0.25^{0.5-\alpha} ,0.5], \\[10pt]
      u-\dfrac{4^{\alpha-1}}{u}+0.5, & u\in(0.5, 0.25^{0.5-\alpha}], \\[10pt]
      u, & u\in(0.25^{0.5-\alpha},1].
    \end{cases}
  \end{align*}
  The change points of $\W$ are then $\delta_0 = 0, \delta_1 = F_X(0.5)=0.5$
  and $\delta_2 = 1$. A plot of $\W_2\coloneqq\W$ for $\alpha=0.3$ is shown in
  Figure~\ref{fig:example:W:trans} (top-right).
\item\label{eg:W:trafo:zig:zag} \emph{W-transform with more change points.} Let $X \sim \U(0, 1)$.
  Let $T$ have change points $t_0=0, t_1 = \frac{1}{4}, t_2=\frac{1}{3}, t_3=\frac{2}{3}$, $t_4=1$,
  and define
  \begin{align*}
    T(x) = \begin{cases}
      \exp(3(x-\frac{1}{4})^2), & x\in[0,\frac{1}{3}],\\
      -x+\frac{3}{2}, & x\in(\frac{1}{3},\frac{2}{3}], \\
      \frac{1}{x}, & x\in(\frac{2}{3},1].\\
    \end{cases}
  \end{align*}
  Then the corresponding W-transform $\W$ has four pieces and exhibits
  a ``zig-zag'' pattern; see Figure~\ref{fig:example:W:trans} (bottom-left). Its
  change points are $\delta_0=0$, $\delta_1 = 1/4$, $\delta_2=1/3$,
  $\delta_3=2/3$, $\delta_4=1$. According to
  Proposition~\ref{prop:W:property}~\ref{prop:W:property:partition:sq}, with
  $v=0.6$, we have $\W^{-1}_{|1}(0.6)\approx 0.20328$,
  $\W^{-1}_{|2}(0.6) \approx 0.29672$, and
  $\W^{-1}_{|3}(0.6)\approx 0.49343$. One can thus identify each one of
  the $S_k(v)$'s and indeed confirm that
  $\P(\biguplus_{k=1}^{K} S_k(v)) = v$ (shown in the bottom-left panel of
  Figure~\ref{fig:example:W:trans}). The explicit functional form of
  $\W$ is omitted here for brevity, but can be given explicitly via
  \eqref{eq:W:transform}.
\item \emph{Countably infinitely many change points.}
  Let $X$ follow a Pareto Type~I distribution with
  distribution function $F_X(x)=1-1/x^2$, $x\in[1,\infty)$, and $T(x)=x^2-\lceil x^2 \rceil+1$, $x\in [1,
  \infty)$. Then the functional form of \eqref{eq:W:transform} is
  \begin{align*}
    \W(u)=\sum_{n\in\bar{\IN}}F_X\Bigl(\sqrt{n+(F_X^{-1}(u))^2-\lceil (F_X^{-1}(u))^2
    \rceil+1}\Bigr)-F_X(\sqrt{n}),
  \end{align*}
  A plot of $\W_4\coloneqq \W$ is shown in Figure~\ref{fig:example:W:trans} (bottom-right).
\end{enumerate}
\end{example}

As we have seen, our W-transforms generalise the v-shape of v-transforms to allow for
more general piecewise monotone functions. The top-right plot of
Figure~\ref{fig:example:W:trans} motivates the question when W-transforms are
piecewise linear. We now provide three sufficient conditions under which this holds.
The first one exploits the injectivity of $T$; the second one takes $T=F_X$; and the
last one considers symmetry across admissibly dissected pieces of $T$, where
``dissected'' refers to the partitioning of a monotone piece of $T$ into subpieces.

\begin{proposition}[Sufficient conditions for $\W$ to be piecewise linear]\label{prop:suff:con:lin}
  \begin{enumerate}
  \item Let $T:D\rightarrow \IR$ be injective except possibly at the change
    points $t_0, \dots, t_K$. Then, for any continuous $F_X$ with
    $\supp(F_X)= D$, the W-transform $\W$ in~\eqref{eq:W:transform} is
    piecewise linear.
  \item If $T=F_X$, then $\W(u)=u$, $u\in[0, 1]$.
  \item Let $X\sim\U(t_0, t_K)$ with $K<\infty$ and $T:D\rightarrow \IR$. Suppose there exist
    $\{t_0', \dots, t_{K'}'\}\supseteq \{t_0, \dots, t_K\}$ with $t_0'=t_0$ and $t_{K'}'=t_K$ such that $T$
    is pcsm with $K'\ge K$ pieces. Then, the restriction $T_{|k'}\coloneqq T|_{(t'_{k-1}, t_k']}$ is
    continuous and strictly monotone for any $k'\in\{1, 2\dots, K'\}$. For any fixed $\ell'\in\{1,
    \dots, K'\}$, if for all $k'\in\{1, \dots, K'\}$ one of the following properties holds, then
    $\W$ is piecewise linear.
    \begin{enumerate}[label=\roman*), labelwidth=\widthof{iii)}]
    \item $\ran(T_{|k'})\cap\ran(T_{|\ell'})\subseteq\{T_{|k'}(t_{\ell'-1}')\}$ (disjoint range,
      where $\ran(f)$ denotes the range of the function $f$);
    \item $T_{|k'}(x)=T_{|\ell'}(x+t'_{k'}-t_{\ell'}')$, $x\in (t_{k'-1}', t'_{k'}]$ (translation invariance); or
    \item $T_{|k'}(x)=T_{|\ell'}(-x+t_{k'-1}'+t'_{\ell'})$, $x\in (t_{k'-1}', t'_{k'}]$ (reflection invariance).
    \end{enumerate}
  \end{enumerate}
\end{proposition}
Interpreted geometrically, translation invariance implies that the graph of $T$ is identical
on the intervals $(t'_{k'-1}, t'_{k'}]$ and $(t'_{\ell'-1},t'_{\ell'}]$, and reflection
invariance means that the graph on $(t'_{k'-1}, t'_{k'}]$ is the mirror image of the
graph on $(t'_{\ell'-1},t'_{\ell'}]$.

\cite[Proposition 6]{porubskysalatstrauch1988} showed the following result, which we will frequently
refer to.
\begin{lemma}[Characterisation of uniformity-preservation under
  differentiability]\label{lem:unif:preserv}
  Let $\W: [0,1]\rightarrow [0,1]$ be piecewise differentiable. Then $\W$ is
  uniformity-preserving if and only if $\sum_{u\in\W^{-1}(v)} \frac{1}{|\W'(u)|}=1$
  for almost every $v\in [0,1]$.
\end{lemma}
Lemma~\ref{lem:unif:preserv} implies that $|\W'(u)|\ge 1$ almost
everywhere, meaning $\W$ is only allowed to \emph{stretch}
neighbourhoods (that is for any $J\subseteq [0,1]$, the Lebesgue
measure $\lambda$ satisfies $\lambda(J)\le\lambda(\W(J))$). Intuitively,
$\W$ cannot ``compress'' intervals while preserving uniformity, and any
expansion must be counterbalanced by the preimage condition
$\sum_{u\in\W^{-1}(v)} \frac{1}{|\W'(u)|}=1$. In
Section~\ref{subsec:tail}, we discuss how this constraint influences tail
dependence properties.

To conclude this section, we define the periodicity of a W-transform. The only known such W-transforms
are the \emph{interval-exchange transformations} (IET) as defined by \cite{Keane1975}, which are piecewise
linear and uniformity-preserving. Periodic W-transforms will help us identify shuffle-of-min copulas
of \cite{durantesarkocisempi2009} as a special case of ``W-transformed copulas'' in Section~\ref{subsec:distr} later.
\begin{definition}[Periodic W-transforms]
  Let $\W:[0,1]\rightarrow[0,1]$ be a W-transform and
  $\W^p\coloneqq\W\circ\cdots\circ\W$ be the $p$-fold composition of $\W$.
  Then $\W$ is \emph{$p$-periodic} if there exists a $p\in \IN$ such that, for almost every $u\in[0,1]$,
  one has
  \begin{align*}
    \W^{p}(u)=u
  \end{align*}
  and $p$ is the smallest such natural number, that is for $q\in\{1,\dots,p-1\}$, $\W^q(u)\neq u$ on a
  set of positive Lebesgue measure.
\end{definition}

We now provide a necessary condition for W-transforms to be $p$-periodic.
\begin{proposition}[Only bijective piecewise linear W-transforms can be $p$-periodic]\label{prop:linear:period}
  Let $\W:[0, 1]\rightarrow [0,1]$ be a W-transform. If $\W$
  is $p$-periodic, then for a Lebesgue null set $N$, $\W$ is bijective
  on $[0, 1]\setminus N$ and piecewise linear on $[0,1]$.
\end{proposition}

The W-transform given in Example~\ref{eg:W:trafo}~\ref{eg:W:trafo:shuffle} is
4-periodic. On the other hand, the function considered by \cite{Nogueira1989} in
the discussion of IETs with
\begin{align*}
  \W(u)=\begin{cases}
    \frac{2}{3}-\alpha+u, & u\in[0,\alpha],\\
    \frac{1}{3}+u-\alpha, & u\in(\alpha,\frac{1}{3}],\\
    \frac{4}{3}-u, & u\in(\frac{1}{3}, \frac{2}{3}],\\
    1-u, & u\in(\frac{2}{3}, 1],
  \end{cases}
\end{align*}
where $\alpha\in(0,1/3)$ is irrational, is not $p$-periodic. This can be quickly seen by observing
that $\W^3$ maps $[0,1/3]$ to itself by a translation: $\W^3(u)=
u-\alpha\ \mod(1/3),\ u\in[0,1/3]$, which is equivalent to rotating a circle by an irrational multiple
of its circumference and therefore, $\W$ is not $p$-periodic.
Hence, the converse of Proposition~\ref{prop:linear:period} is not true in general.

Bijective piecewise linear W-transforms as in
Proposition~\ref{prop:linear:period} have been identified as IETs in
\cite{Keane1975} and \cite{Nogueira1989}. In their definition, all pieces of
IETs are defined on non-degenerate open subintervals of $[0, 1]$, but we
slightly extended the domain of W-transforms to the endpoints of these
subintervals. For the sake of uniformity-preservation, this extension is
irrelevant since these endpoints are part of the null set $N$ in
Proposition~\ref{prop:linear:period}.

\subsection{A parametric family}
We now propose a flexible parametric family of W-transforms which we call \emph{piecewise surjective and strictly
  monotone (pssm)} W-transforms that allow us to control three features:
\begin{enumerate}
\item Change points: The number $K$ and the locations $\{\delta_k\}_{k=1}^{K}\subseteq[0, 1]$
  can be freely specified.
\item Monotonicity: The monotonicity of each piece is determined by parameters
  $\{r_k\}_{k=1}^K\subseteq\{0,1\}^K$, where $r_k=0$ ($r_k=1$) means that
  $T_{|k}\coloneqq T|_{(t_{k-1},t_k]}$ is decreasing (increasing).
\item Shape: The non-linearity of the resulting W-transform is controlled by the
  base distribution $F_X$.
\end{enumerate}
The family of pssm W-transforms generalises the class of v-transforms (recovered
for $K=2$, $r_0=0$, and $r_1=1$, see Example~\ref{eg:param:fam}~\ref{eg:param:fam:v:trans}
below) to more flexible piecewise functions. To provide its
form, let $T:[0, 1]\rightarrow[0,1]$ be pcsm with change points
$0=t_0<t_1<t_2<\cdots<t_K=1$ where $K\in\bar{\IN}$. For $k\in\{1, \dots, K\}$,
the $k$th piece $T_{|k}$ is given by
\begin{align*}
  T_{|k}(t)=(-1)^{1-r_k}\frac{t-c_k}{t_k-t_{k-1}}, \quad \text{where } c_k=r_kt_{k-1}
  +(1-r_k)t_k,
\end{align*}
and $r_k\in\{0, 1\}$ indicates whether $T_{|k}$ is increasing. The resulting
pssm W-transform $\W_{\bm{t}, \bm{r}, F_X}$ has change points at
$\{F_X(t_k)\}_{k=0}^K$ and is given by
\begin{align}\label{eq:W:param:fam}
  \W_{\bm{t}, \bm{r}, F_X}(u)&=\sum_{k=1}^K\bigl[F_X\bigl(T(F_X^{-1}(u))t_k+
                               \bigl(1-T(F_X^{-1} (u))\bigr)t_{k-1}\bigr)-F_X(t_{k-1})\bigr]^{r_k}\times \notag\\
                             &\phantom{=\sum_{k=1}^K\bigl[}\bigl[F_X(t_k)-F_X\bigl(T(F_X^{-1}(u))t_{k-1}+\bigl(1-
                               T(F^{-1}_X(u))t_k\bigr)\bigr)\bigr]^{1-r_k}.
\end{align}

The following example shows that v-transforms and piecewise linear surjective
W-transforms are pssm W-transforms.
\begin{example}[Flexibility of pssm W-transforms]\label{eg:param:fam}
  \begin{enumerate}
  \item\label{eg:param:fam:v:trans} \emph{v-transforms.} Consider an absolutely continuous $F_X$
    on $[0, 1]$. Let $\bm{t}=(0, F_X^{-1}(\delta), 1)$ where $\delta\in(0, 1)$ is the fulcrum
    and $\bm{r}=(0, 1)$. Then \eqref{eq:W:param:fam} can be written as
    \begin{align}\label{eq:param:v:trans}
      \W_{(0, F_X^{-1}(\delta), 1), (0,1), F_X}(u)=\begin{cases}
        F_X\bigl(1-\frac{1-F_X^{-1}(\delta)}{F_X^{-1}(\delta)}F_X^{-1}(u)\bigr)-u, & u\leq \delta,\\
        u-F_X\bigl(\frac{1-F_X^{-1}(u)}{1-F_X^{-1}(\delta)}F_X^{-1}(\delta)\bigr), & u>\delta.
      \end{cases}
    \end{align}
    Example~\ref{eg:v:trans:use}~\ref{eg:v:trans:use:chara:v} gave necessary and sufficient
    conditions for a function $\mathcal{V}:[0,1]\rightarrow[0,1]$ to be a
    v-transform.  If one takes
    \begin{align*}
      G(x)=\frac{1-F_X(1-\frac{1-F_X^{-1}(\delta)}{F_X^{-1}(\delta)}F_X^{-1}(\delta x))}{1-
      \delta}, \quad x\in[0,1]
    \end{align*}
    in \eqref{eq:vtrans}, one obtains that \eqref{eq:param:v:trans} is a
    v-transform. The left-hand side of Figure~\ref{fig:example:param:fam} shows an example
    for which $F_X(x)=x^2$, $x\in[0, 1]$, $\bm{t}=(0, 0.5, 1)$ so that
    $G(x)=\frac{4\sqrt{x}-x}{3}$, $x\in[0, 1]$, and $\delta= 0.25$.
    \begin{figure}[htbp]
      \centering
      \includegraphics[width=0.48\textwidth]{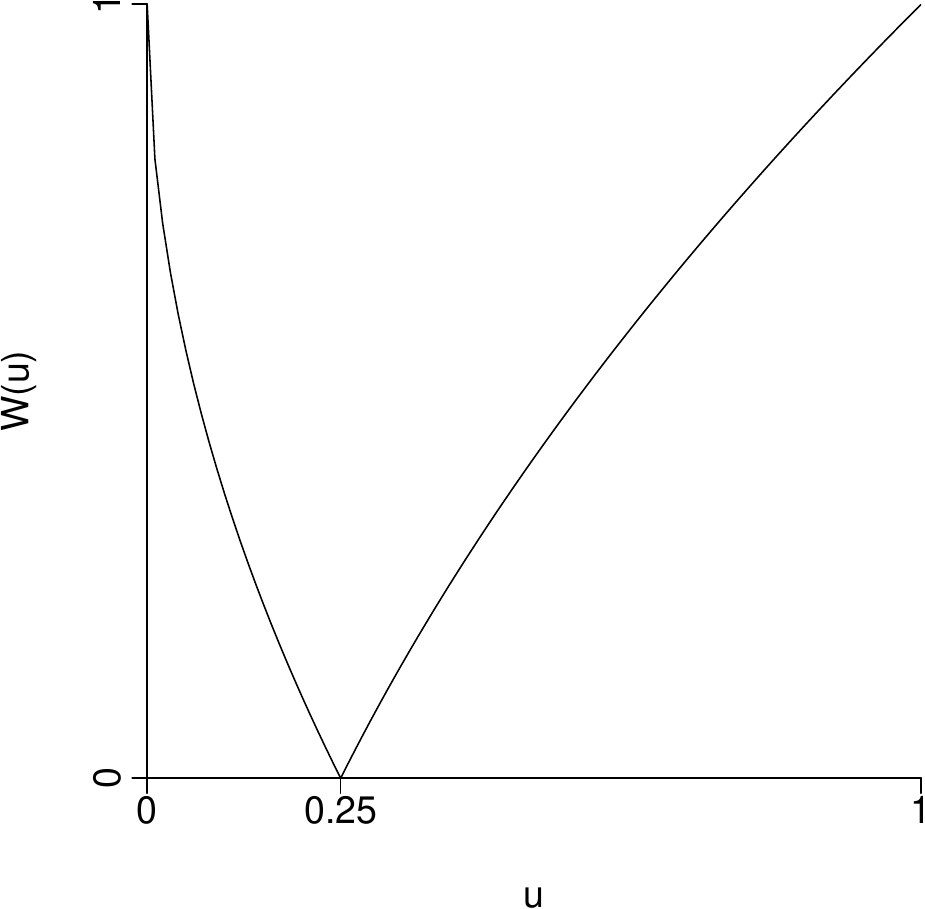}\hfill
      \includegraphics[width=0.48\textwidth]{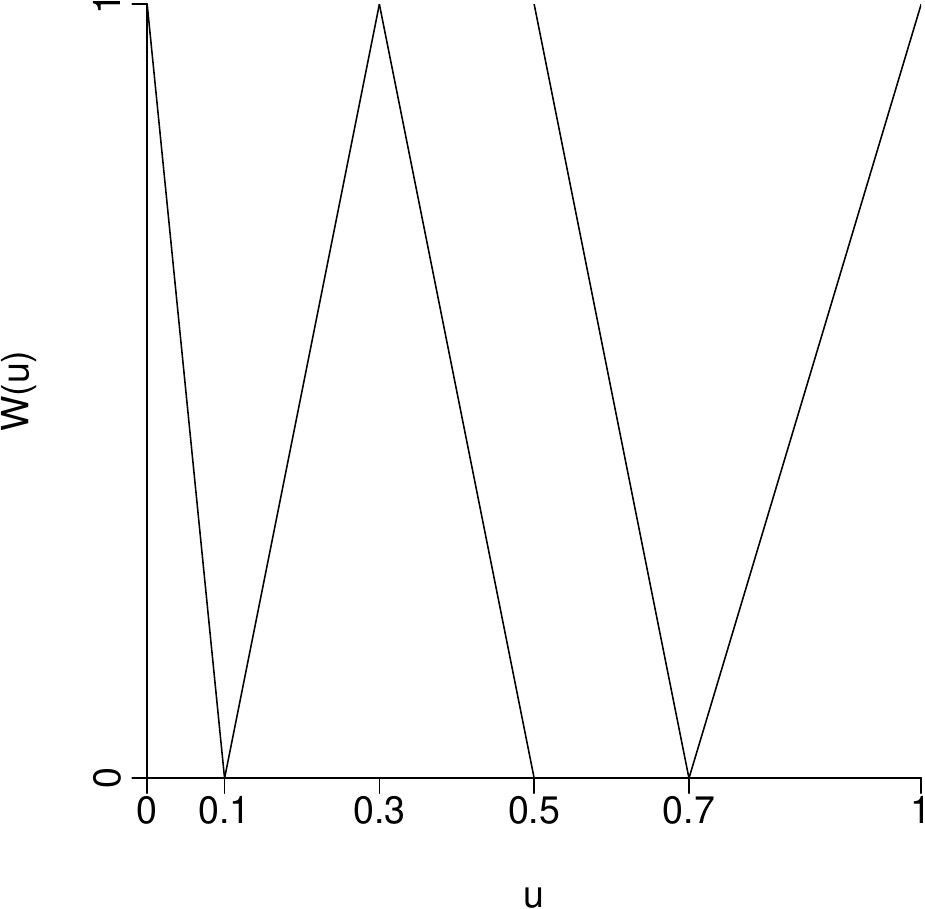}
      \caption{A v-transform recovered from \eqref{eq:W:param:fam} for
        $G(x)=\frac{4\sqrt{x}-x}{3}$ (left), and
        a pssm W-transform $\W_{\bm{t}, \bm{r}, F_X}$ constructed from
        \eqref{eq:W:param:fam} for $\bm{t}=(0, 0.1, 0.3, 0.5, 0.7, 1)$, $\bm{r}=(0, 1, 0, 0, 1)$
        and $X\sim\U(0, 1)$ (right).}
      \label{fig:example:param:fam}
    \end{figure}
  \item \emph{Piecewise surjective and linear W-transforms.}
    Let $X\sim\U(0,1)$. Then \eqref{eq:W:param:fam} reduces to
    \begin{align*}
      \W_{\bm{t}, \bm{r}, F_X}(u)=\sum_{k=1}^K[(t_k-t_{k-1})T(u)]^{r_k}
      [2t_k-(t_k+t_{k-1})T(u)]^{1-r_k}.
    \end{align*}
    Since $T$ is piecewise linear, so is the W-transform $\W_{\bm{t}, \bm{r}, F_X}$.
    The right-hand side of Figure~\ref{fig:example:param:fam} shows an example for which
    $\bm{t}=(0, 0.1, 0.3, 0.5, 0.7, 1)$ and $\bm{r}=(0, 1, 0, 0, 1)$.
  \end{enumerate}
\end{example}

To conclude this section, we provide a lemma for the derivatives of
$\W_{\bm{t}, \bm{r}, F_X}$ at both endpoints of the support. In
Section~\ref{subsec:cop:ordinal:sum} later, this will be useful for
modifying the tails of a ``W-transformed copula''.

\begin{lemma}[Derivatives at the boundary]\label{lem:param:derivative}
  Let $\W_{\bm{t}, \bm{r}, F_X}$ be a pssm W-transform as
  in~\eqref{eq:W:param:fam} with absolutely continuous $F_X$ and density
  $f_X$. Let $\bm{r}=\bm{1}$, that is $\W_{\bm{t}, \bm{1}, F_X}$ is
  piecewise increasing.
  \begin{enumerate}
  \item If $f_X(0+)=\infty$ and $f_X(x) < \infty$, $x\in(0, 1)$, then $\W'_{
      \bm{t}, \bm{1}, F_X}(0+)=1$.
  \item If $f_X(1-)=\infty$ and $f_X(x) < \infty$, $x\in(0, 1)$, then $\W'_{
      \bm{t}, \bm{1}, F_X}(1-)=1$.
  \end{enumerate}
\end{lemma}

\section{Generalised W-transforms}\label{sec:general:W:trafo}
In the previous section, we considered continuous $F_X$. If $F_X$ is not
continuous, the probability transform for $F_X$ %
fails to be $\U(0, 1)$ distributed.  To
generalise W-transforms to arbitrary distributions $F_X$, we utilise the notion
of a \emph{generalised probability transform} of \cite{rueschendorf2009} in this
section. To this end, let $X\sim F_X$ and $V\sim\U(0, 1)$ be independent. In terms of the
\emph{modified distribution function}
$F_X(x, v)\coloneqq\P(X<x)+v \P(X=x)$, $v\in[0,1]$, $x\in\IR$, the
\emph{generalised probability transform} is
$F_X(X,V) = F_X(X-)+V(F_X(X)-F_X(X-))$. By construction, $F_X(X,V)\sim\U(0, 1)$
and $F_X^{-1}(U)=X$ a.s.; see \cite[Proposition~2.1]{rueschendorf2009}.

With these notions at hand, we can now generalise W-transforms to arbitrary
random variables $X\sim F_X$.
\begin{definition}[Generalised W-transform]\label{def:general:W:trafo}
  Let $T: D\to\IR$ be pcsm with change points $\{t_k\}_{k=0}^K$ and $X\sim F_X$
  with $\inf\supp(F_X)=t_0$ and $\sup\supp(F_X)=t_K$.  Let $V\sim\U(0,1)$ be
  independent of $X$. Then the \emph{generalised W-transform}
  $\Wg:[0,1]\to[0,1]$ is
  \begin{align}\label{eq:general:W:trafo}
    \Wg(F_X(x, V)) =\begin{cases}
      \lim_{u\to 0+} \W_{\text{g}}(u), & F_X(x, V)=0,\\
      F_{T(X)}(T(x), V), & F_X(x, V) \in (0, 1]
                           .    \end{cases}
  \end{align}
\end{definition}

A (generalised) W-transform operates on the (generalised) probability transform of $X$, mapping
it to that of $T(X)$. This implies that $\W$ (respectively $\Wg$) must be
uniformity-preserving.

We now present two examples, the first one is a continuation of
Example~\ref{eg:non:unif:preserv} and the second one features a $\Wg$
constructed from a mixed-type distribution.

\begin{example}[Generalised W-transforms $\Wg$]\label{ex:gen:W}
  \begin{enumerate}
  \item \emph{Continuation of Example~\ref{eg:non:unif:preserv}.} Consider
    $X\sim\B(1,p)$, $p\in[0,1]$. If $p=0$, then $X=0$ a.s., and
    $F_X(x, v) = \I_{[0,\infty)}(x)+v\I_{\{x=0\}}, \ x\in\IR$. Furthermore,
    $T(X)=T(0)$ a.s.\ for any $T$ and so
    $F_{T(X)}(T(x), v) = \I_{\{T(x)\geq T(0)\}}+v\I_{\{T(x)=T(0)\}},\
    x\in\IR$. Since $\Wg$ maps $F_X(x, v)$ to $F_{T(X)}(T(x), v)$ by
    \eqref{eq:general:W:trafo}, we have $\Wg(u) =u$, $u\in(0,1)$,
    $\Wg(0)\in\{0, 1\}$ and $\Wg(1)\in\{0, 1\}$. Similarly,
    if $p=1$, one has $\Wg(u) = u$ for any $u\in(0,1)$ and
    $\Wg(0),\Wg(1)\in\{0, 1\}$.  Hence, $\Wg$ is
    the identity on $(0, 1)$ and is thus uniformity-preserving.

    If $p\in(0,1)$, then $F_X(x, v)=(1-p)(\I_{(0,\infty)}(x)+v\I_{
      \{x=0\}})+p(\I_{(1,\infty)}(x)+v\I_{\{x=1\}})$, $v\in[0, 1]$,
    $x\in\IR$, and $F_{T(X)}(T(x), v)=(1-p)(\I_{(T(0),\infty)}(T(x))+v\I_
    {\{T(x)=T(0)\}})+p(\I_{(T(1),\infty)}(T(x))+v\I_{\{T(x)=T(1)\}})$, $v\in[0, 1]$, $x\in\IR$.
    \begin{enumerate}[label=\roman*), labelwidth=\widthof{iii)}]
    \item If $T(1)>T(0)$, then $\Wg(u)=u$, $u\in[0,1]$.

    \item If $T(1)<T(0)$, then $\Wg(u)=\begin{cases}
      u+p, & u\in[0,1-p], \\
      u-1+p, & u\in(1-p, 1].
    \end{cases}$

  \item If $T(1)=T(0)$, then $T(X)=T(0)$ a.s. and
    $F_{T(X)}(T(x), v)=v\I_{\{T(x)=T(0)\}}+\I_{(T(0),\infty)}(T(x))$, $x\in\IR$.
    It follows that
    $\Wg(u)=\begin{cases}
      u/(1-p), & u\in[0,1-p], \\
      (u-(1-p))/p, & u\in(1-p, 1].
    \end{cases}$
  \end{enumerate}
  In all cases, $\Wg$ is uniformity-preserving.
\item\label{ex:gen:W:mixed} \emph{Mixed-type distribution.} Consider $X\sim F_X$ with
  \begin{align*}
    F_X(x) = \begin{cases}
      1-e^{-0.5(x+1)}, & x\in[-1,0),\\
      e^{-0.5}, & x=0,\\
      1+e^{-0.5}-e^{-0.5x}, &x\in(0, 1],
    \end{cases}
  \end{align*}
  so that $\P(X=0)=2e^{-0.5}-1$. For $\alpha\in [0,1]$,
  consider $T: [-1, 1]\rightarrow \IR$ with $T(x; \alpha)=\begin{cases}
    \alpha, & x=0,\\
    |x| & x\neq 0.
  \end{cases}$
  By~\eqref{eq:general:W:trafo}, we have
  \begin{align*}
    \Wg(u) = \begin{cases}
      1+e^{-0.5}-u-\dfrac{e^{-0.5}}{1-u}, & u\in[0, 1-e^{0.5(\alpha-1)}],\\
      2-e^{-0.5}-u-\dfrac{e^{-0.5}}{1-u}, & u\in(1-e^{0.5(\alpha-1)},1-e^{-0.5}),\\
      u-e^{-0.5\alpha}+e^{0.5(\alpha-1)}, & u\in[1-e^{-0.5}, e^{-0.5}],\\
      u-2e^{-0.5}+\dfrac{e^{-0.5}}{1+e^{-0.5}-u}, & u\in(e^{-0.5},1+e^{-0.5}-e^{-0.5\alpha}],\\
      u+\dfrac{e^{-0.5}}{1+e^{-0.5}-u}-1, & u\in (1+e^{-0.5}-e^{-0.5\alpha},1].
    \end{cases}
  \end{align*}
  Plots of $\Wg$ for $\alpha\in\{0, 0.5, 1\}$ are shown in
  Figure~\ref{fig:general:W:trafo}.
  \begin{figure}[htbp]
    \includegraphics[width=0.32\textwidth]{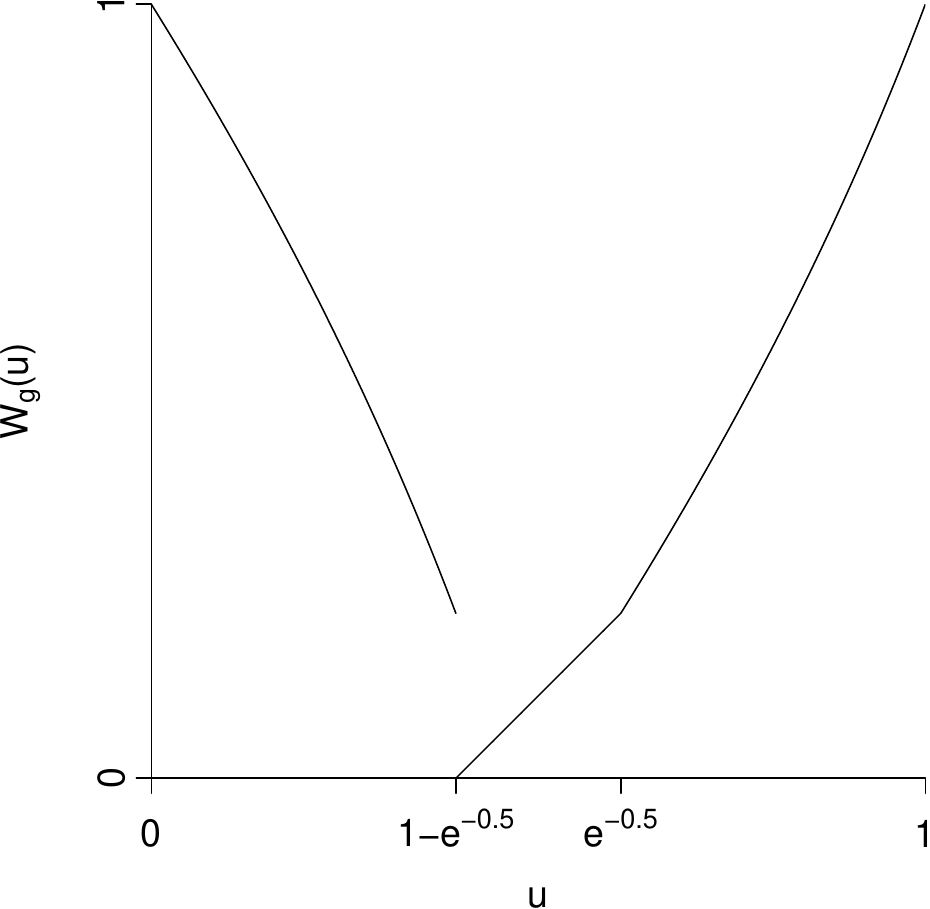}\hfill
    \includegraphics[width=0.32\textwidth]{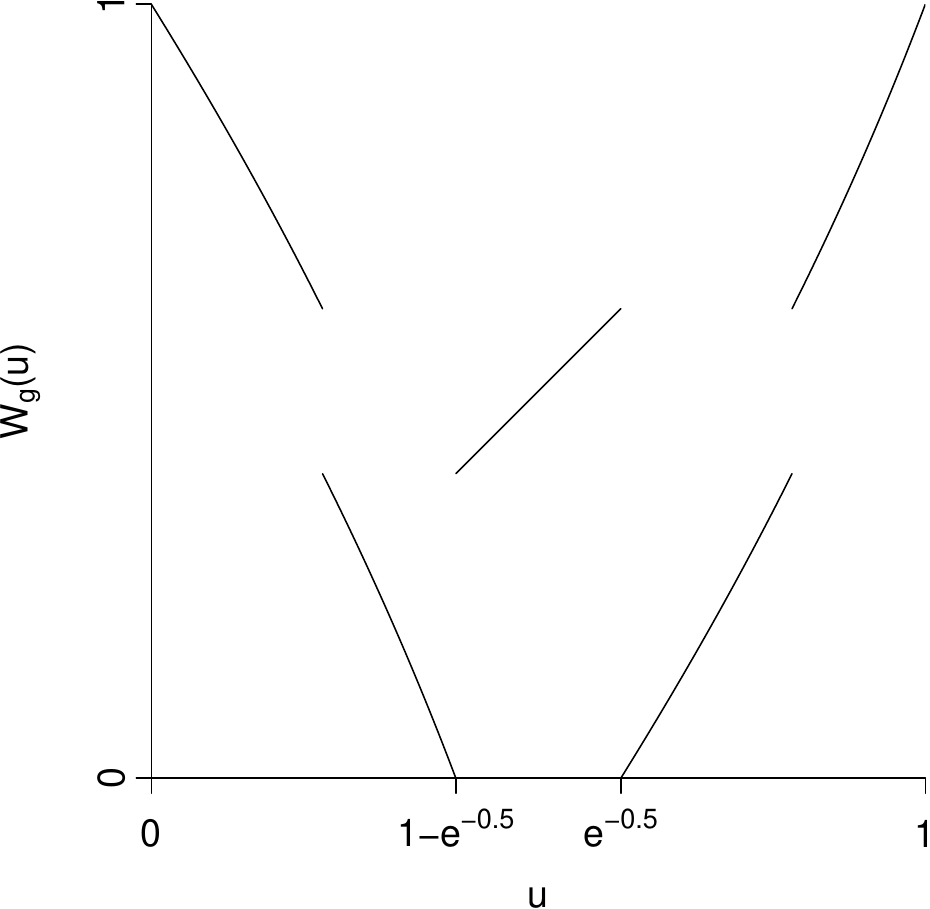}\hfill
    \includegraphics[width=0.32\textwidth]{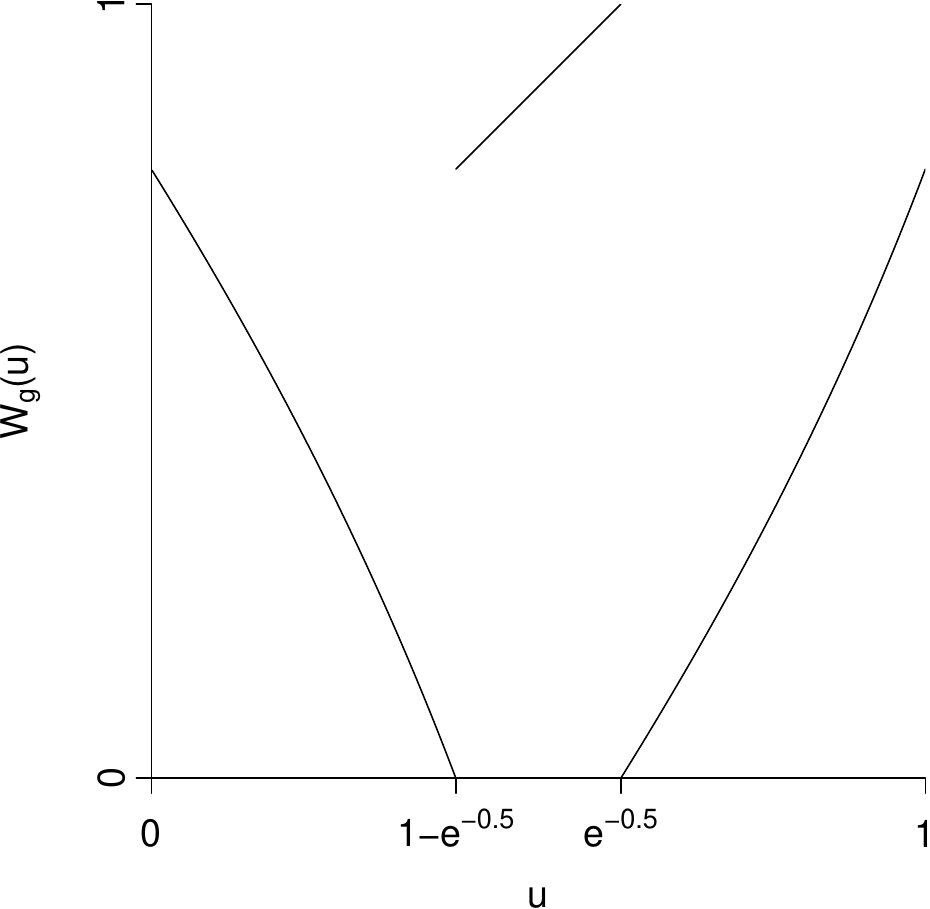}
    \caption{Generalised W-transform $\Wg$ induced by $T(.;\alpha)$ and a mixed-type
      $F_X$ for $\alpha=0$ (left), $\alpha=0.5$ (centre) and $\alpha=1$ (right).}
    \label{fig:general:W:trafo}
  \end{figure}
\end{enumerate}
\end{example}

We observe from Example~\ref{ex:gen:W}~\ref{ex:gen:W:mixed} that $\Wg$
jumps at $1-e^{-0.5}$ and is linear on $[1-e^{-0.5}, e^{-0.5}]$. Moreover, the
length of the linear part is exactly $\P(X=0)$, that is, the probability of $X$
taking on its discrete value.  In general, if $X$ jumps, then $\Wg$
induced by $X$ and $T$ also jumps. The following result proves this more
formally.

\begin{proposition}[Discontinuous generalised W-transforms $\Wg$ have linear pieces]\label{prop:jump:linear}
  Let $X\sim F_X$ and suppose that $F_X$ jumps at $x_0$
  (so $\P(X=x_0)=F_X(x_0)-F_X(x_0-)>0$) for some $x_0\in\IR\setminus N$ for a
  Lebesgue null set $N$.  Let $T$ and $\Wg$ be as in Definition
  \ref{def:general:W:trafo}.  Then $\Wg$ is linear on $(F_X(x_0-),
  F_X(x_0))$. Furthermore, if $T$ maps multiple jump points $x_0,\dots,x_L$ of
  $F_X$ to the same value $s\coloneqq T(x_0)$, then the slope of each linear
  piece of $\Wg$ on $(F_X(x_\ell-), F_X(x_\ell))$, $\ell\in\{0,\dots,L\}$, is
  $(\sum_{\ell=0}^L\P(X=x_\ell))/\P(X=x_\ell)$.
\end{proposition}

Although, as we saw in this section, one can generalise the construction of
uniformity-preserving transformations to arbitrary distributions $F_X$, by
Proposition~\ref{prop:jump:linear} the resulting generalised W-transforms
$\Wg$ are always linear on $(F_X(x_0-), F_X(x_0))$. Moreover, by
Proposition~\ref{prop:suff:con:lin}, there are various ways for constructing linear
parts in $\W$ even if $F_X$ is continuous. Therefore, in what follows,
we focus on W-transforms constructed from continuous $F_X$ as we did in Section~\ref{sec:cont:W:trafo}.

\section{W-transformed copulas}\label{sec:cop}
Since W-transforms are uniformity-preserving, they serve naturally as
copula-to-copula transformations, and thus allow us to construct more flexible
dependence structures from given ones.  In this section, we thus apply
W-transforms marginally to investigate the resulting copulas, that is given
$\bm{U}\sim C$ for a \emph{base copula} $C$ and marginal W-transforms
$\W_1, \dots, \W_d$ constructed from continuous $F_{X_1},\dots,F_{X_d}$,
we study the \emph{W-transformed copula} $C_{\bm{\W}}$ of $(\W_1(U_1), \dots, \W_d(U_d))$;
note that by Proposition~\ref{prop:w:unif:preserv}, we have
\begin{align}
  (\W_1(U_1), \dots, \W_d(U_d))\sim C_{\bm{\W}}.\label{eq:W:transformed:copula}
\end{align}

In Section~\ref{subsec:sto:inv}, we derive the stochastic inverse of
$\W$ and the copula of $(U,\W(U))$ for $U\sim\U(0, 1)$. In
Section~\ref{subsec:distr}, we derive the analytical form of
$C_{\bm{\W}}$ and show that it can be interpreted as a sum of
$C$-volumes. Thereafter, in Section~\ref{subsec:tail}, we derive bounds on
the tail dependence coefficients of $C_{\bm{\W}}$, which provide
meaningful guidance on how W-transforms may increase tail dependence.  Despite
the lack of closed-form formulas, we also investigate concordance measures of
W-transformed copulas; see Section~\ref{subsec:concor}. Finally, in
Section~\ref{subsec:symmetry}, we address symmetry properties of
$C_{\bm{\W}}$ in relation to $C$, determining when W-transforms break
or preserve symmetries of $C$.

\subsection{Stochastic inverse of $\W$}\label{subsec:sto:inv}
To facilitate the construction and sampling of W-transformed copulas considered
later, we need the notion of a stochastic inverse of any W-transform
$\W$.

For a W-transform $\W$, consider $D_k:=(\delta_{k-1},\delta_k]$, $k\in\{1, \dots, K\}$ and
define the restriction $\W_{|k}\coloneqq\W|_{D_k}$ of $\W$ on $D_k$.
Let $O_k\coloneqq\{\W(u):u\in D_k\}$, and for $v\in O_k$ define the inverse of $\W$
locally on $D_k$ via the restriction $\W_{|k}$ as
\begin{align*}
  \W^{-1}_{|k}(v) = \begin{cases}
    \sup \{u\in D_k: \W_k(u)\geq v\}, & \text{if } \W_{|k}
                                        \text{ is strictly decreasing},\\
    \inf \{u\in D_k: \W_k(u)\geq v\}, & \text{if } \W_{|k}
                                        \text{ is strictly increasing},
  \end{cases}
\end{align*}
with the convention that $\sup\emptyset = \delta_{k-1}$ and $\inf\emptyset=\delta_k$.
For any $v\in [0, 1]$, let $N(v)\coloneqq\{k\in \{1,\dots,K\}:\W^{-1}_{|k}(v)
\in(\delta_{k-1}, \delta_k)\}$, that is, $N(v)$ identifies the pieces of $\W$
where $\W^{-1}_{|k}(v)$ are not change points.

Section~\ref{sec:cont:W:trafo} defined W-transforms from continuous $F_X$ and
have shown that such W-transforms are uniformity-preserving, pcsm, and satisfy the partition of
square property (Proposition~\ref{prop:W:property}~\ref{prop:W:property:partition:sq}).
With these at hand, we are now ready to derive the copula of $(U,\W(U))$, which is our first main result
in this section.
\begin{theorem}[Copula of $(U, \W(U))$]\label{theorem:U:V:cop}
  Consider a W-transform $\W$ with increasing (decreasing) pieces indexed by
  $I\subseteq \{1, \dots, K\}$ ($I^C=\{1,\dots,K\}\setminus I$).
  Let $U\sim\U(0, 1)$ and $V=\W(U)$.
  \begin{enumerate}
  \item The joint distribution function of $(U, V)$ is given, for all $u,v\in[0,1]$, by the copula
    \begin{align}\label{eq:U:V:cop}
      C(u, v)&=\sum_{k\in I}\max\bigl\{\min\{u,\W^{-1}_{|k}(v)\}-\delta_{k-1}, 0\bigr\} +\sum_{k\in I^C}\max \bigl\{\min\{\delta_k,u\}-\W^{-1}_{|k}(v), 0\bigr\}.
    \end{align}
  \item\label{theorem:U:V:cop:cond} Let $v\in[0,1]$. If, for every $u$ such that
    $\W(u)=v$, $\W$ is differentiable at $u$, then,
    conditional on $V=v$, the distribution of $U=\W_{|k}^{-1}(v)$ is
    \begin{align}\label{eq:joint:U:V}
      \P(U\leq u\,|\,V=v)&=\sum_{k\in N(v)}p_k\I_{\{\W_{|k}^{-1}(v)\le u\}},
    \end{align}
    where $p_k\coloneqq\bigl|\frac{\rd}{\rd v}\W_{|k}^{-1}(v)\bigr|$ for each $k\in N(v)$.
    Notably, non-differentiability only occurs at countably many points and is hence
    stochastically negligible.
  \end{enumerate}
\end{theorem}

Theorem~\ref{theorem:U:V:cop} gives a method to stochastically invert the
non-injective $\W$ through a probability allocation. When multiple
solutions exist to the equation $\W(u)=v$ (see,
Example~\ref{eg:W:trafo}~\ref{eg:W:trafo:zig:zag} for $v=0.6$), the inverse of
$\W$ distributes values according to a multinomial
distribution. However, if $\W$ is not differentiable at
$u\in \{u\in[0,1]:\W(u) = v\}$ for some $v\in [0,1]$, then
Theorem~\ref{theorem:U:V:cop}~\ref{theorem:U:V:cop:cond} fails.  For example, in
Example~\ref{eg:W:trafo}~\ref{eg:W:trafo:piece:in} with $v=1/\sqrt[5]{4}\approx 0.7579$, the unique
solution $u=1/\sqrt[5]{4}$ coincides with a change point in which $\W$ is not
differentiable. Here, $p_2\approx 0.6025\neq1$. Since there are only
countably many change points, there are only countably many $v$'s for which
$\W$ is not differentiable at $u\in \{u\in[0,1]:\W(u) =
v\}$. Since $\{u\in[0, 1]:\W(u) = v\}$ is countable, $\W$ is
differentiable almost everywhere.

\begin{definition}[Stochastic inverse of W-transforms]
  Let $\W$ be a W-transform constructed from a continuous $F_X$ and $U'\sim \U(0,1)$. Let
  $D\coloneqq \{u: \W \mbox{ is differentiable at } u\}$. Define the
  stochastic inverse $\W^{-1}: D\times [0, 1]\rightarrow [0, 1]$ of
  $\W$ by
  \begin{align*}
    \W^{-1}(v, U') = \sum_{k\in N(v)} \W^{-1}_{|k}(v)\I\biggl\{U'\in\biggl(\,\sum_{\ell=1}^{k-1}p_\ell, \sum_{\ell=1}^{k}p_\ell\biggr]\biggr\},\quad v\in D.
  \end{align*}
\end{definition}

The following result establishes basic properties of stochastic inverses of W-transforms.
\begin{proposition}[$\W\circ\W^{-1}$ is a stochastic identity]\label{prop:sto:id}
  Let $V, U'\sim \U(0, 1)$ be independent. Then $\W(\W^{-1}(V, U')) =V$ and
  $\W^{-1}(V, U')\sim \U(0, 1)$. %
\end{proposition}
As stochastic inverses, $\W^{-1}(\W(u), U')$ may not be equal to $u$.
To see this, let $\W(u)=|2u-1|$ with stochastic inverse
$\W^{-1} (v, U')=\frac{1-v}{2}+v\I\{T>\frac{1}{2}\}$. Then
$\W^{-1}(\W(\frac{1}{4}), \frac{3}{4})=\frac{3}{4}\neq \frac{1}{4}$. This is because of
the stochastic choice among the preimages
$\{\W^{-1}_{|k}(v):k\in N(v)\}$, which reflects the general non-invertibility
of W-transforms.

To end this section, we consider shuffles of copulas and can relate them
to W-transforms.%
\begin{example}[Shuffle of copulas]
  Consider a random vector $(U, V)\sim C$. Then $C$ is a \emph{shuffle-of-min} copula as
  detailed by \cite{durantesarkocisempi2009} if and only if a bijective, piecewise continuous function
  $f$ exists such that $V=f(U)$ almost surely. For example, if one takes $(U,V)\sim M$, a $p$-periodic
  W-transform $\W$ and constructs $V\coloneqq\W(U)$, then the joint distribution function of
  $(U,\W(U))$ is a shuffle-of-min copula. \cite{durantesarkocisempi2009} further generalised this construction
  to shuffle-of-$C$ copulas. Starting from $(U,V)\sim C$ and a bijective measure-preserving
  $\mathcal{T}: [0, 1]\to[0, 1]$, they defined a new copula $C_{\mathcal{T}}$ as the joint distribution
  function of $(\mathcal{T}(U), V)$.  In the context of W-transforms, this can be achieved by replacing
  $\mathcal{T}$ by a $p$-periodic W-transform $\W$.  The analytical form of $C_{\W}=C_\mathcal{T}$
  is given in \eqref{eq:U:V:cop}.
\end{example}

\subsection{Componentwise W-transforms as multivariate measure-preserving transformations}\label{subsec:distr}
We now find the analytical form of $C_{\bm{\W}}$ and its density, if it
exists. To this end, if all $\W_j$ in~\eqref{eq:W:transformed:copula} are
identical, we call $C_{\bm{\W}}$ \emph{homogeneous W-transformed copula}. Also,
the \emph{$C$-volume} of a copula $C$ of the hyperrectangle
$B=\prod_{j=1}^d(a_j,b_j]$ is
\begin{align*}
  V_C(B)=\Delta_B C=\sum_{\bm{i}\in\{0, 1\}^d}(-1)^{\sum_{j=1}^d i_j}C(a_1^{i_1}b_1^{1-i_1},
  \dots,a_d^{i_d}b_d^{1-i_d});
\end{align*}
for $\bm{U}\sim C$, note that $V_C(B)=\Delta_B C=\P(\bm{U}\in B)$.

The following theorem provides the closed-form expression of $C_{\bm{\W}}$
in terms of $C$, which is the main result of this section.
\begin{theorem}[W-transformed copulas and their densities]\label{theorem:W:trafo:cop}
  For $j=1, \dots, d$, let $\W_j:[0, 1]\rightarrow[0, 1]$ be a W-transform with
  change points $\delta_{j,k}$ for $k\in\{1,\dots, K_j\}$, $K_j\in\bar{\IN}$,
  where $\delta_{j,1}=0$ and $\delta_{j,K_j}=1$. Let
  $\W_{j|k}=\W_j|_{(\delta_{j,k-1},\delta_{j,k}]}$, $k=1,\dots,K_j$, be the
  piecewise restrictions of $\W_j$ and suppose each $\W_j$ has its increasing
  (decreasing) pieces indexed by $I_{j}\subseteq\{1, \dots, K_j\}$ ($I_j^C$),
  where $\W_{j|k}$ is increasing if and only if $k\in I_j$.  Then the
  distribution function of $\bm{W}(\bm{U})=(\W_1(U_1),\dots,\W_d(U_d))$ for
  $\bm{U}\sim C$ is given by the copula
  \begin{align}
    C_{\bm{\W}}(\bm{u})=
    \sum_{k_d=1}^{K_d}\dots\sum_{k_1=1}^{K_1}\Delta_{B_{\bm{\delta}_{\bm{k}},\bm{\W^{-1}(\bm{u})},\bm{I}}}C,
    \label{eq:W:trafo:cop}
  \end{align}
  where
  \begin{align*}
    B_{\bm{\delta}_{\bm{k}},\bm{\W^{-1}(\bm{u})},\bm{I}}
    = \prod_{j=1}^d\mathcal{I}_j^{(k_j)}, \quad\mathcal{I}_j^{(k_j)}=
    \begin{cases}
      (\delta_{j, k_j-1}, \W^{-1}_{j|k_j}(u_j)], & k_j\in I_j,\\
      (\W^{-1}_{j|k_j}(u_j),\delta_{j,k_j}], & k_j\notin I_j.
    \end{cases}
  \end{align*}
  Moreover, if $C$ has density $c$, then $C_{\bm{\W}}$ has density
  \begin{align*}
    c_{\bm{\W}}(\bm{u})=\sum_{\substack{j\in\{1,\dots,d\}:\\k_j\in N_j(u_j)}}
    \,\prod_{\ell=1}^d
    \frac{(-1)^{d+\sum_{m=1}^d\I_{I_m}(k_m)}
    c(\bm{u}_{\bm{W}})}{\W'_{\ell|k_\ell}(\W^{-1}_{\ell|k_\ell}(u_\ell))},
  \end{align*}
  where $\bm{u}_{\bm{W}}=(u_{W_1},\dots,u_{W_d})$ with
  $u_{W_j}=\mathcal{W}^{-1}_{j|k_1}(u_j)$, $j\in\{1, \dots, d\}$, and
  $N_j(u_j)\coloneqq\{k\in \{1,\dots,K_j\}:\W^{-1}_{j|k}(u_j) \in(\delta_{j,k-1},
  \delta_{j,k})\}$.
\end{theorem}

We see from~\eqref{eq:W:trafo:cop} that the W-transformed copula
$C_{\bm{\W}}$ is a sum of volumes of $C$. The change points of
$\W_1,\dots,\W_d$ induce a rectilinear grid inside the unit
hypercube $[0,1]^d$, and the piecewise monotonicity of each W-transform
determines at which corners of the rectilinear grid the $C$-volume is evaluated.
The following examples highlight this novel construction.
\begin{example}[Special cases]\label{eg:w:trans:volume}
  \begin{enumerate}
  \item
    \emph{Reflection of copulas.} For $a,b\in[0,1]$, \cite[Exercise~2.6]{nelsen1999} defined the copula
    \begin{align*}
      K_{a,b}(u_1,u_2)=\Delta_{[a(1-u_1),u_1+a(1-u_1)]\times[b(1-u_2),u_2+b(1-u_2)]}C.
    \end{align*}
    Consider the W-transform $\W(u;\delta)=\begin{cases}
      1-u/\delta, & u\in [0,\delta],\\
      (u-\delta)/(1-\delta), & u\in(\delta,1],
    \end{cases}$
    which is a piecewise linear v-transform, for $\delta\in(0,1)$ and $\W(u;0)=u$, $\W(u;1)=1-u$.
    Then $K_{a,b}$ is obtained by applying
    $\W(u;\delta)$ with $\delta=a$ ($\delta=b$) to the first (second) margin of $(U_1,U_2)\sim C$.
    Specifically, $K_{0,1}(u_1,u_2)=u_1-C(u_1,1-u_2)$ is a reflection of $C$ in the second component
    (with stochastic representation $(U_1,1-U_2)$), $K_{1,0}(u_1,u_2)=u_2-C(1-u_1,u_2)$ is a
    reflection of $C$ in the first component (with stochastic representation $(1-U_1,U_2)$) and
    $K_{1,1}(u_1,u_2)=-1+u_1+u_2+C(1-u_1,1-u_2)$ is the \emph{survival copula} $\hat{C}$
    of $C$ (with stochastic representation $(1-U_1,1-U_2)$).
  \item\label{item:inter:mono}
    \emph{$C_{\bm{\W}}$ is a sum of volumes of $C$.} Consider $(U_1, U_2)\sim C$,
    $\W_1(u) = \bigl|3|u-\frac{2}{3}|-1\bigr|$, and $\W_2(u)=1-|2u-1|$, $u\in[0, 1]$.
    Then one has $\delta_{1,k_1}=k_1/3$ for $k_1\in\{0, 1, 2, 3\}$, $\delta_{2,k_2}=k_2/2$
    for $k_2\in\{0,1,2\}$, $I_1=\{2\}$ and $I_2=\{1\}$. For any $u_1, u_2\in[0,1]$,
    the W-transformed copula $C_{\bm{\W}}$ is given by
    \begin{align*}
      C_{\bm{\W}}(u_1, u_2)&=\Delta_{\bigl(\frac{1-u_1}{3},\frac{1+u_1}{3}\bigr]\times
                             \bigl(0, \frac{u_2}{2}\bigr]}C+
                             \Delta_{\bigl(\frac{1-u_1}{3},\frac{1+u_1}{3}\bigr]\times
                             \bigl(\frac{2-u_2}{2}, 1\bigr]}C\,+\\
                           &\phantom{={}}\Delta_{\bigl(\frac{3-u_1}{3}, 1\bigr]\times
                             \bigl(0, \frac{u_2}{2}\bigr]}C+
                             \Delta_{\bigl(\frac{3-u_1}{3}, 1\bigr]\times\bigl(\frac{2-u_2}{2}, 1\bigr]}C,
    \end{align*}
    so $C_{\bm{\W}}(u_1, u_2)$ is obtained by
    summing the volumes of $C$ in the area depicted by the shaded region in
    the top-left panel of Figure~\ref{fig:grid:volume}.
    \begin{figure}[htbp]
      \centering
      \hspace{0.3cm}
      \includegraphics[width=0.46\textwidth]{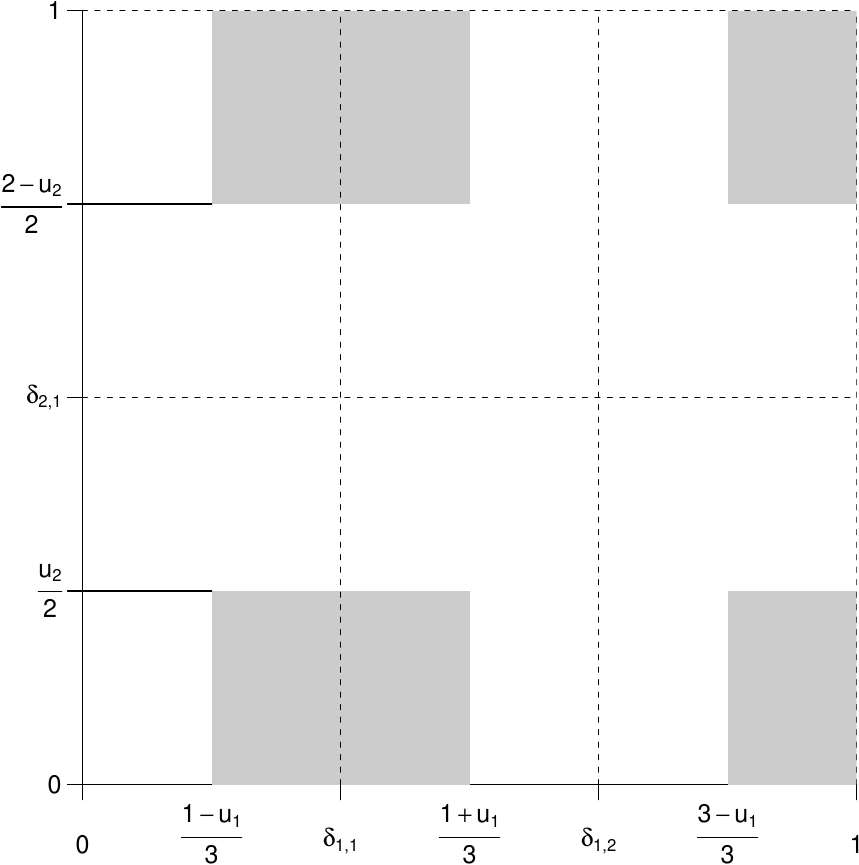}
      \hfill
      \includegraphics[width=0.46\textwidth]{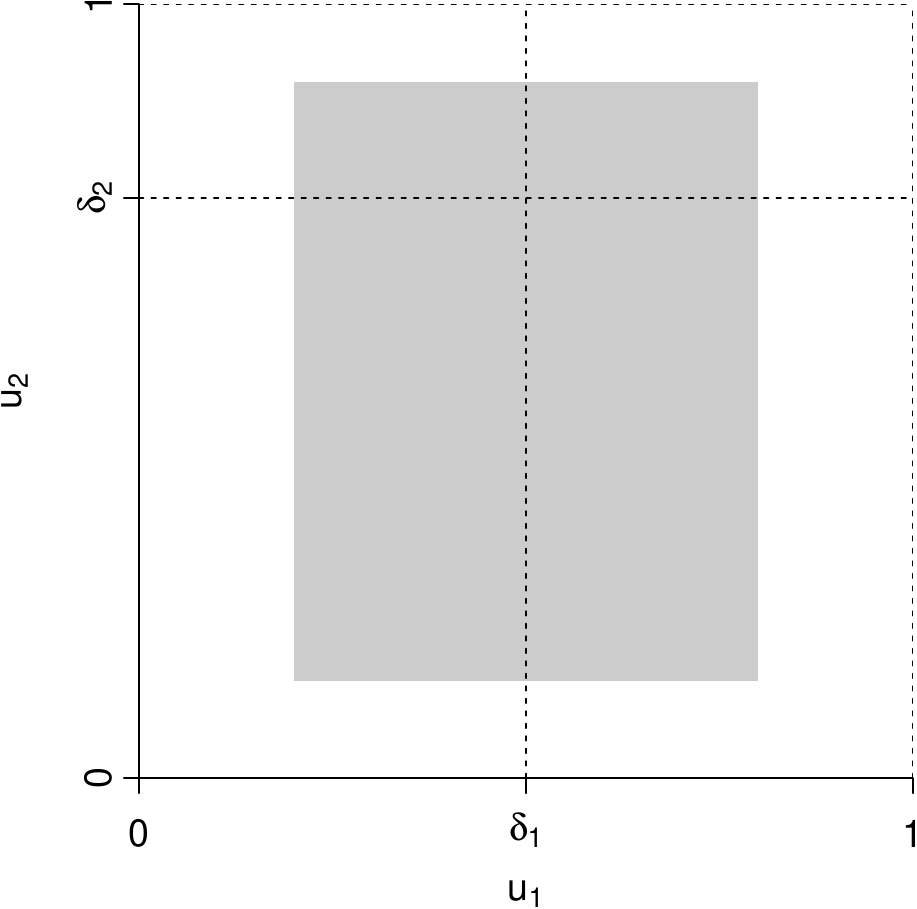}\\[5mm]
      \includegraphics[width=0.49\textwidth]{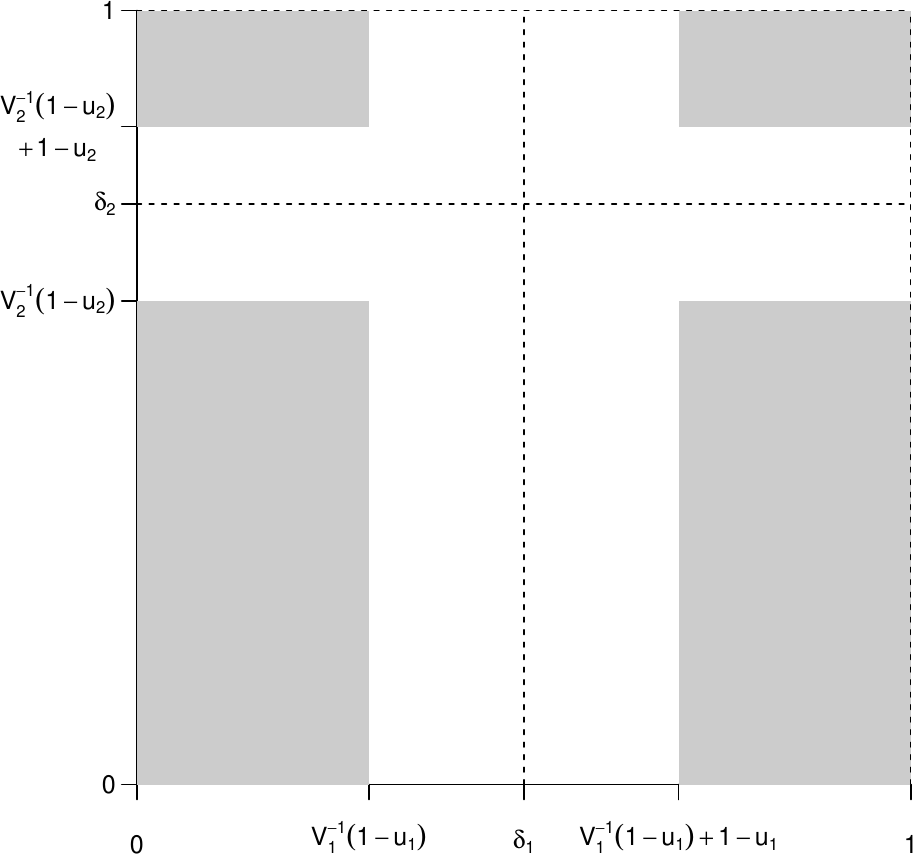}
      \hfill
      \includegraphics[width=0.45\textwidth]{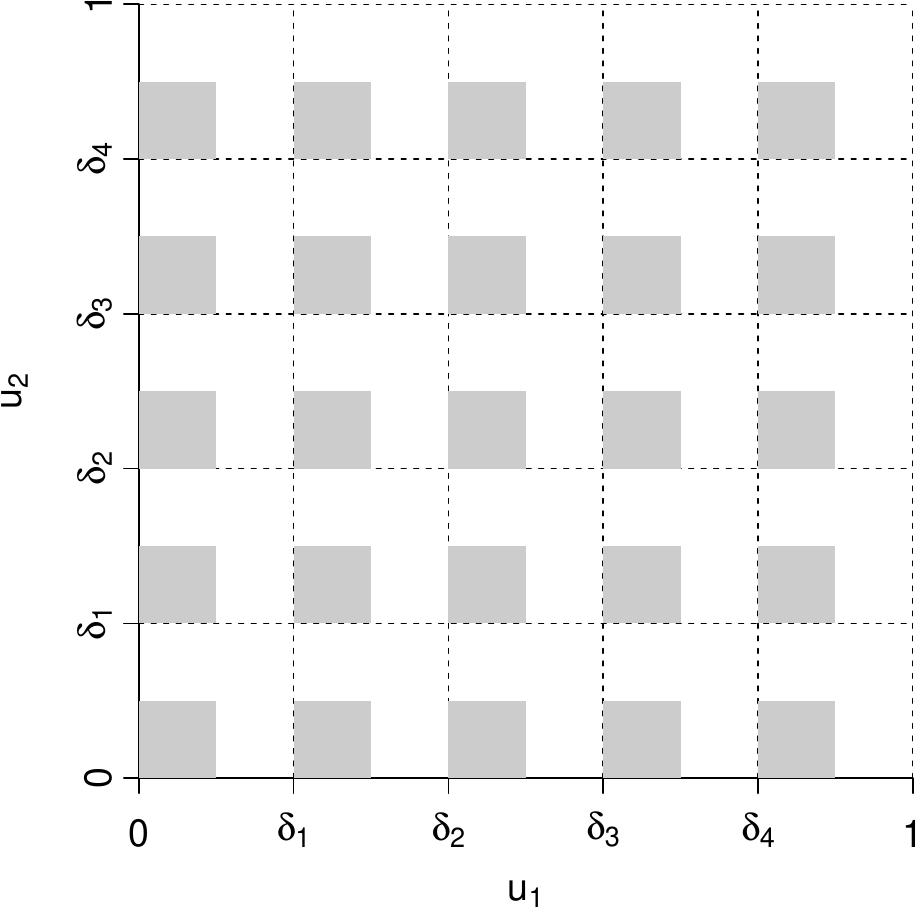}
      \caption{The shaded
        areas depict the rectangular regions of which the volumes of $C$
        are summed up to determine the value of $C_{\bm{\W}}$, and
        this for four different $\bm{\W}$ (top-left: general
        W-transform; top-right: v-transform; bottom-left: flipped
        v-transform; bottom-right: piecewise increasing W-transform).  All
        vertices strictly inside $[0, 1]^2$ are determined by
        applying the respective piecewise inverse of
        $\bm{\W}$ componentwise.}
      \label{fig:grid:volume}
    \end{figure}
  \item\label{eg:w:trans:volume:v:trans} \emph{(Flipped) v-transformed copulas.} For $j=1,\dots,d$, let
    $\mathcal{V}_j$ be as in \eqref{eq:vtrans}, so a special W-transform
    with change points $\delta_{j,0}=0$, $\delta_{j,1}=\delta_j$,
    $\delta_{j,2}=1$. Let
    $\mathcal{V}^{-1}_j:(0, 1]\rightarrow [0,\delta_j)$ be the inverse of
    the left branch of $\mathcal{V}_j$. Let $C$ be any $d$-dimensional copula.  Then by
    \eqref{eq:W:trafo:cop} and Proposition~\ref{prop:W:property}~\ref{prop:W:property:partition:sq},
    the W-transformed (or here: v-transformed) copula
    $C_{\bm{\mathcal{V}}}$ is
    \begin{align*}
      C_{\bm{\mathcal{V}}}(u_1, \dots, u_d)=\Delta_BC,\quad B= \prod_{j=1}^d(\mathcal{V}_j^{-1}(u_j),
      \mathcal{V}^{-1}_j(u_j)+u_j], \quad \bm{u}\in[0, 1]^d.
    \end{align*}
    For $j=1,\dots,d$, consider the \textit{flipped v-transform}
    $\mathcal{V}^\ast_j=1-\mathcal{V}_j$, which is also a special
    W-transform (but not a v-transform) with the same change points as $\mathcal{V}_j$ but with the
    monotonicity flipped on each piece. The inverse of the left branch of
    $\mathcal{V}^\ast_j$ is
    $\mathcal{V}^{*\,-1}_j(v)=\mathcal{V}_j^{-1}(1-v)$, $v\in(0,1]$.

    By Theorem~\ref{theorem:U:V:cop}, the copula obtained by marginally applying
    the flipped v-transforms is
    \begin{align*}
      C_{\bm{\mathcal{V}}^\ast}(u_1, \dots, u_d)=\sum_{k_d=1}^2\cdots\sum_{k_1=1}^2
      \Delta_{B_{\bm{\delta}, \bm{\mathcal{V}^{-1}(\bm{u}),\bm{I}}}}C, \quad \bm{u}\in[0, 1]^d,
    \end{align*}
    where
    \begin{align*}
      B_{\bm{\delta}, \bm{\mathcal{V}^{-1}(\bm{u}),\bm{I}}}=
      \prod_{j=1}^d\mathcal{I}_j^{(k_j)}, \quad\mathcal{I}_j^{(k_j)}=
      \begin{cases}
        (0,\mathcal{V}^{-1}_j(1-u_j)], & k_j =1,\\
        [\mathcal{V}^{-1}_j(1-u_j)+1-u_j, 1], & k_j=2.
      \end{cases}
    \end{align*}
    We thus see that $C_{\bm{\mathcal{V}}}$ is obtained by evaluating $C$-volumes of hyperrectangles
    at the centre of $[0,1]^d$ (see the top-right panel of Figure~\ref{fig:grid:volume} for $d=2$,
    $\mathcal{V}_1(u)=|2u-1|$ and
    $\mathcal{V}_2(u)=\begin{cases}
      2-\sqrt{1+4u}, & u\in[0, 0.75],\\
      2\sqrt{u-0.75}, & u\in(0.75, 1].
    \end{cases}$), while
    $C_{\bm{\mathcal{V}}^\ast}$ is obtained by evaluating $C$-volumes of hyperrectangles anchored at the
    corners of $[0,1]^d$ (see the bottom-left panel of Figure~\ref{fig:grid:volume} for $d=2$).
    Notably, the behaviour of $C_{\bm{\mathcal{V}}}$ and $C_{\bm{\mathcal{V}}^\ast}$ are opposites
    as $\bm{u}$ approaches $\bm{0}$ and $\bm{1}$. As $\bm{u}\to\bm{0}$, the lower tail of
    $C_{\bm{\mathcal{V}}}$ is aggregated by the centre volume of $C$, while the lower tail of
    $C_{\bm{\mathcal{V}}^\ast}$ is aggregated by the volumes at the four corners. Conversely, as
    $\bm{u}\to1$, their upper tails are aggregated by the four corners and the centre volume,
    respectively.
  \item \emph{Piecewise monotone W-transformed copulas.} For $j=1,\dots,d$,
    consider $d$ piecewise increasing (decreasing) W-transforms $\W_j$ with change points
    $\delta_{j,k}$, $j\in\{1,\dots, d\}$, $k\in\{1,\dots, K_j\}$. By
    Theorem~\ref{theorem:W:trafo:cop}, the joint distribution of
    $(\W_1(U_1), \dots, \W_d(U_d))$ is given by
    \begin{align*}
      C_{\bm{\W}}(u_1,\dots, u_d) &= \sum_{k_d=1}^K\cdots
                                    \sum_{k_1=1}^K \Delta_{B_{\bm{\delta}, \bm{W}^{-1}(\bm{u})}}C, \quad \bm{u}\in[0, 1]^d,
    \end{align*}
    where
    \begin{align*}
      B_{\bm{\delta}, \bm{W}^{-1}(\bm{u})}=\begin{cases}
        \prod_{j=1}^d\bigl(\delta_{j,k_j-1}, \W^{-1}_{j|k_j}(u_j)\bigr],& \text{if each }
                                                                          \W \text{ is piecewise increasing,}\\
        \prod_{j=1}^d\bigl(\W^{-1}_{j|k_j}(u_j), \delta_{j,k_j}\bigr], & \text{if each }
                                                                         \W \text{ is piecewise decreasing.}
      \end{cases}
    \end{align*}
    We deduce that for piecewise increasing (decreasing) W-transforms, the
    $C$-volumes are always evaluated at hyperrectangles anchored at the lower
    (upper) corner of each grid cell of the rectilinear grid. As a concrete
    example, consider $\W_1=\dots=\W_d=:\W$
    with $\W(u)=5u-\lceil 5u\rceil+1$ for the change points
    $\delta_k=k/5$, $u\in[0, 1]$, $k= 0,\dots,5$. Then the shaded area in the
    bottom-right panel of Figure~\ref{fig:grid:volume} displays the regions
    over which the volumes of $C$ are aggregated to get the values of the homogeneous
    W-transformed copula $C_{\bm{\W}}$.
  \end{enumerate}
\end{example}

\begin{remark}[Computational aspects]
  The density formula in Theorem~\ref{theorem:W:trafo:cop} involves at most $\prod_{j=1}^dK_j$ terms.
  In practice, this poses a computational burden when $d$ is large. However, we note three points
  that collectively show this complexity is not a practical barrier to the framework's applicability.
  First, for low dimensions $d\in\{2,3,4\}$, the number of summands is small and likelihood evaluation is
  straightforward. Second, as we shall present in Section~\ref{subsec:tail}, for large $d$, the ordinal
  sum construction with piecewise increasing W-transforms (Equation~\ref{eq:os:in:W:cop}) yields
  a density whose evaluation cost is linear in the number of pieces $K$, remaining fully tractable
  regardless of dimension. For W-transforms with interchanging monotonicity over adjacent pieces, many
  rectangular regions merge, substantially reducing the number of summands.
  Examples~\ref{eg:w:trans:volume}~\ref{item:inter:mono} and~\ref{eg:w:trans:volume:v:trans}
  require only 4 and 1 evaluations of volumes (rather than 6 and 4), respectively. The effective
  cost has a lower bound of $O(\prod_{j=1}^d(K_j/2))$ in such cases, far below the worst case
  $O(\prod_{j=1}^dK_j)$.
\end{remark}

As we have seen in~\eqref{eq:W:trafo:cop} of Theorem~\ref{theorem:W:trafo:cop},
the value of the W-transformed copula $C_{\bm{\W}}$ is a sum of
$C$-volumes. We now present the analytical form of
$C_{\bm{\W}}$-volumes and their relationship to $C$-volumes through the
W-transforms $\W_1, \dots, \W_d$.
\begin{proposition}[Volume of $C_{\bm{\W}}$]\label{prop:W:volume}
  Let $C$ be a copula, $\W_1,\dots, \W_d$ be W-transforms, and $C_{\bm{\W}}$
  the corresponding W-transformed copula. Then the $C_{\bm{\W}}$-volume of $(\bm{a}, \bm{b}]$
  with $\bm{0}\leq\bm{a}\leq\bm{b}\leq\bm{1}$ is
  \begin{align*}
    \Delta_{(\bm{a},\bm{b}]}C_{\bm{\W}}=\sum_{k_d=1}^{K_d}\cdots\sum_{k_1=1}^{K_1}
    \Delta_{B_{\bm{\W^{-1}}(\bm{a}),\bm{\W^{-1}}(\bm{b}),\bm{I}}}C,
  \end{align*}
  where
  \begin{align*}
    B_{\bm{\W^{-1}}(\bm{a}),\bm{\W^{-1}}(\bm{b}),\bm{I}}
    =\prod_{j=1}^d\bigl(\W_{j|k_j}^{-1}(a_j^{\I\{k_j\in I_j\}}b_j^{\I\{k_j\notin I_j\}}),
    \W_{j|k_j}^{-1}(a_j^{\I\{k_j\notin I_j\}}b_j^{\I\{k_j\in I_j\}})\bigr).
  \end{align*}
\end{proposition}
It follows from Proposition~\ref{prop:W:volume} that, as $C_{\bm{\W}}$ values (see~\eqref{eq:W:trafo:cop}),
$C_{\bm{\W}}$-volumes are also sums of $C$-volumes, but instead
of hyperrectangles anchored at the corners of each rectilinear grid, hyperrectangles in the centre enter.

\subsection{Tail dependence}\label{subsec:tail}
In this section, we study tail dependence of W-transformed copulas.  As
W-transforms allow one to introduce (tail) asymmetry, we start by considering an
asymmetric notion of tail dependence, namely the notion of maximal tail
concordance of \cite{koikekatohofert2023}. We then focus on homogeneous
W-transformed copulas and study the influence of W-transforms on the tail
dependence coefficients, defined by
$\lambda_{\text{u}}=\lim_{t\rightarrow 1-}(1-2t+C(t,t))/(1-t)$ in the upper
and $\lambda_{\text{l}}=\lim_{t\rightarrow 0+}C(t, t)/t$ in the lower
tail.

In terms of the \emph{(lower) tail copula} $\Lambda(x,y;C)=\lim_{p\to0+}C(px,py)/p$, $(x,y)\in [0,\infty)^2$, of a
copula $C$, the \emph{maximal tail concordance measure (MTCM)} of $C$ is
\begin{align*}
  \lambda_{\operatorname{MTCM}}^\ast(C) = \sup_{b\in(0,\infty)}\Lambda(b,1/b;C),
\end{align*}
which equals $\Lambda(b^*,1/b^*;C)$ if a unique maximiser $b^*\in(0,\infty)$
exists.  The following result provides the MCTM of flipped v-transformed
copulas, showcasing how flipped v-transforms (as in
Example~\ref{eg:w:trans:volume}~\ref{eg:w:trans:volume:v:trans}) affect the
direction and the magnitude of the MTCM of $C$.
\begin{proposition}[MTCM of flipped v-transformed copulas]\label{prop:flip:v:trans:tail}
  Let $C$ be a bivariate copula with MTCM
  $\lambda_{\operatorname{MTCM}}^\ast(C)=\Lambda(b^*,1/b^*;C)$ for some
  $b^*\in(0,\infty)$. Consider a flipped v-transform $\mathcal{V}^*$
  with change point $\delta$.
  If $C$ is tail independent in the upper-left, upper and lower-right tails %
  and lower tail dependent, then the flipped v-transformed copula
  $C_{\bm{\mathcal{V}}^*}$ has MTCM
  \begin{align*}
    \Lambda(b^*_{C_{\bm{\mathcal{V}}^*}},1/b^*_{C_{\bm{\mathcal{V}}^*}};C_{\bm{\mathcal{V}}^*})=\sqrt{\alpha_1\alpha_2}\Lambda(b^*,1/b^*;C),
  \end{align*}
  where $b^*_{C_{\bm{\mathcal{V}}^*}}=\sqrt{\alpha_2/\alpha_1}b^*$,
  $\alpha_1=(\mathcal{V}_{1|1}^{*\,-1})'(0+)=(\mathcal{V}_{1|1}^{-1})'(1-)$ and
  $\alpha_2=(\mathcal{V}_{2|1}^{*\,-1})'(0+)=(\mathcal{V}_{2|1}^{-1})'(1-)$.
\end{proposition}
Proposition~\ref{prop:flip:v:trans:tail} implies that if
$\alpha_1,\alpha_2\neq0$, the MTCM of flipped v-transformed copulas is attained
along the line with slope
$1/b^{*\,2}_{C_{\bm{\mathcal{V}}^*}}=\alpha_1/(\alpha_2b^{*\,2})$ and intercept
$0$. Otherwise if $\alpha_1$ or $\alpha_2$ is $0$, then the flipped
v-transformed copula is lower tail independent. Furthermore, since
$\alpha_1,\alpha_2\in[0,1]$, the value of MTCM is scaled down by the factor
$\sqrt{\alpha_1\alpha_2}$.

We now turn our attention to homogeneous W-transformed copulas, and study their tail dependence coefficients.
As we have seen in Example~\ref{eg:w:trans:volume}, the tails
of the W-transformed copula depends on the centre volume of $C$ in general, but how $C$ behaves in
the centre is indeterminate. We therefore turn to v-transforms and the v-transformed copula to derive
the upper tail dependence coefficient $\lambda_{\text{u}}$.
However, we are not able to do so for the lower tail dependence coefficient $\lambda_{\text{l}}$, for the
lower tail of $C_{\mathcal{V}}$ depends on the centre of $C$.

\begin{proposition}[Upper tail dependence for $C_{\mathcal{V}}$]\label{prop:tail:coeff}
  Let $C$ be a copula with tail dependence coefficients
  $\lambda_{\text{l}}, \lambda_{\text{u}}$, and let $\mathcal{V}$ be a
  v-transform with change point $\delta$. In terms of the $j$th piece
  $\mathcal{V}_{|k}$, $k=1,2$, of $\mathcal{V}$, the upper tail dependence
  coefficient $\lambda_{\text{u}}^{C_\mathcal{V}}$ of the homogeneous
  v-transformed copula is
  \begin{align*}
    \lambda_{\text{u}}^{C_\mathcal{V}}&=\frac{1}{-\mathcal{V}_{|1}'(0+)}\lambda_{\text{l}}^{C}+\biggl(1-\frac{1}{-\mathcal{V}_{|1}'(0+)}\biggr)\lambda_{\text{u}}^{C}\\
                                      &\mathrel{\phantom{=}}+\frac{2}{-\mathcal{V}_{|1}'(0+)}-\lim_{t\to1-}\frac{C(\mathcal{V}^{-1}_{|1}(t)+t,
                                        \mathcal{V}^{-1}_{|1}(t))+C(\mathcal{V}^{-1}_{|1}(t), \mathcal{V}^{-1}_{|1}(t)+t)}{1-t}.
  \end{align*}
  In particular, if $C$ is tail-independent in the upper-left and lower-right tails, then
  \begin{align*}
    \lambda_{\text{u}}^{C_\mathcal{V}}&=\frac{1}{-\mathcal{V}_{|1}'(0+)}\lambda_{\text{l}}^{C}+\biggl(1-\frac{1}{-\mathcal{V}_{|1}'(0+)}\biggr)\lambda_{\text{u}}^{C}.
  \end{align*}
\end{proposition}

Proposition~\ref{prop:tail:coeff} says that the upper tail dependence
coefficient is a convex combination of $\lambda_{\text{l}}^C$
and $\lambda_{\text{u}}^C$ (plus a constant if $C$ is not tail independent in the
upper-left and lower-right tails). This structure creates an opportunity to design copulas
$C$ that yield W-transformed copulas with specific tail properties.
While Lemma~\ref{lem:unif:preserv} implies that W-transforms
generally redistribute probability mass away from the tails (a geometric ``dragging'' effect that could
lead to a decrease in tail dependence), Proposition~\ref{prop:tail:coeff}
does not rule out the possibility of enhancing tail dependence through a
strategic choice of $C$ and $\W$. In Section~\ref{subsec:distr} we have seen that
W-transformed copulas are sums of volumes of $C$, and such volumes are
determined precisely by the change points of $\W$. This connection motivates us to consider
\textit{ordinal sum} copulas investigated by \cite[Section~2.3.3]{nelsen1999} in
the bivariate case and by \cite{mesiar2010ordinal} in the multivariate case,
given by
\begin{align}\label{eq:ordinal:sum}
  C_S(\bm{u}) = \sum_{k=1}^K(\delta_k-\delta_{k-1})C_k\Bigl(\min\Bigl\{\max\Bigl\{
  \dfrac{\bm{u}-\bm{\delta}_{k-1}}{\delta_{k}-\delta_{k-1}}, \bm{0}\Bigr\}, \bm{1}\Big\}\Bigr),
\end{align}
where $C_1, \dots, C_K$ are copulas and $\min\{\bm{u}\}$, $\max\{\bm{u}\}$ denote the elementwise
minimum and maximum of the vector $\bm{u}$. Take a piecewise surjective and
increasing W-transform $\W$. The homogeneous W-transformed copula
$C_{S,\W}$ then is
\begin{align}\label{eq:os:homo:in:W:cop}
  C_{S,\W}(\bm{u})=\sum_{k=1}^K(\delta_{k}-\delta_{k-1})C_k(G_{k}(u_1),\dots,G_{k}(u_d)),\quad
  G_{k}(u) \coloneqq \begin{cases}
    0, & u=0,\\
    \frac{\W^{-1}_{|k}(u)-\delta_{k-1}}{\delta_k-\delta_{k-1}} & u\in(0,1),\\
    1, & u=1.
  \end{cases}
\end{align}
To facilitate the application of W-transformed ordinal sums in
Section~\ref{subsec:cop:ordinal:sum} later, we now derive the analytical form of
the coefficients of tail dependence of $C_{S,\W}$.
\begin{proposition}[Tail dependence coefficients of $C_{S, \W}$]\label{prop:os:tail:dependence}
  Consider a piecewise surjective and increasing W-transform $\W$ with change points
  $\{\delta_k\}_{k=0}^K$. Define $G_{k}(u)$ as in \eqref{eq:os:homo:in:W:cop} and
  its derivative $g_{k}(u)=\frac{\rd }{\rd u} G_{k}(u)$ which exists
  almost everywhere on $[0,1]$. Let $d=2$ and $\lambda_{\text{l}}, \lambda_{\text{u}}$
  be the lower and upper tail dependence coefficients of the homogeneous W-transformed copula $C_{S,\W}$
  as in~\eqref{eq:os:homo:in:W:cop}. Denote by $\lambda_{\text{l},k}, \lambda_{\text{u},k}$ the lower and upper
  tail dependence coefficients of the component copula $C_k$, $k=1,\dots,K$. Then,
  \begin{align*}
    \lambda_\text{l}=\sum_{k=1}^K\alpha_k\lambda_{\text{l},k}\quad \text{and} \quad
    \lambda_\text{u}=\sum_{k=1}^K\beta_k\lambda_{\text{u},k},
  \end{align*}
  where $\alpha_k=(\delta_k-\delta_{k-1})g_k(0+)\ge0$, $\beta_k=(\delta_k-\delta_{k-1})g_k(1-)\ge0$
  and $\sum_{k=1}^K\alpha_k= \sum_{k=1}^K\beta_k=1$.
\end{proposition}
A similar result can be derived for general piecewise surjective W-transforms, since, if the
$k$th piece of $\W$ is decreasing, the $k$th component of the ordinal sum contributes its upper
(lower) tail mass to the lower (upper) tail of the W-transformed ordinal sum, scaled by
$(\delta_k-\delta_{k-1})g_k(1-)$ ($(\delta_k-\delta_{k-1})g_k(0+)$).

\begin{remark}[Corrections of \cite{quessy2024general}]
  \cite{quessy2024general} has presented results on the coefficients of tail
  dependence under pcsm W-transforms with interchanging monotonicity between
  neighbouring pieces. However, the proof of his Lemma~2 (provided in the
  Supplementary Materials) appears to contain an oversight. Specifically, the
  second term in the first equality following the statement ``Then, an
  application of the general formula yields [...]'' is missing a denominator
  $x$. This omission affects subsequent derivations, leading to conclusions
  that may not hold in general. In particular, Proposition~3, Corollary~1, and
  Proposition~4 rely on Lemma~2, and, as demonstrated by our counterexample in
  what follows, these results do not appear to be valid under the given
  conditions.
\end{remark}

\begin{example}[Tail properties of W-transformed copulas]\label{eg:os:tails}
  \begin{enumerate}
  \item \emph{MTCM of flipped v-transformed copulas.}
    Consider the flipped v-transform $\mathcal{V}^*_\delta(u)=\begin{cases}
      u/\delta & u\in[0, \delta],\\
      (1-u)/(1-\delta) & u\in(\delta, 1].
    \end{cases}$
    Consider a Clayton copula $C$ with Kendall's tau $0.7$. A simulated sample of size 2000
    from $C$ and a sample of the same size from the corresponding flipped v-transformed copula
    $C_{\mathcal{V}^*_{0.2},\mathcal{V}^*_{0.8}}$ are shown in Figure~\ref{fig:flip:v:trans}.
    \begin{figure}[htbp]
      \centering
      \includegraphics[width=0.48\textwidth]{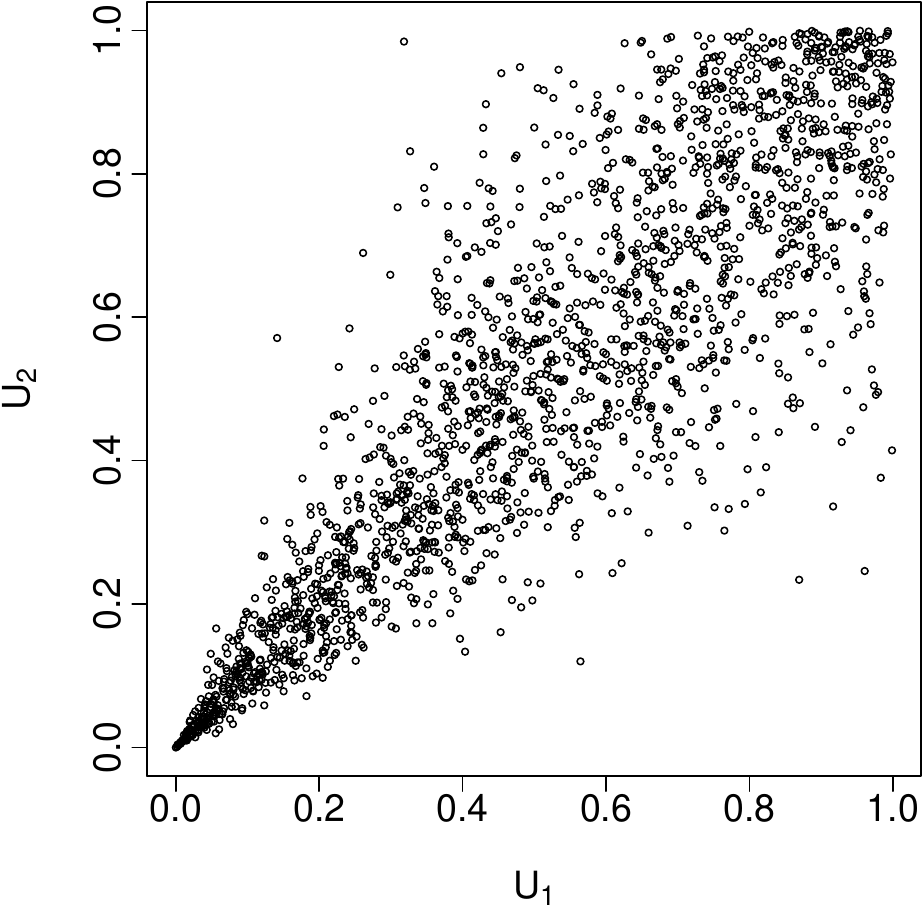}
      \hfill
      \includegraphics[width=0.48\textwidth]{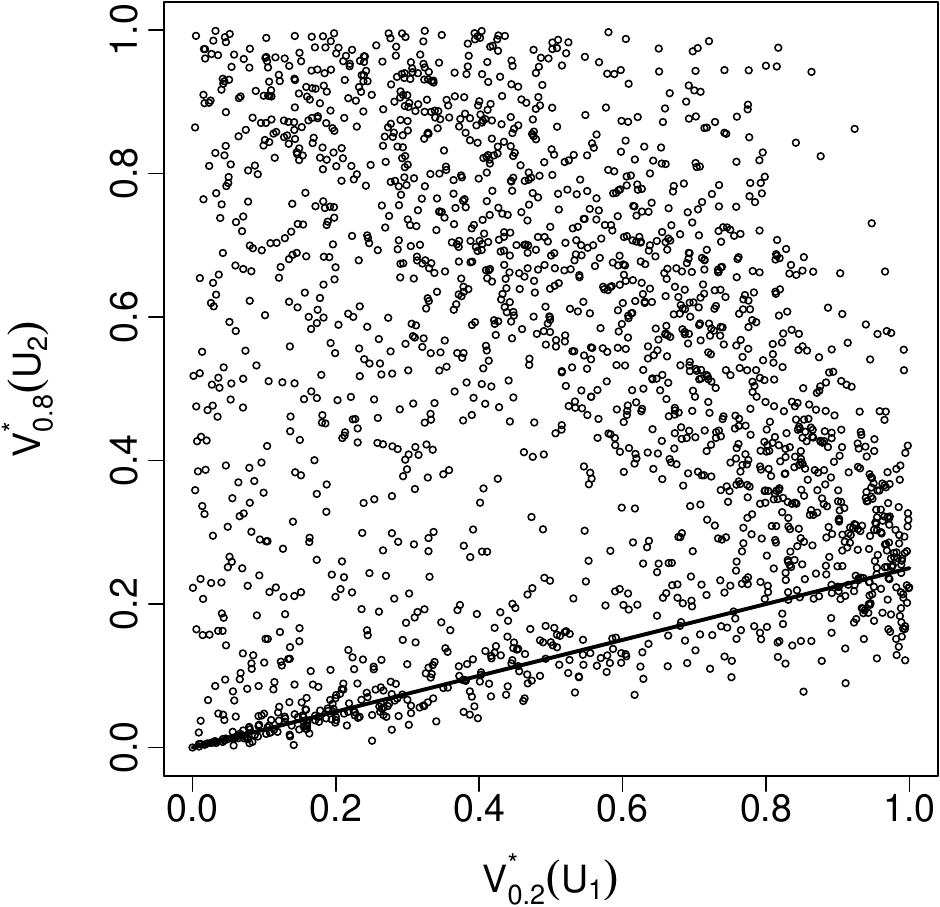}
      \caption{Samples of size 2000 from a Clayton copula $C$ with Kendall's tau $0.7$
        (left) and corresponding flipped v-transformed copula
        $C_{(\mathcal{V}^*_{0.2},\mathcal{V}^*_{0.8})}$ with a black straight line
        indicating the direction in which the MTCM is attained (right).}\label{fig:flip:v:trans}
    \end{figure}
  \item \emph{Counterexample to \cite{quessy2024general}.}
    Consider $\W$ as in Example~\ref{eg:v:trans:use}~\ref{eg:v:trans:use:chara:v}
    which is a v-transform and an ordinal sum based on two survival Gumbel copulas, both
    with Kendall's tau $0.7$. A simulated sample of size 2000 from this ordinal
    sum and a sample of the same size from the corresponding W-transformed ordinal
    sum $C_{S, \mathcal{V}}$ (which exhibits both lower and upper tail dependence)
    are shown in Figure~\ref{fig:counter:eg}.
    \begin{figure}[htbp]
      \centering
      \includegraphics[width=0.48\textwidth]{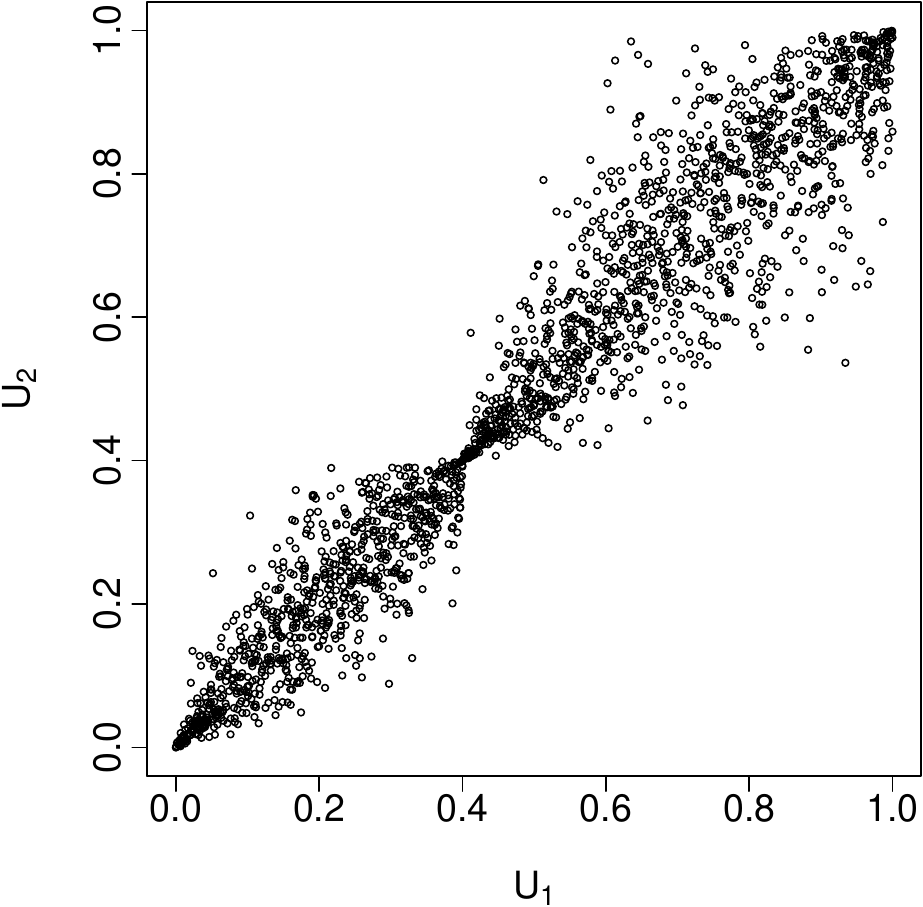}
      \hfill
      \includegraphics[width=0.48\textwidth]{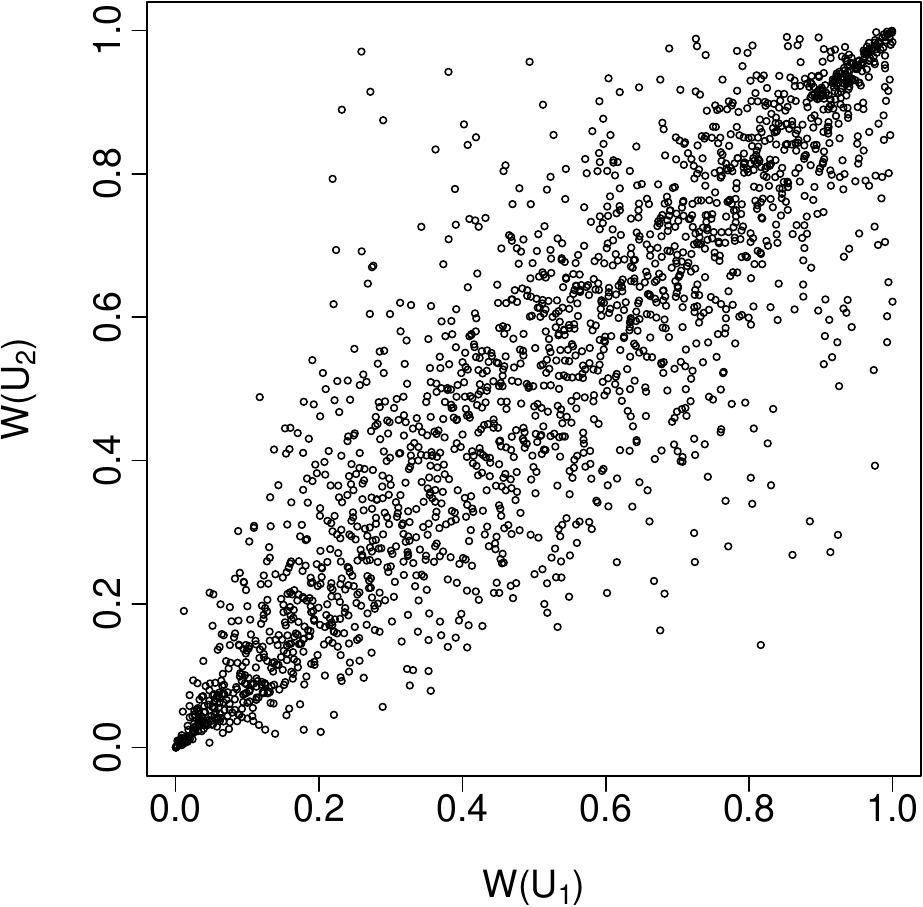}
      \caption{Samples of size 2000 from an ordinal sum $C_{S}$ (left) and
        corresponding W-transformed ordinal sum $C_{S, \mathcal{W}}$ (right).}
      \label{fig:counter:eg}
    \end{figure}
  \item \emph{Modification of Gaussian copula tails.}\label{eg:os:tails:modif:gauss}
    Consider $d=2$. Let $\bm{r}=\bm{1}$ and $F_X(x)=\frac{1}{1+(1/x-1)^a}$, $x\in[0, 1]$,
    $a\in(0,1)$. Then, $f_X(0+)=f_X(1-)=\infty$ and therefore, by Lemma~\ref{lem:param:derivative},
    the induced W-transform $\W_{\bm{t},\bm{r},F_X}$ in \eqref{eq:W:param:fam} satisfies
    $\W_{\bm{t},\bm{r},F_X}(0+)=\W_{\bm{t},\bm{r},F_X}(1-)=1$.
    By Proposition~\ref{prop:os:tail:dependence}, the homogeneous W-transformed ordinal sum
    $C_{S,\bm{\W}_{\bm{t},\bm{r},F_X}}$ with component copulas $C_1,\dots, C_K$ has
    $\lambda_{\text{l}}=\lambda_{\text{l},1}$ and $\lambda_\text{u}=\lambda_{\text{u},d}$.

    Furthermore, for $a=0.5$, $\bm{t}=(0, 0.1, 0.9, 1)$, $\bm{r}=(1,1,1)$,
    let $C_S$ have components $C_1$ (a Clayton copula with $\lambda_{\text{l},1} = 0.5$),
    $C_2$ (a Gaussian copula with correlation parameter $\rho=0.7$ and
    $\lambda_{\text{l},2}=0$), and $C_3$ (a Gumbel copula with $\lambda_{\text{u},3}=0.8$). Then
    the W-transformed ordinal sum
    $C_{S, \W_{\bm{t}, \bm{r} ,F_X}}$ has tail dependence
    coefficients $\lambda_\text{l}=0.5$ and $\lambda_\text{u}=0.8$. A plot of the W-transform
    $\W_{\bm{t},\bm{r},F_X}$ (with change points being
    $\bm{\delta}=F_X(\bm{t})=(0, 0.25, 0.75, 1)$), a sample from the ordinal sum
    $C_S$, and a sample from the corresponding W-transformed copula
    $C_{S, \W_{\bm{t},\bm{r} ,F_X}}$ are shown in
    Figure~\ref{fig:os:example}.
    \begin{figure}[htbp]
      \centering
      \includegraphics[width=0.32\textwidth]{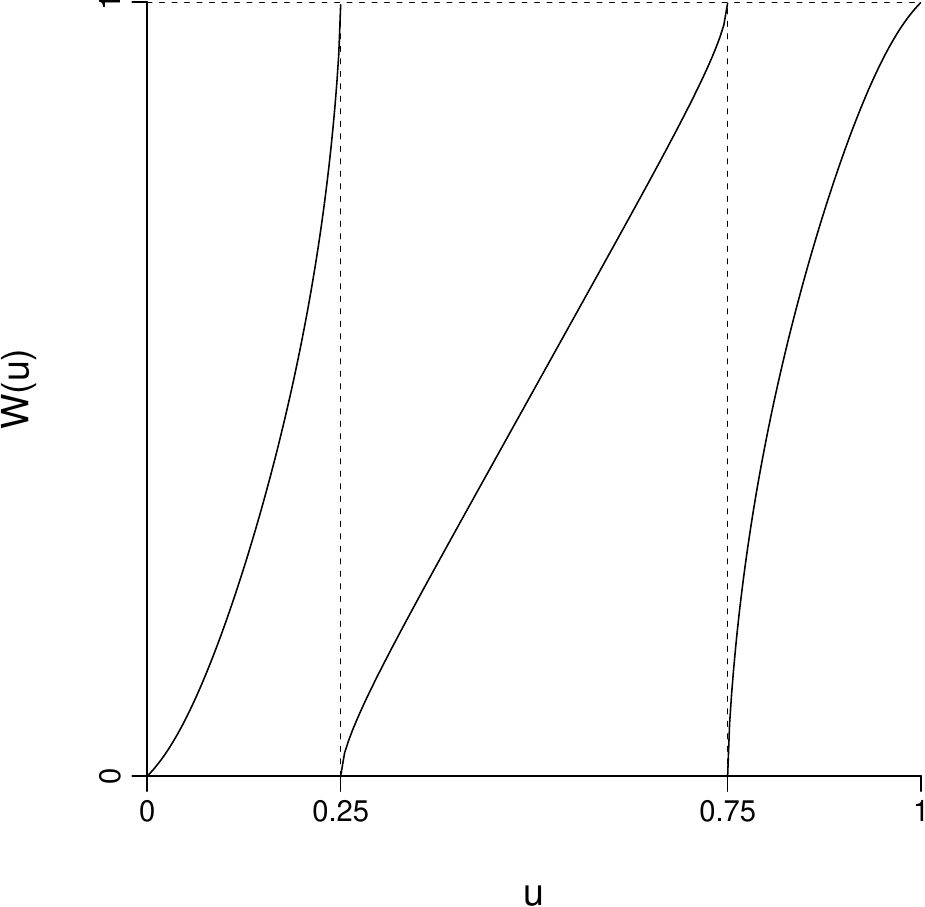}
      \hfill
      \includegraphics[width=0.32\textwidth]{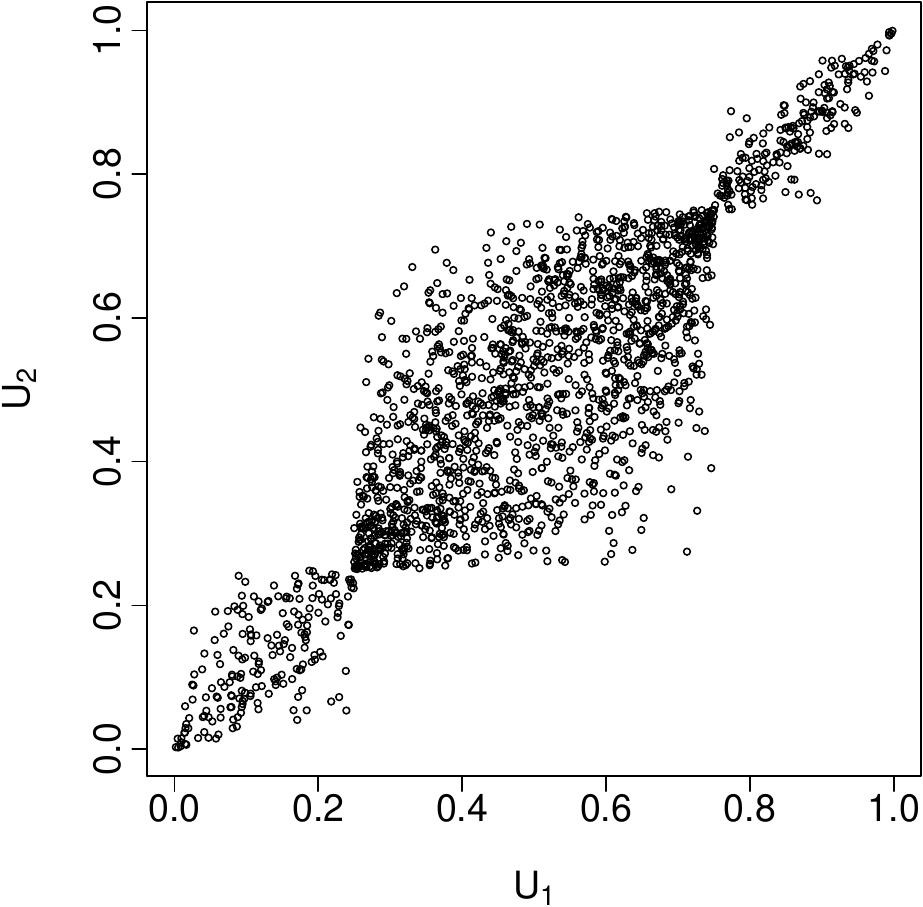}
      \hfill
      \includegraphics[width=0.32\textwidth]{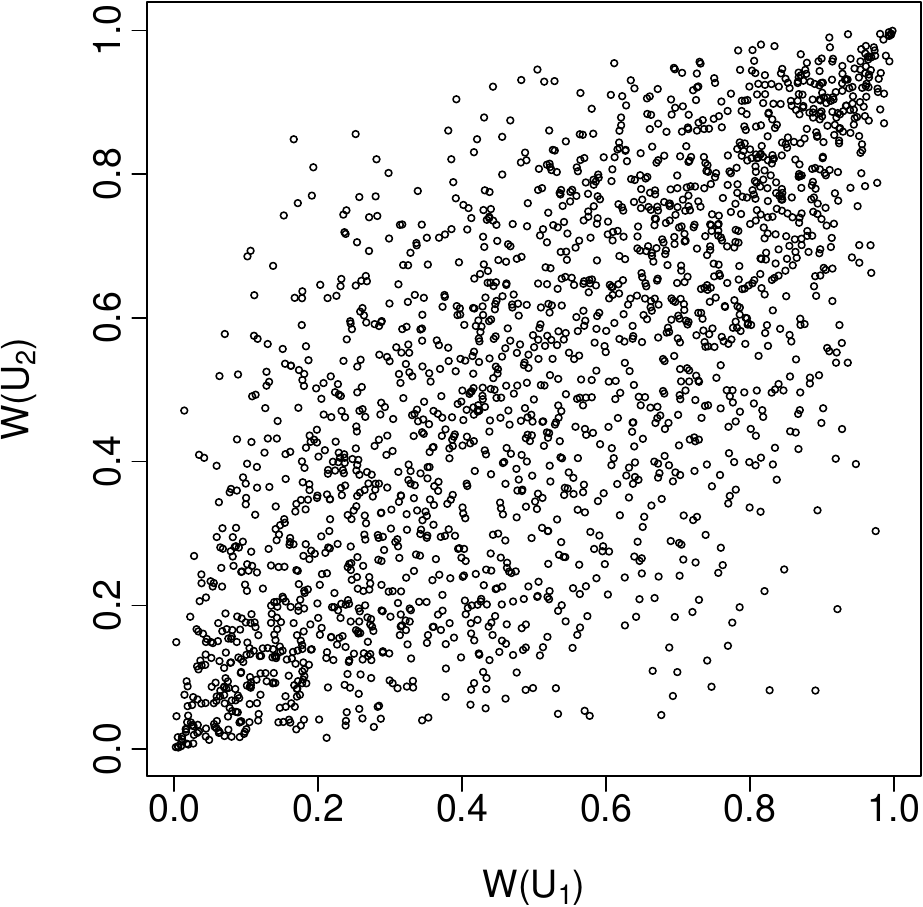}
      \caption{Piecewise increasing W-transform
        $\W_{\bm{t},\bm{r},F_X}$ from~\eqref{eq:W:param:fam}
        (left), a sample of size 2000 from the ordinal sum copula
        $C_S$ under consideration (centre) and a sample of the same
        size from the corresponding W-transformed
        copula $C_{S, \W_{\bm{t},\bm{r} ,F_X}}$ (right).}
      \label{fig:os:example}
    \end{figure}
  \end{enumerate}
\end{example}
We close this subsection with a remark on the interpretation of the tail dependence reduction
discussed above. First, for asymmetric W-transformed copulas, which is the primary motivation
for the W-transformed framework, the traditional tail dependence coefficients $\lambda_{\text{l}}$
and $\lambda_{\text{u}}$ are of limited usefulness as they are limits based on copula values along
the copula diagonal, not directions of maximal tail concordance. The MTCM introduced in this
subsection provides a more appropriate tool for this setting, and Proposition~\ref{prop:flip:v:trans:tail}
shows explicitly how the MTCM of a flipped v-transformed copula relates to that of the base copula.
Second, even in the homogeneous case, the reduction does not limit the utility of the pure
W-transform. One may simply start from a base copula $C$ that is already heavily tail-dependent
and apply a W-transform to dilute this strong dependence into a more realistic, asymmetric pattern.
We view this as an interesting application for copulas with high concordance that have so far seen
limited practical use. Third, for settings requiring explicit control over tail dependence alongside
asymmetry, the ordinal sum construction provides the appropriate tool, where the component copulas
supply the tail dependence while the W-transform mixes the blocks to introduce asymmetry. The two
constructions are thus complementary: the pure W-transform dilutes and warps an existing strongly
dependent base, while the W-transformed ordinal sum assembles a new dependence structure with
user-specified tail properties. A similar perspective applies to global concordance measures, as discussed
in the next subsection.

\subsection{Concordance measures}\label{subsec:concor}
The concordance measures Spearman's rho $\rho_{\text{S}}$ and Kendall's tau $\tau$ for a
bivariate copula $C$ are
$\rho_{\text{S}}(C)=12\iint_{[0, 1]^2}C(u_1,u_2)\,\rd u_1\rd u_2-3$ and
$\tau(C)=4\iint_{[0, 1]^2} C(u_1,u_2)\,\rd C(u_1, u_2)-1$, respectively; see, for example,
\cite[Chapter~10]{jaworskidurantehaerdlerychlik2010}. For W-transformed
copulas $C_{\bm{\W}}$, these measures can be written as
\begin{align*}
  \rho_{\text{S}}(C_{\bm{\W}})&=12\sum_{k_2=1}^{K_2}\sum_{k_1=1}^{K_1}\int_0^1\int_0^1
                                \Delta_{B_{\bm{\delta}_{\bm{k}},\bm{\W^{-1}(\bm{u})},\bm{I}}}C
                                \,\rd u_1\rd u_2-3,\\
  \tau(C_{\bm{\W}})&=4\sum_{k_2=1}^{K_2}\sum_{k_1=1}^{K_1}\iint_{[0,1]^2}
                     \Delta_{B_{\bm{\delta}_{\bm{k}},\bm{\W^{-1}(\bm{u})},\bm{I}}}C\,\rd C(u_1, u_2)-1,
\end{align*}
where
\begin{align*}
  B_{\bm{\delta}_{\bm{k}},\bm{\W^{-1}},\bm{I}}
  =\prod_{j=1}^2\mathcal{I}_j^{(k_j)},\quad \mathcal{I}_j^{(k_j)}=\begin{cases}
    (\delta_{j, k_j-1}, \W^{-1}_{j|k_j}(u_j)], & k_j\in I_j,\\
    (\W^{-1}_{j|k_j}(u_j),\delta_{j, k_j}], & k_j\notin I_j.
  \end{cases}
\end{align*}
While there is little hope of getting a closed-form formula for these measures
even for piecewise linear v-transforms, these measures admit an interpretable
decomposition. Specifically, Spearman's rho can be viewed as a mixture of local
Spearman's rho values over all rectilinear regions of the form
$[\delta_{1, k_1}, \delta_{1, k_1+1}]\times [\delta_{2, k_2}, \delta_{2,
  k_2+1}],$ $k_1\in\{1,\dots,K_1\}$, $k_2\in\{1, \dots, K_2\}$.  Consider a
Cauchy copula with correlation parameter $\rho=0$, so an uncorrelated Student's $t$ copula with
$\nu=1$ degrees of freedom. Applying the W-transform
$\W(u)=|2u-1|$, $u\in[0, 1]$, homogeneously to both margins yields a
transformed copula $C_{\bm{\W}}$ with $\rho_{\text{S}} \approx 0.47$;
for samples, see Figure~\ref{fig:cauchy:cop}.
\begin{figure}[htbp]
  \centering
  \includegraphics[width=0.48\textwidth]{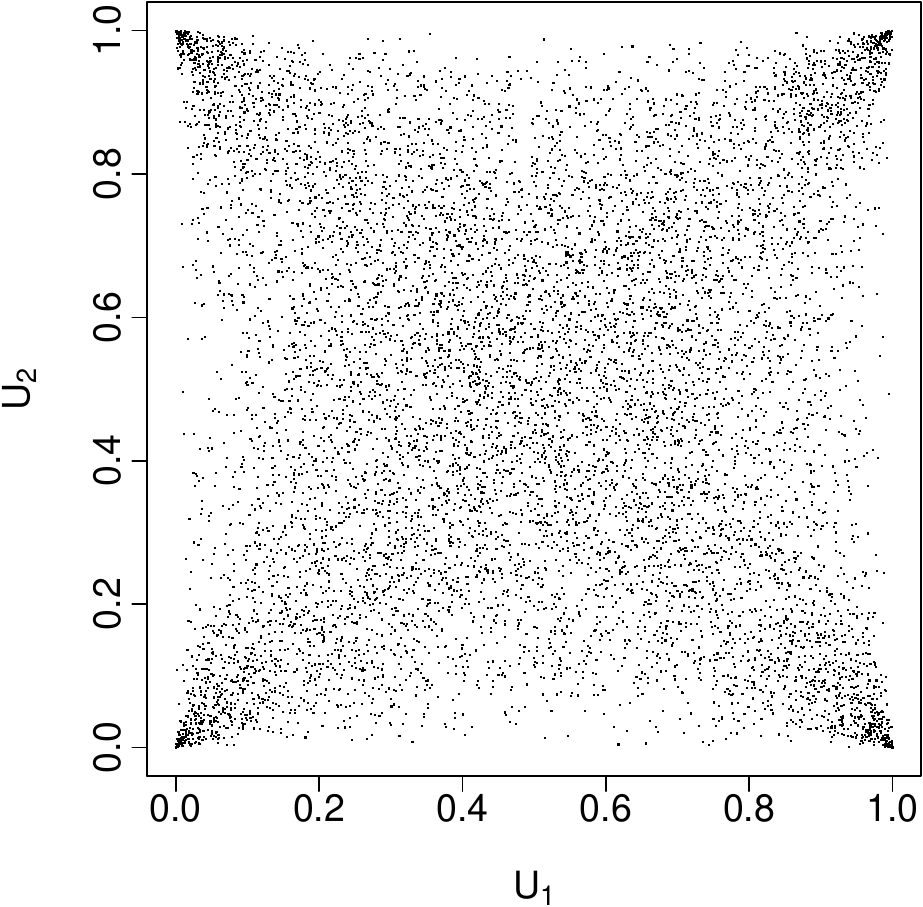}
  \hfill
  \includegraphics[width=0.48\textwidth]{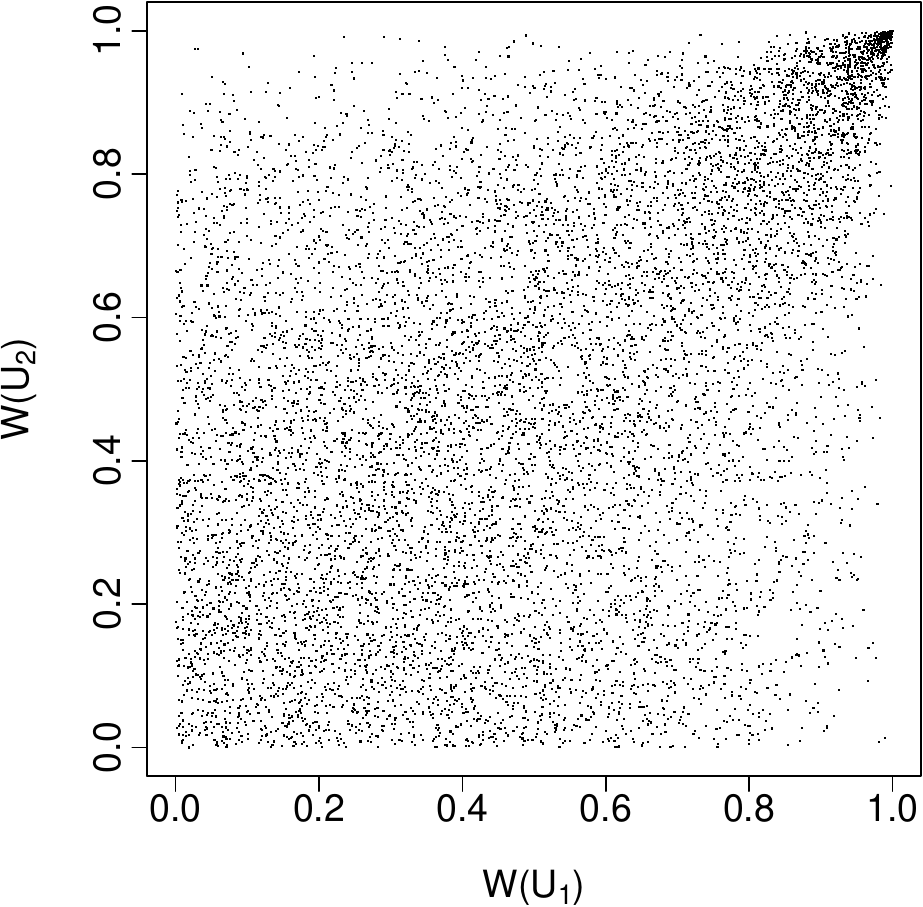}
  \caption{A sample from the radially symmetric Cauchy copula with Spearman's
    rho $\rho_{\text{S}}=0$ (left), and a sample of the same size generated from the
    W-transformed Cauchy copula with $\rho_{\text{S}} \approx 0.47$ (right).}\label{fig:cauchy:cop}
\end{figure}
This increase in $\rho_{\text{S}}$ appears because within each grid cell
($[0,0.5]^2$, $[0,0.5]\times[0.5, 1]$, etc.), the local Spearman's rho is
non-zero, and the W-transform cumulatively integrates these local dependencies
into (the global) Spearman's rho. However, most copula families do not have
significant non-zero local correlation across all grid cells. Thus,
W-transformed copulas typically exhibit lower concordance than their underlying base
copulas, except in special cases where all local dependencies are non-trivial.

\subsection{Symmetries}\label{subsec:symmetry}
The impact of W-transforms on distributional symmetries of $\bm{U}\sim C$
depends on whether the same or different transforms are applied across the $d$
margins. This section studies symmetries of homogeneous W-transformed copulas
$C_{\W}$, characterising which symmetries of $C$ are preserved by $C_{\W}$. The
complementary scenario, where different W-transforms $\W_1,\dots,\W_d$ induce
asymmetric dependence, is explored in Section~\ref{subsec:asym:tail} later.

A $d$-dimensional copula is \emph{exchangeable} if, for any permutation $\sigma$
of the indices $\{1,\dots,d\}$, one has
$C(u_{\sigma(1)}, \dots, u_{\sigma(d)})=C(u_1, \dots, u_d)$ for all
$u_1, \dots, u_d\in[0, 1]$; examples of exchangeable copulas are Archimedean and
homogeneous elliptical copulas. Our first result establishes that W-transforms,
when applied homogeneously to a copula-distributed random vector, preserve exchangeability.
\begin{proposition}[Exchangeability]\label{prop:exchan}
  Let $C$ be an exchangeable copula. Then the homogeneous $C_{\W}$ is also exchangeable.
\end{proposition}
The converse of Proposition~\ref{prop:exchan} is not true in general, that is given exchangeable
homogeneous $C_{\W}$, $C$ may not be exchangeable. To see this, consider the maltese copula
\begin{align}
  C(u_1, u_2) = \begin{cases}
    \max\{0, 4u_1u_2-3u_2\}, & u_2\leq \frac{1}{4},\\
    \min\{\frac{4}{3}u_1u_2-\frac{1}{3}u_1,u_2-\frac{1}{4}\}+\max\{0,u_1-\frac{3}{4}\},
                             & u_2>\frac{1}{4},\\
  \end{cases}\label{eq:non:ex:copula}
\end{align}
which puts mass uniformly on the rectangles $[0, 3/4]\times[1/4,1]$ and $[3/4,1]\times[0,1/4]$.
Clearly, $C$ is not exchangeable as $C(1/3, 1/2)=1/9\neq1/12=C(1/2, 1/3)$, but for the W-transform
$\W(u)=\begin{cases}
  -4u+1, & u\le\frac{1}{4},\\
  \frac{4}{3}x-\frac{1}{3}, & u>\frac{1}{4},
\end{cases}$
the homogeneous W-transformed copula is $C_{\W}(u_1, u_2)=u_1u_2$ which is the independence copula and
thus exchangeable.
Let us now turn to radial symmetry. A bivariate copula $C$ is \emph{radially
  symmetric} if $C(u_1,u_2)=-1+u_1+u_2+C(1-u_1,1-u_2)$, $u_1,u_2\in[0, 1]$, or equivalently, if
its survival copula $\hat{C}$ equals $C$. Radial symmetry implies that for any
$(u_1,u_2)\in [0,1]^2$, the $C$-volume of the rectangle $(0, u_1]\times(0,u_2]$ is the
same as that of its radially opposite counterpart $(1-u_1, 1]\times (1-u_2,
1]$. However, radial symmetry is not preserved under arbitrary W-transforms. To
see this, take a Cauchy copula with correlation parameter $\rho=0$ and the
W-transform $\W(u)=|2u-1|$, $u\in[0, 1]$. Then the homogeneous transformed copula
$C_{\W}$ is not radially symmetric anymore in general; see the right-hand
side of Figure~\ref{fig:cauchy:cop}.

The symmetric linear W-transform $\W(u)=|2u-1|$, $u\in[0,1]$ redistributes tail mass
asymmetrically. It collapses both tail masses into the upper-right tail, which violates radial symmetry.
On the other hand, radially symmetric W-transformed copulas may arise from non-radially
symmetric copulas. To see this, consider the copula \eqref{eq:non:ex:copula}, which
is not radially symmetric (since $\Delta_{(0,3/4]\times (0,1/4]}C=0$, but
$\Delta_{(1/4, 1]\times(3/4, 1]}C=1/6$), however, its W-transformed copula
$C_{\W}$ is the independence copula, which is clearly radially symmetric.

\section{Applications}\label{sec:application}
In this section, we demonstrate the practical use of W-transforms by applying them to specific
copulas $C$. We consider three key scenarios in the next three sections:
\begin{enumerate}
\item Removing tail dependence in one tail of a copula $C$, while retaining the tail dependence in the other tail of $C$.
\item Creating an asymmetric $C_{\bm{\W}}$ by applying different W-transforms $\W_1, \W_2$
  to the margins of $(U_1,U_2)\sim C$ of a symmetric copula $C$.
\item Constructing new copulas $C_{\bm{\W}}$ using ordinal sums which are connected to
  mixtures of copulas.
\end{enumerate}
In Section~\ref{subsec:illu:real:data} we then consider an application of
W-transformed copulas to a real-life dataset.

\subsection{Removal of tail dependence in one tail}\label{subsec:remove:tail}
In applications, it is often desirable to model dependence structures that
exhibit asymmetric tail behaviour. Empirical studies have highlighted this need,
for example \cite{garcia2011dependence} found strong extremal tail dependence
across countries %
in equity and bond markets; \cite{chollete2011international} presented evidence
in data of extreme asymmetry of tail dependence where one tail has significant
dependence; %
\cite{hautsch2015financial} and
\cite{tobias2016covar} measured risk spillovers via value-at-risk, focusing
solely on downside dependence. %
In this section, we demonstrate how W-transforms can be used to selectively
remove tail dependence of one tail of a copula $C$, while preserving the other,
in order to model asymmetric dependencies.

Consider the W-transform
\begin{align}\label{eq:W:id:piece}
  \W(u)=\begin{cases}
    \dfrac{9}{10} - \dfrac{3\sqrt{5u}}{5}, & u\in[0, 0.45],\\[6pt]
    \dfrac{3\sqrt{20u-9}}{10}, & u\in(0.45, 0.9],\\[6pt]
    u, & u\in(0.9, 1].
  \end{cases}
\end{align}
Applied to both margins of a bivariate $t$-copula $C_{\nu, \rho}$ with $\nu=2$
degrees of freedom and correlation parameter $\rho=0.9$, this W-transform
retains the upper tail dependence while removing the lower one. Plots of the
W-transform~\eqref{eq:W:id:piece}, a simulated sample of size 2000 from
$C_{\nu, \rho}$ and a sample of the same size from the corresponding
homogeneous W-transformed copula $C_{\W}$ are shown in Figure~\ref{fig:id:W:t:cop}.
\begin{figure}[htbp]
  \includegraphics[width=0.32\textwidth]{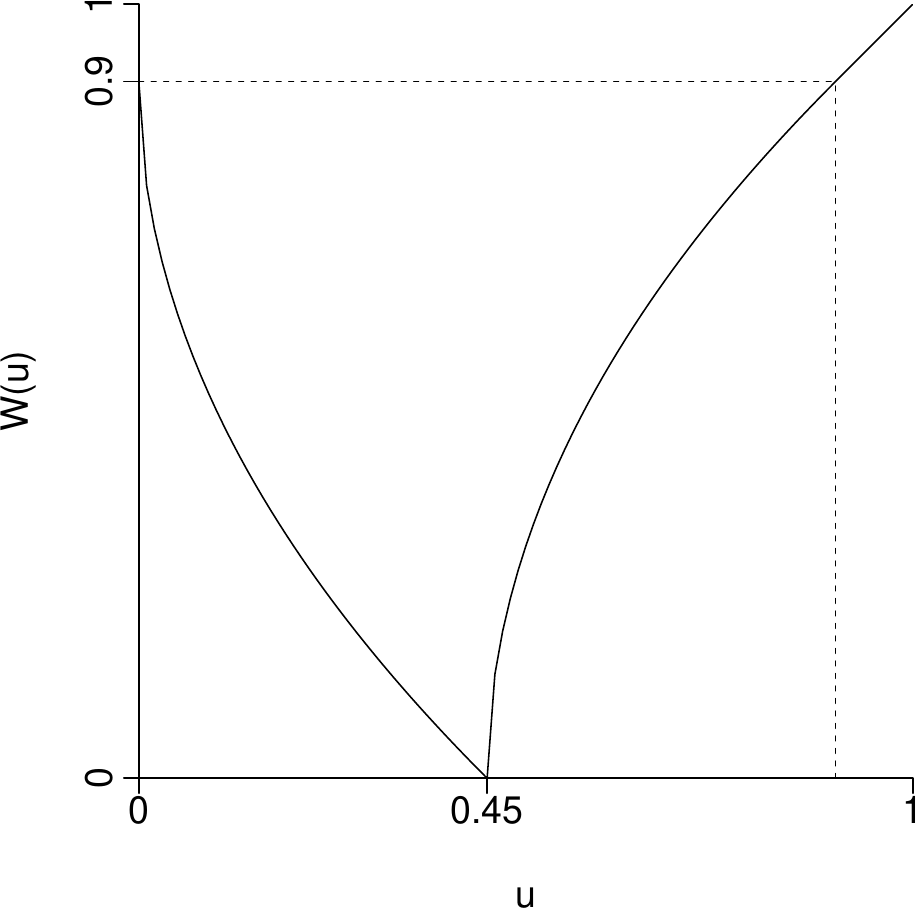}
  \hfill
  \includegraphics[width=0.32\textwidth]{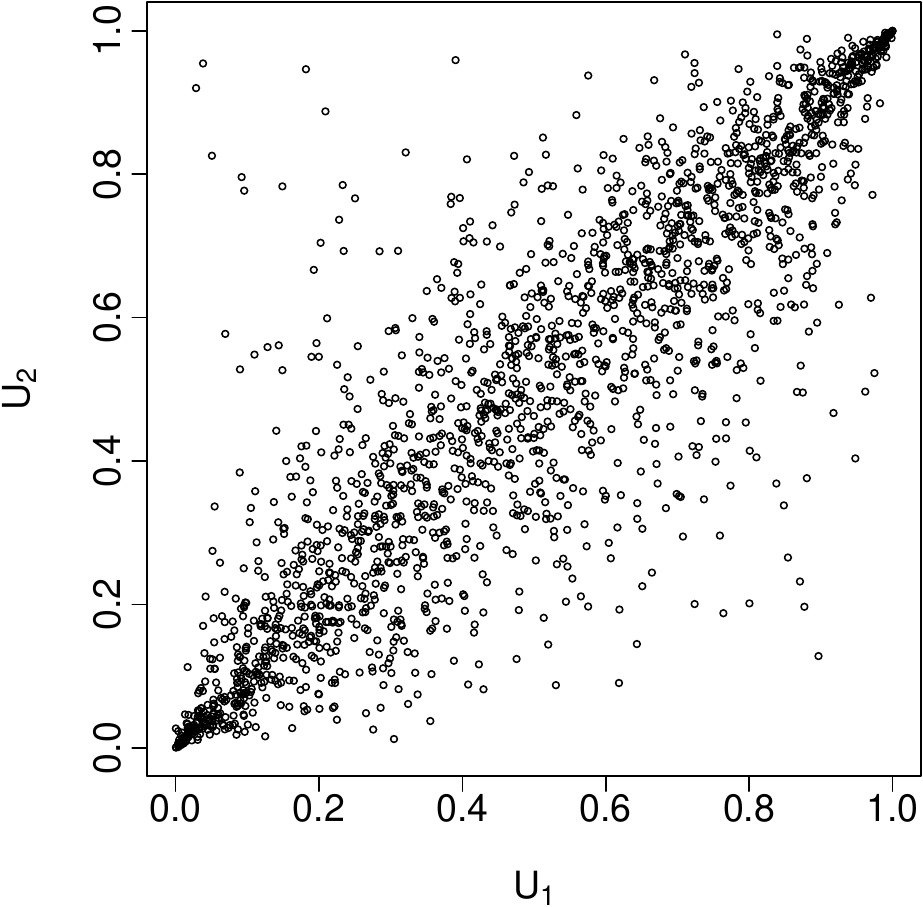}
  \hfill
  \includegraphics[width=0.32\textwidth]{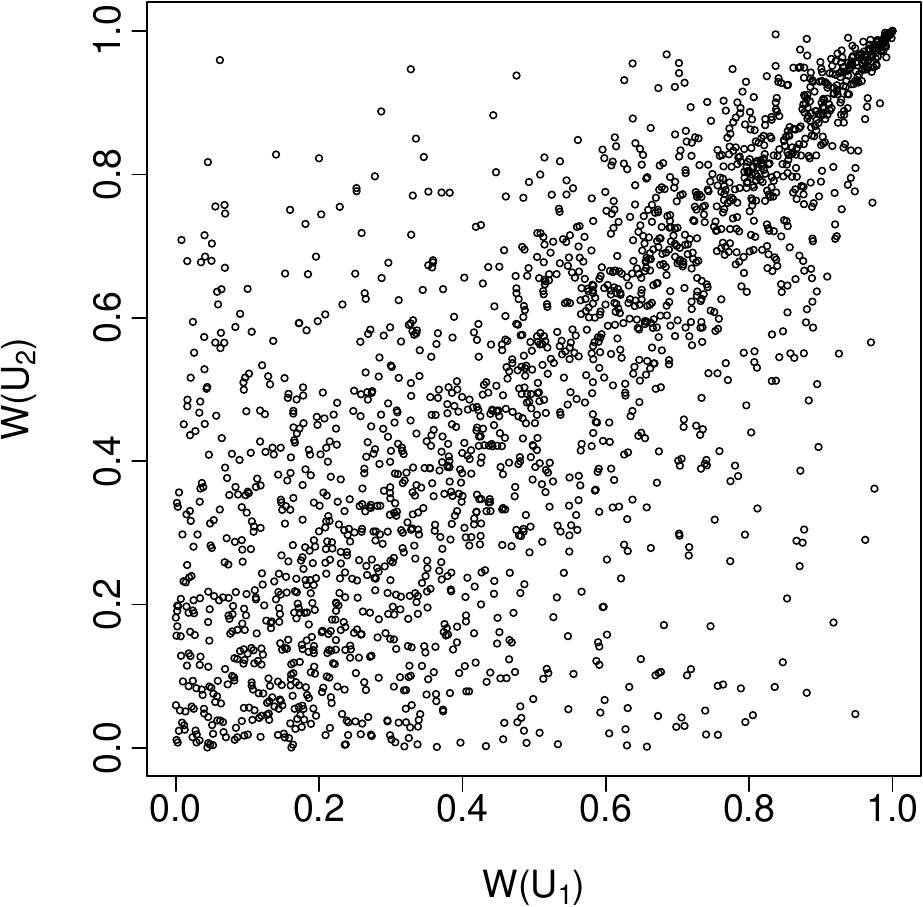}
  \caption{W-transform~\eqref{eq:W:id:piece} (left), a sample of size 2000
    from the $t$-copula $C_{\nu=2,\rho=0.9}$ (centre) and a sample of the same
    size from the corresponding homogeneous W-transformed copula $C_{\W}$ (right).}
  \label{fig:id:W:t:cop}
\end{figure}
Clearly, the identity piece $\W(u) = u$, $u\in(0.9, 1]$, preserves the
upper tail clustering.  However, since $|\W^\prime(u)|\ge 1$,
$u\in[0, 0.45)$, $\W$ ``drags'' the mass clustered in the lower tail
outward and re-distributes it over $[0, 0.9]^2$. On the other hand, samples near
0.45 (lacking co-movement) are mapped to the lower tail, eliminating lower tail
dependence. Note that the flipped W-transform $1-\W(u)$ would retain
lower tail dependence while removing the upper one. This highlights the
flexibility of W-transforms for modifying tails. Clearly, an application to
higher dimensions can easily be constructed.

\subsection{Creating asymmetry}\label{subsec:asym:tail}
Many copula families, such as the aforementioned homogeneous elliptical copulas
and Archimedean copulas, are exchangeable. However, in practice, rarely do we
encounter perfectly symmetric data, which calls for copulas families that can
capture non-exchangeability. For example, \cite{durante2016asymmetric}
considered such copulas and applied them to experimental designs. Or
\cite{mcneilsmith2012} and \cite{kollo2017tail} considered skewed
$t$-copulas. General methods for constructing non-exchangeable copulas include
Khoudraji's device, see \cite{khoudraji1995} and \cite{freesvaldez1998}, its
extensions via $P$-increasing functions in \cite{durante2009construction} and
generalisations of Archimedean copulas in \cite{liebscher2008},
\cite{mcneil2008} and \cite{hofert2010b}.  In this section, we propose a simple
method to break exchangeability of any copula $C$ by applying distinct
W-transforms to each margin of $\bm{U}\sim C$, thus generating intrinsically
asymmetric $C_{\bm{\W}}$.

Consider the W-transform parametrised by $\theta\in(0,0.5)$
\begin{align}\label{eq:W:asym}
  \W_\theta(u) = \begin{cases}
    \dfrac{u}{2\theta}, & u\in[0, \theta],\\[6pt]
    \dfrac{u-\theta}{1-2\theta}, & u\in(\theta, 1-\theta],\\[6pt]
    \dfrac{u-1+2\theta}{2\theta}, & u\in(1-\theta, 1],
  \end{cases}
\end{align}
which is piecewise linear. If one applies such W-transforms for different
parameters $\theta$ to $\bm{U}\sim C$, one naturally expects an asymmetrically
distributed $(\W_{\alpha_1}(U_1),\W_{\alpha_2}(U_2))\sim C_{\bm{\W}}$. As
we shall show, a portion of the tail dependence of $C$ is retained even though
mass is not concentrated on the diagonal anymore. Consider the same $t$-copula
$C_{\nu=2, \rho=0.9}$ as in Section~\ref{subsec:remove:tail} and $\alpha_1=0.3$,
$\alpha_2=0.45$.  Then the W-transformed copula
$C_{(\W_{0.3}, \W_{0.45})}$ is no longer exchangeable, exhibiting
an asymmetric mass distribution about the diagonal in the tails. Plots of
$\W_{0.45}$ and simulated samples from the $t$-copula and its
W-transformed copula are shown in Figure~\ref{fig:asym:W:t:cop}.
\begin{figure}[htbp]
  \includegraphics[width=0.32\textwidth]{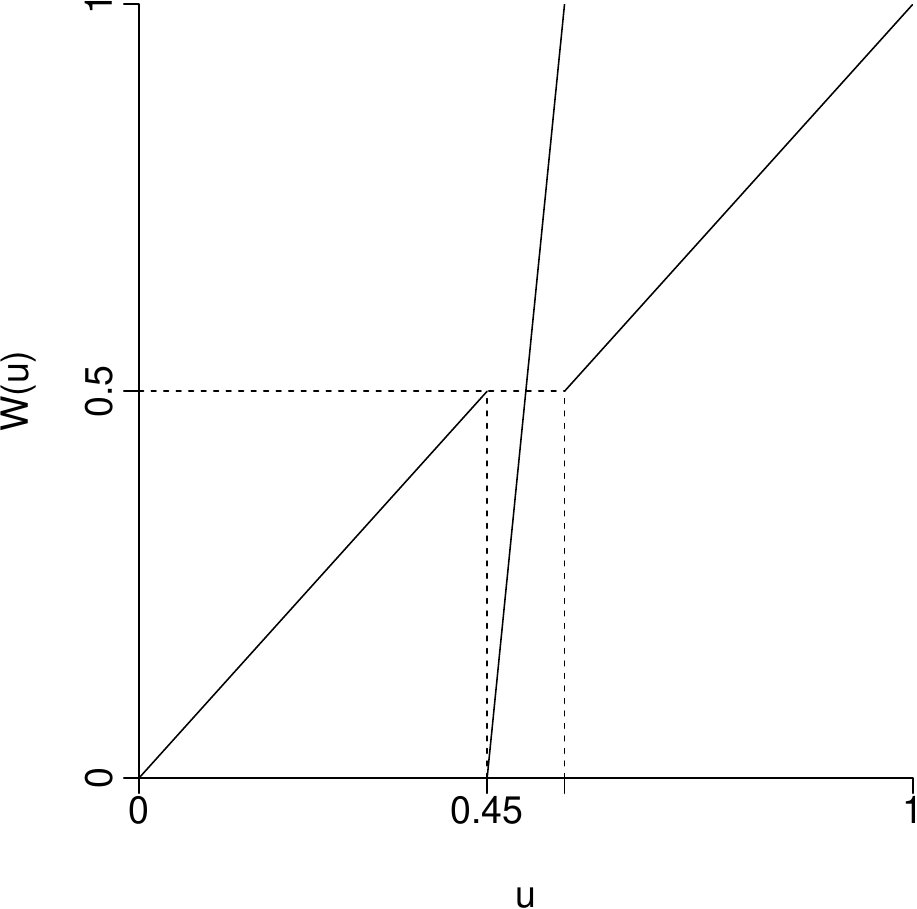}
  \hfill
  \includegraphics[width=0.32\textwidth]{fig_t_copula.pdf}
  \hfill
  \includegraphics[width=0.32\textwidth]{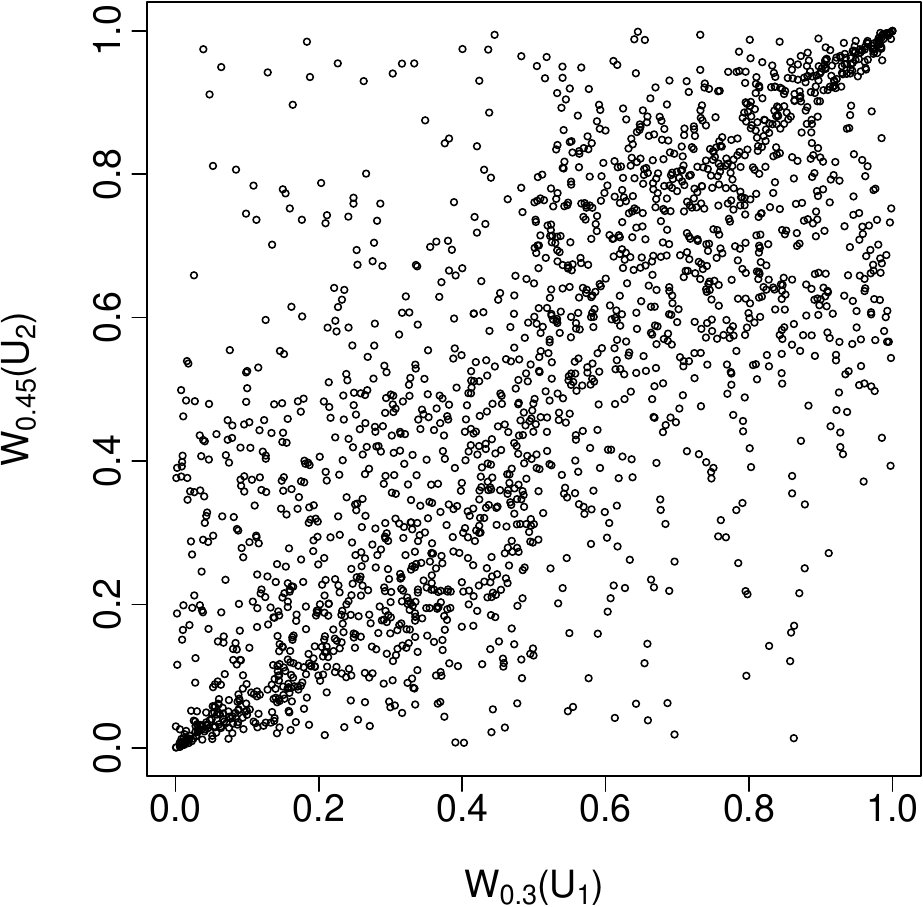}
  \caption{W-transform $\W_{0.45}$~\eqref{eq:W:asym} (left),
    a sample of size 2000 from the $t$-copula $C_{\nu=2, \rho=0.9}$ (centre) and
    a sample of the same size from the corresponding W-transformed copula
    $C_{(\W_{0.3}, \W_{0.45})}$ (right).}
  \label{fig:asym:W:t:cop}
\end{figure}

The two W-transforms $\W_{0.3}, \W_{0.45}$ redistribute the
probability mass over $[0, 0.3]\times[0, 0.45]$ and $(0.7, 1]\times(0.55, 1]$,
which is stretched from a rectangle to a square and redistributed over $[0, 0.5]^2$
and $(0.5, 1]^2$, respectively, hence the asymmetry.  We further validate this
behaviour using a test of exchangeablility proposed by \cite{genest2012tests}
(\texttt{exchTest()} in the \R\ package \texttt{copula}) on the W-transformed copula sample,
which gave a p-value of $0.0005$ and thus evidence against exchangeability.

\subsection{Construction of copulas using ordinal sums}\label{subsec:cop:ordinal:sum}
Section~\ref{subsec:tail} showed that W-transforms typically reduce the tail
dependence of a copula $C$ and that one may construct flexible tails based on ordinal sums. In this
section, we detail this construction further by generalising convex mixtures of copulas via ordinal sums,
targeting tail dependence by strategically choosing a set of copulas in the ordinal sum,
and interpreting this construction as a mixture of copulas.

Consider an ordinal sum $C_S$ as in \eqref{eq:ordinal:sum}. For homogeneous
W-transformed copulas, the W-transform $\W$ has change points
partitioning $[0,1]$ into non-overlapping and non-degenerate intervals
$\Delta_k=(\delta_{k-1},\delta_k]$, and one then scales $C_k$ to the hyperrectangle
$(\delta_{k-1}, \delta_k]^d$, $k=1,\dots,K$. By Section~\ref{subsec:distr}, one
thus expects that $\W$ aggregates all component copulas by precisely adding their
volumes of $(\delta_{k-1}, \delta_k]^d$, $k=1,\dots,K$.
For pssm W-transforms $\W_j$, $j=1,\dots, d$ (see, for example,
\eqref{eq:param:v:trans}) with the same change points $\{\delta_k\}_{k=0}^K$,
define
\begin{align}\label{eq:os:G:df}
  G_{j,k}(u) \coloneqq \begin{cases}
    0, & u=0,\\
    \frac{\W^{-1}_{j|k}(u)-\delta_{k-1}}{\delta_k-\delta_{k-1}} & u\in(0,1),\\
    1, & u=1,
  \end{cases}
\end{align}
and let $g_{j,k}(u)=\frac{\rd }{\rd u} G_{j,k}(u)$ be the almost everywhere existing derivative of $G_{j,k}$.
Let $I_j$ be the index set such that $\W_{j|k}$ is increasing (decreasing) if and only if
$k\in I_j$ ($I_j^C$). Then the \emph{W-transformed ordinal sum} $C_{S,\bm{\W}}$ is
\begin{align}\label{eq:os:W:cop}
  C_{S,\bm{\W}}(\bm{u}) = \sum_{k=1}^K(\delta_{k}-\delta_{k-1})\Delta_{B_k}C, \quad
  B_k=\prod_{j=1}^d \Bigl(\Bigl(\biguplus_{k\in I_j}(0,G_{j,k}(u_j)]\Bigr)\cup\Bigl(
  \biguplus_{k\notin I_j} (G_{j,k}(u_j), 1]\Bigr)\Bigr).
\end{align}
If each $\W_j$ is piecewise increasing, then \eqref{eq:os:W:cop} reduces to
\begin{align}\label{eq:os:in:W:cop}
  C_{S,\bm{\W}}(\bm{u})=\sum_{k=1}^K(\delta_{k}-\delta_{k-1})C_k(G_{1,k}(u_1),\dots,G_{d,k}(u_d)),
\end{align}
that is $C_{S,\bm{\W}}$ reduces to a mixture of the copulas $C_1,\dots,C_K$.

This construction shows two significant improvements over existing models.
First, by Proposition~\ref{prop:os:tail:dependence}, it achieves more flexible
tail dependencies through the convex combinations $\sum\alpha_k\lambda_{\text{l},k}$
and $\sum\beta_k\lambda_{\text{u},k}$ in contrast to the tail
dependence coefficients $\min_k\{\lambda_{\text{l},k}\}$ and $\min_k\{\lambda_{\text{u},k}\}$
of the methods of \cite{khoudraji1995} and \cite{liebscher2008} which only take
into account the smallest $\lambda_{\text{l},k}$ and $\lambda_{\text{u},k}$,
respectively. Second, it maintains intermediate degrees of concordance while
allowing for non-exchangeability through applying different W-transforms to the
margins.

By Sklar's theorem, construction principle \eqref{eq:os:in:W:cop} equivalently
defines a mixture of joint distributions
\begin{align}
  C_{S,\bm{\W}}(\bm{u}) = \sum_{k=1}^K(\delta_{k}-\delta_{k-1})C_k(G_{1,k}(u_1),
  \dots,G_{d,k}(u_d)) = \sum_{k=1}^K(\delta_{k}-\delta_{k-1})F_k(\bm{u}),\label{eq:mix:joint:df2}
\end{align}
subject to
\begin{align}\label{eq:margin:df:constrain}
  \sum_{k=1}^K (\delta_k-\delta_{k-1}) g_{j,k}(u) = 1, \quad u\in[0, 1],\ j=1,\dots,d,
\end{align}
since \eqref{eq:os:G:df} is a strictly increasing distribution function on $[0,1]$.
Hence, one can interpret \eqref{eq:mix:joint:df2} as a construction principle of copulas
by choosing $d(K-1)$ marginal distributions on $[0,1]$ subject to \eqref{eq:margin:df:constrain}.

\begin{remark}[Relationship to \cite{li2014distorted}]
  The construction principle \eqref{eq:os:in:W:cop} is equivalent to the one
  presented in \cite[Equation~(1)]{li2014distorted} which has been named
  ``distorted mixture copula'' (DM copula). \cite{li2014distorted} have shown that
  \eqref{eq:os:in:W:cop} is able to achieve any tail dependence function
  as defined by \cite[Equations~(2.2) and~(2.3)]{joe2010tail} or any tail dependence
  coefficient. In this reference, DM copulas were used to construct a copula $C$
  that is arbitrarily close to a Gaussian copula in terms of absolute mean deviation
  but has lower (upper) tail dependence function identical to that of a Clayton (Gumbel)
  copula. This can be done via~\eqref{eq:os:in:W:cop} using W-transforms,
  and an illustration is provided in Example~\ref{eg:os:tails}~\ref{eg:os:tails:modif:gauss}.
\end{remark}

\subsection{Application to the Danube dataset}\label{subsec:illu:real:data}
As an illustration of the flexibility and usefulness of W-transforms for
statistical modelling, we consider the dataset \texttt{danube} from the \R\
package \texttt{lcopula}. It consists of 659 pseudo-observations of monthly base
flow observations from the Global River Discharge Project of the Oak Ridge
National Laboratory Distributed Active Archive Center, determined from joint
observations over 55 years until 1991 at two stations, one being in
Sch\"arding (Austria) on the Inn and the other one being in Nagymaros (Hungary) on the
Danube.

Upon visual inspection, the pseudo-observations, shown on the left-hand side of
Figure~\ref{fig:danube}, exhibit non-exchangeability; a formal test using the
function \texttt{exchTest()} from the \R\ package \texttt{copula} yields a
p-value of $0.0005$, confirming statistically significant non-exchangeability.
Additionally, the data demonstrate co-movement in the upper tail, making
upper-tail dependent copulas such as the Gumbel, rotated Clayton or Joe suitable
candidate models for the data. As demonstrated by \cite[Section~4,
5]{hofertkojadinovicmaechleryan2018}, a Gumbel copula is not rejected (with a
p-value of $0.07343$) under a parametric bootstrap goodness-of-fit test with the
function \texttt{gofCopula()} of the \R\ package \texttt{copula} using the
inversion of Kendall's tau estimation method.  However, since the Danube data
are inherently non-exchangeable, the exchangeable Gumbel family requires
adjustment. A previous attempt using a Khoudraji–transformed Gumbel copula, as
discussed by \cite[Section 4, 5]{hofertkojadinovicmaechleryan2018}, provided
only weak evidence for this model (yielding a p-value of $0.04745$). In this
study, we build on this %
attempt to improve a statistically sound fit for the Danube data.

We begin by fitting a one-parameter Gumbel copula to the Danube data via
maximum pseudo-likelihood estimation. The parameter estimate is $2.1383$, with a
log-likelihood of $278.148$. A simulated sample from this fitted copula is shown
on the left-hand side of Figure~\ref{fig:danube:fits}. To account for the observed
\begin{figure}[htbp]
  \centering
  \includegraphics[width=0.48\textwidth]{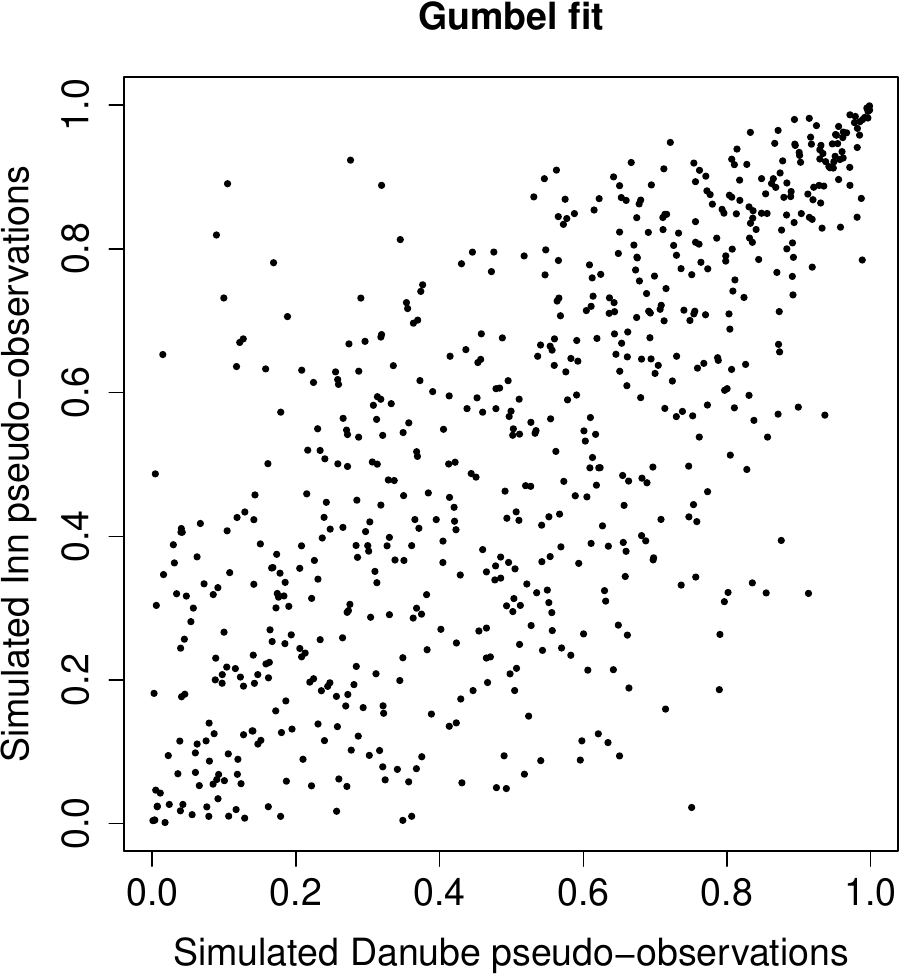}
  \hfill
  \includegraphics[width=0.48\textwidth]{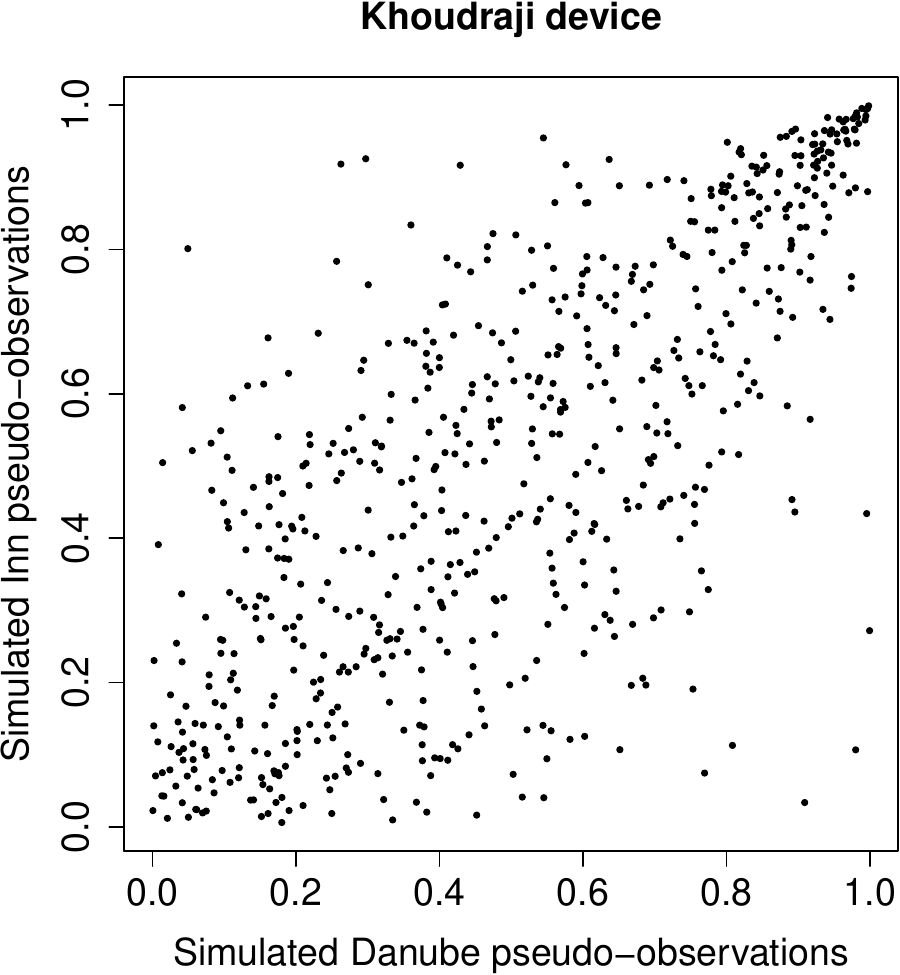}
  \caption{Simulated sample of size 659 of the fitted Gumbel (left) and of the
    fitted Khoudraji-transformed Gumbel copula (right).}
  \label{fig:danube:fits}
\end{figure}
non-exchangeability, we next consider an ordinal sum copula as
in~\eqref{eq:ordinal:sum}, with $\delta = 0.5$, two Gumbel copulas with unknown
parameters $\alpha_1, \alpha_2$, a piecewise linear W-transform
$\W(u)=2u-\lceil 2u -1\rceil$, $u\in[0, 1]$, applied to the first
margin (Danube), and a parametric W-transform $\W_{\theta}$ applied to
the second margin (Inn), given by
\begin{align}
  \W_\theta(u)=\begin{cases}
    \dfrac{\sqrt{\theta u + 1}-1}{D}, & u\in[0,0.5],\\
    \dfrac{\theta-2D-\sqrt{\theta^2-4\theta D+4D^2+2\theta D^2-4\theta D^2u}}{2D^2}, &u\in(0.5, 1],
  \end{cases}\label{eq:W:trafo:Inn}
\end{align}
where $D=\sqrt{0.5\theta+1}-1$ and $\theta\in(0,\infty)$. A plot of
$\W_{\theta}(u)$ for $\theta = 20$ is shown on the left of
Figure~\ref{fig:danube:diagostics}.
\begin{figure}[htbp]
  \centering
  \includegraphics[width=0.32\textwidth]{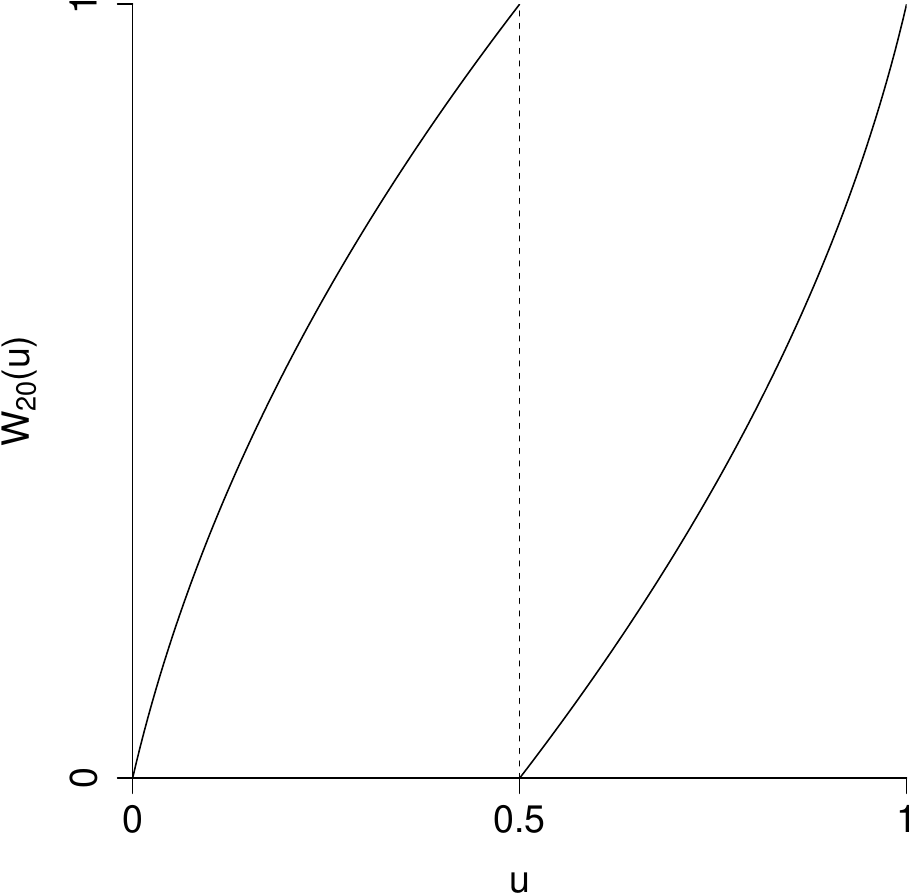}\hfill
  \includegraphics[width=0.32\textwidth]{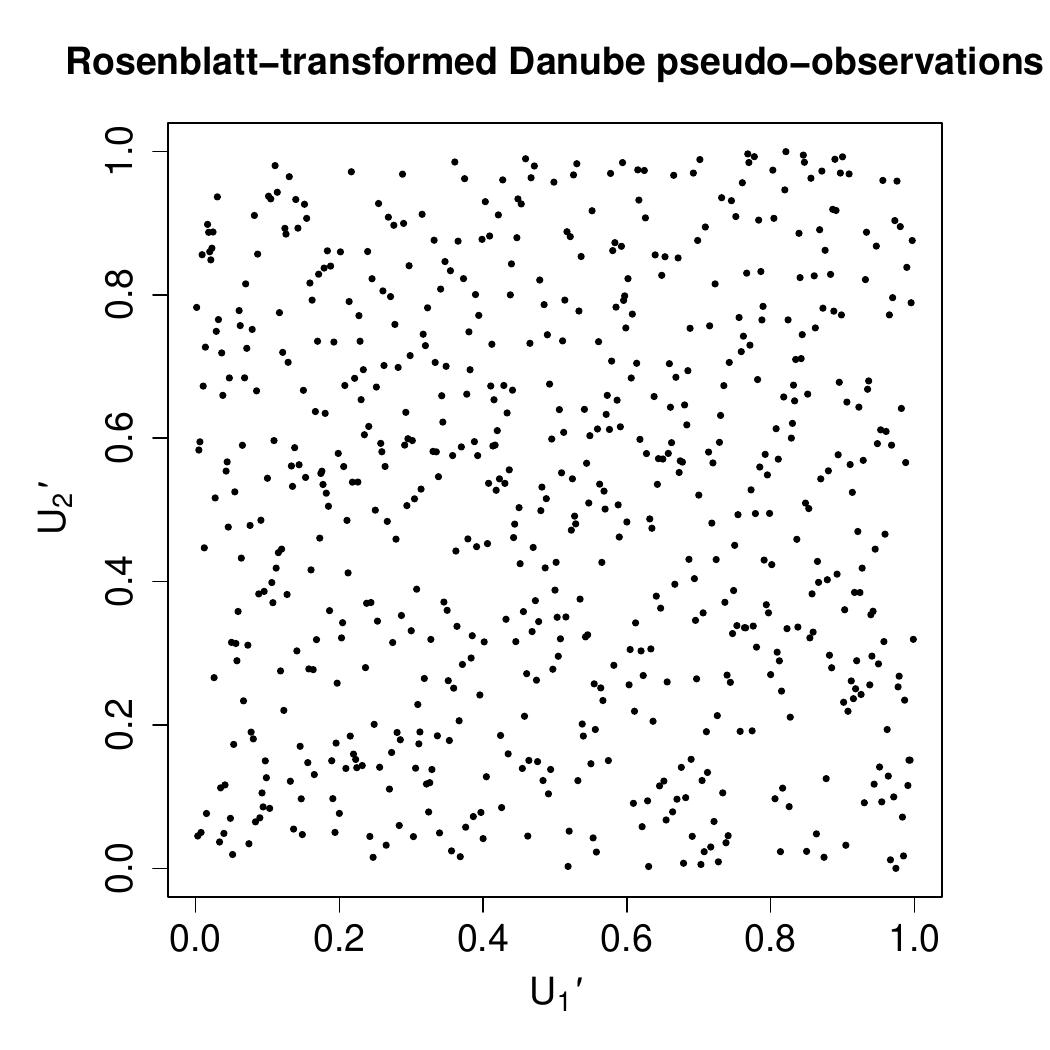}\hfill
  \includegraphics[width=0.32\textwidth]{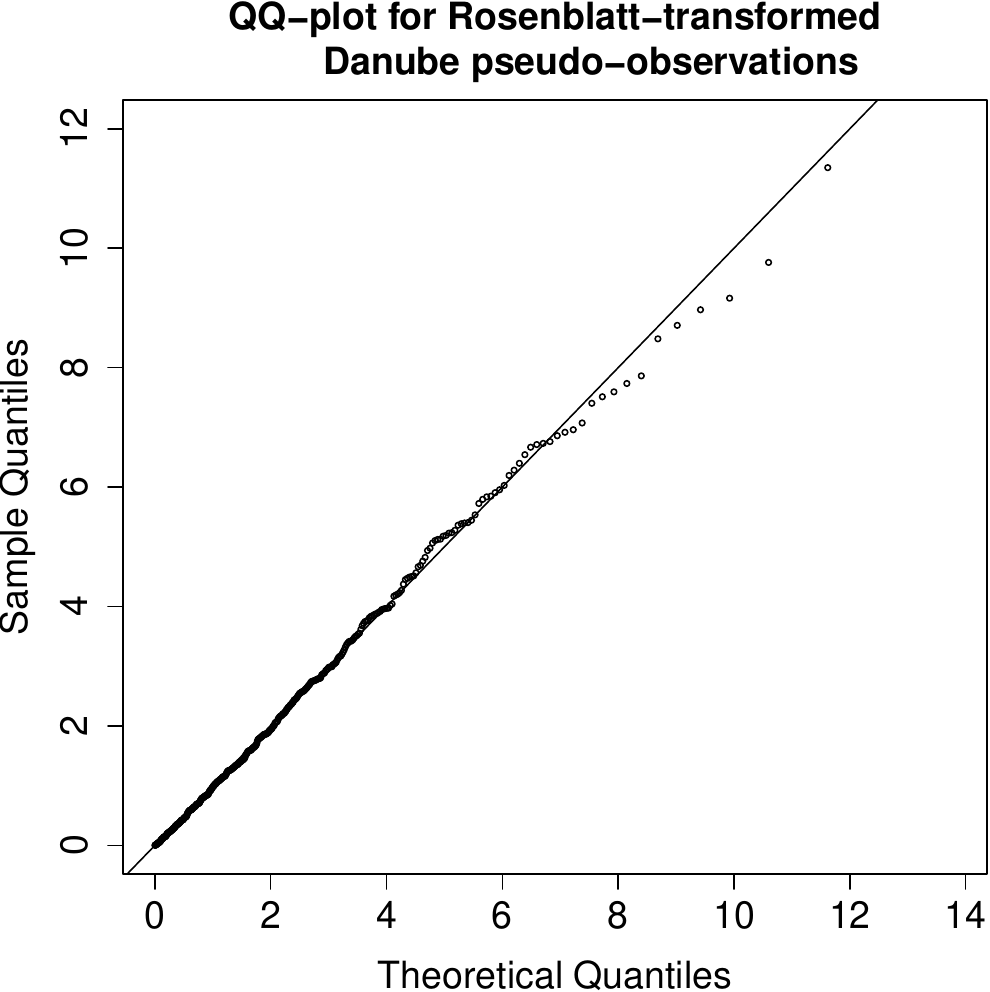}
  \caption{W-transform $\W_{20}$~\eqref{eq:W:trafo:Inn} (left), implied Rosenblatt-transformed
    Danube data $(U_1',U_2')$ (centre) and Q-Q plot of empirical quantiles of
    $(\Phi^{-1}(U_1'))^2 + (\Phi^{-1}(U_2'))^2$ against the theoretical $\chi^2_2$
    quantiles (right).}
  \label{fig:danube:diagostics}
\end{figure}
Our choice of these W-transforms is motivated by two key observations. First,
the Sch\"arding station is located upstream of Nagymaros, leading to generally
higher monthly average flow rates at Sch\"arding. This is reflected in the concave
structure of the data (see the top-left of Figure~\ref{fig:danube:fits}), which
we aim to capture through the concave shape of the first piece of
$\W_\theta$ (see the left of
Figure~\ref{fig:danube:diagostics}). Second, to facilitate a comparison with the
Gumbel copula model used by \cite{hofertkojadinovicmaechleryan2018}, our
W-transformed ordinal sum should extend this special case. Notably, when
$\theta\rightarrow 0$, $\W_\theta$ converges to the piecewise linear
W-transform $\W_{0}(u) = 2u-\lceil 2u -1\rceil$. Hence, a homogeneous
application of this transformation to both margins of the ordinal sum copula
with two equal Gumbel copula components recovers a Gumbel copula and we thus
indeed generalise the latter.

We estimate the parameters of the W-transformed ordinal sum via maximum
likelihood, obtaining $\alpha_1= 2.8437$, $\alpha_2 = 2.0412$ and $\theta = 21.2635$, with
a log-likelihood of $284.319$. A likelihood ratio test with respect to the
Gumbel model yields a p-value of $0.0021$, indicating that the W-transformed
ordinal sum provides a statistically significant improvement over the Gumbel
model. A simulated sample from the fitted W-transformed ordinal sum is displayed
on the right-hand side of Figure~\ref{fig:danube}. For comparison, we also fit
the Khoudraji–transformed Gumbel copula as described in
\cite{hofertkojadinovicmaechleryan2018}, which gives a log-likelihood of
$281.902$. A corresponding sample from this fitted copula is shown on the
right-hand side of Figure~\ref{fig:danube:fits}.

Visual inspection of the samples already indicates the superiority of the
W-transformed ordinal sum copula compared to both the Gumbel copula and the
Khoudraji–transformed Gumbel copula. To further validate our model, we perform
two more visual diagnostics. First, we apply the Rosenblatt transform of the
fitted W-transformed ordinal sum on the data (realisations of $(U_1, U_2)$) and the
resulting transformed data (realisations of $(U_1', U_2')$; see centre of
Figure~\ref{fig:danube:diagostics}) exhibit no visible departure from
independence. Second, for the Rosenblatt-transformed data (realisations of
$(U_1', U_2')$), we compute the realisations of
$(\Phi^{-1}(U_1'))^2 + (\Phi^{-1}(U_2'))^2$ and plot their empirical quantiles
against the theoretical $\chi^2_2$ quantiles in the form of a Q-Q plot (see the
right of Figure~\ref{fig:danube:diagostics}). The close alignment indicates no
departure from the identity, thus confirming our model's ability to capture the
dependence structure of the Danube data.

We complement these visual diagnostics with a formal parametric bootstrap
goodness-of-fit test (via the function \code{gofCopula()} of the \R\ package \code{copula} of \cite{copula}) %
for all three models using maximum pseudo-likelihood
estimation. The three tests yield p-values of $0.02048$ for the Gumbel copula,
$0.04745$ for the Khoudraji-transformed Gumbel copula, and $0.1013$ for our
W-transformed ordinal sum. These results provide numerical evidence that our
proposed model outperforms the alternatives in terms of goodness-of-fit.

\section{Conclusion}\label{sec:concl}
We introduced W-transforms, a class of
transformations %
constructed from a distribution $F_X$ and a piecewise strictly monotone
function $T$, and studied their properties for continuous and discontinuous
$F_X$. Specifically, W-transforms constructed from continuous $F_X$ are
piecewise strictly monotone, uniformity-preserving, invariant under compositions
and satisfy the partition of square property. When $F_X$ is not continuous, we
extended the definition of W-transforms and showed that the resulting generalised W-transforms
always have linear pieces.

By applying W-transforms componentwise to a copula-distributed random vector, we
derived the corresponding W-transformed copulas and analysed their functional
form, density, tail dependence, concordance measures, and symmetries. We
demonstrated the flexibility and adaptability of W-transforms by showcasing
their ability to produce diverse tail behaviour, to modify copula tails, to
create asymmetric dependencies (in particular, not restricted to
exchangeability), and to lead to flexible copulas based on ordinal
sums. Specifically, we used W-transforms to remove the tail dependence of given
copulas in one tail, to create asymmetric copula models by applying different
W-transforms to a copula-distributed random vector componentwise and constructed
models with flexible tails based on ordinal sums. The resulting models showed
realistic sample clouds as often seen in dependent data. In an empirical
application of W-transforms to the Danube dataset, our suggested W-transformed
ordinal sum copulas outperformed existing models, underscoring the usefulness
and potential of W-transforms for real-life stochastic modelling.

Competing W-transform specifications can be compared using
standard copula goodness-of-fit tests. The problem of a data-driven selection
procedure for the number and placement of change
points in the absence of domain-specific guidance remains open and is
a natural direction of future research.

\appendix

\section{Proofs}
\subsection{Proofs of Section~\ref{sec:cont:W:trafo}}

\begin{proof}[Proof of Proposition~\ref{prop:w:unif:preserv}]
  If $U\sim\U(0,1)$, the quantile transform implies that $F_X^{-1}(U)\deq X\sim F_X$ for any
  $X\sim F_X$. Due to local strict monotonicity of $T$, $X$ being continuously distributed
  implies that $T(X)$ is continuously distributed, so $F_{T(X)}^{-1}$ is strictly increasing by
  \cite{embrechtshofert2013c}.
  Hence,
  \begin{align*}
    \P(\W(U) \leq u) %
    &= \P\bigl(F_{T(X)}(T(X)) \leq u\bigr)
      = \P\bigl(F_{T(X)}^{-1}(F_{T(X)}(T(X))) \leq F_{T(X)}^{-1}(u)\bigr)\\
		&= \P(T(X)\leq F_{T(X)}^{-1}(u)) = F_{T(X)}(F_{T(X)}^{-1}(u))=u,
	\end{align*}
	where the last equality follows from \cite{embrechtshofert2013c}. Hence $\W(U)\sim
	\U(0,1)$. \qedhere
\end{proof}

\begin{proof}[Proof of Proposition~\ref{prop:W:property}]
  \mbox{} %
  \begin{enumerate}
  \item By \eqref{eq:W:transform}, since $F_{T(X)}$ is strictly increasing on
    $\text{ran}(T)$, the change points of $F_{T(X)}(T(x))$ are $t_k,
    k\in\IN$. By continuity of $F_X$, $F_X^{-1} (F_X(u))=F_X(F_X^{-1}(u))=u$ by
    \cite{embrechtshofert2013c}.  Therefore, the change points $\delta_k$ of
    $\W$ are such that $F_X^{-1}(\delta_k)=t_k$, that is
    $\delta_k=F_X(t_k), k\in\IN$. Since $\inf\supp(F_X)=t_0$ and
    $\sup\supp(F_X)=t_K$, $\delta_0=F_X(t_0)=0$ and $\delta_K=F_X(t_K)=1$.
  \item By continuity of $F_X$, $F^{-1}_X$ is strictly increasing. Since $F_{T(X)}$ is strictly
    increasing on $\text{ran}(T)$, the monotonicity of $\W$ depends on $T$ only. Then, by
    \ref{prop:W:property:change:pt}, $\W$ has the same monotonicity on $(F_X(t_{k-1}),
    F_X(t_k)] = (\delta_{k-1}, \delta_k]$ as $T$ has on $(t_{k-1}, t_k]$. Moreover,
    if $T$ is continuous everywhere, then $F_{T(X)}$ is continuous everywhere. As a composition
    of $F_{T(X)}, T$ and $F^{-1}_X$, we obtain that $\W$ is continuous everywhere.
  \item  By \ref{prop:W:property:change:pt} and \ref{prop:W:property:monotone}, $\W$ is pcsm.
    Let $U\sim \U(0,1)$ and $V=\W(U)$. Consider the event $\{V\leq v\}$. By construction,
    this event is equivalent to $\biguplus_{k=1}^{K} S_k(v)$ up to the singleton $\{0\}$ which is a null set.
    As $\W$ is uniformity-preserving, $V\sim\U(0,1)$, and hence \eqref{eq:partition:sq:property} holds.
    \qedhere
  \end{enumerate}
\end{proof}

\begin{proof}[Proof of Proposition~\ref{prop:W:compo}]
  Uniformity-preservation follows since $\W'$ and $\W''$ preserve uniformity, and thus
  their composition inherits this property. Consider a monotone piece $\W''_{|\ell}$ of $\W''$
  with image $I_\ell\coloneqq\W''((\delta_{\ell-1}'', \delta_\ell''])$, $\ell\in\{1, \dots, K''\}$.
  Since $\W'$ is pcsm with change points $\{\delta_k'\}_{k=0}^{K'}$, its restriction to $I_\ell$ is
  strictly monotone on each interval $(\delta_{k-1}', \delta_k']\cap I_\ell$. The preimages
  $\W''^{\,-1}(\delta_k')$ partition $(\delta_{\ell-1}'', \delta_\ell'')$ into subintervals where
  $\W'\circ\W''$ is strictly monotone. As $K', K''\in\bar{\IN}$, the total partition is countable.
\end{proof}

\begin{proof}[Proof of Proposition~\ref{prop:suff:con:lin}]
  \mbox{} %
  \begin{enumerate}
  \item By definition, $\W(u)=F_{T(X)}\bigl(T(F_X^{-1}(u))\bigr)=\P\bigl(T(X)\leq
    T(F_X^{-1}(u))\bigr)$. Let $T_{|k}=T|_{(t_{k-1}, t_k]}$ be the restriction of $T$ on the
    $k$th piece. Suppose that $F_X^{-1}(u)\in(t_{k-1}, t_k]$ for some $\ell$ and that
    $T_{|\ell}$ is increasing. Then, by the law of total probability, we have
    \begin{align*}
      \W(u)&=\sum_{k=1}^K \P\bigl(T_{|k}(X)\leq T_{|\ell}(F_X^{-1}(u)),
             X\in(t_{k-1}, t_k]\bigr).
    \end{align*}
    Since $T$ is injective except possibly at the change points $t_0, \dots, t_K$, and since $F_X$ is
    continuous, the joint probability $\P\bigl(T_{|k}(X)\leq T_{|\ell}(F_X^{-1}
    (u)), X\in(t_{k-1}, t_k]\bigr)$ for any $k\neq \ell$ is 0 if $\inf T_{|k} \geq \sup T_{|\ell}$, and is
    $F_X(t_k)-F_X(t_{k-1})$ if $\sup T_{|k}\leq \inf T_{|\ell}$. Hence,
    \begin{align*}
      \W(u)&=\P\bigl(T_{|\ell}(X)\leq T_{|\ell}(F_X^{-1}(u)), X\in(t_{\ell-1}, t_\ell]\bigr)\\
           &\phantom{={}}+\sum_{k\neq \ell}\I\{\sup T_{|k}\leq \inf T_{|\ell}\}(F_X(t_k)-F_X(t_{k-1}))\\
           &=u-F_X(t_{\ell-1})+ \sum_{k\neq \ell}\I\{\sup T_{|k}\leq \inf T_{|\ell}\}(F_X(t_k)-F_X(t_{k-
             1})).
    \end{align*}
    The proof when $T_{|\ell}$ is decreasing follows similarly and one has $\W(u)=-
    u+F_X(t_\ell)+\sum_{k\neq \ell}\I\{\sup T_{|k}\leq \inf T_{|\ell}\}(F_X(t_k)-F_X(t_{k-1}))$. Hence
    $\W$ is piecewise linear.
  \item Since $F_X(X)\sim\U(0, 1)$, \cite{embrechtshofert2013c} implies that
    \begin{align*}
      \W(u)=F_{T(X)}\bigl(T(F_X^{-1}(u))\bigr)=F_{T(X)}\bigl(F_X(F_X^{-1}(u))\bigr)=F_{T_X(X)}(u)=F_{F_X(X)}(u)=u.
    \end{align*}
  \item For any fixed $\ell'\in \{1, \dots, K'\}$, we partition $\{1, \dots, K'\}$ into three index
    sets $I_1, I_2, I_3$ such that if $k'\in I_i$, condition $i$ in the statement holds on
    $T_{|k'}$. Then, similarly to 1), one has
    \begin{align*}
      \W(u)&=\sum_{k'=1}^{K'} \P\bigl(T_{|k'}(X)\leq T_{|\ell'}(F_X^{-1}(u)),
             X\in(t'_{k'-1}, t'_{k'}]\bigr).
    \end{align*}
    If $k'\in I_1$, the joint probability $\P\bigl(T_{|k'}(X)\leq T_{|\ell'}(F_X^{-1}
    (u)), X\in(t'_{k'-1}, t'_{k'}]\bigr)$ is 0 if $\inf T_{|k'} \geq \sup T_{|\ell'}$, and is $F_X(t'_{k'})-
    F_X(t'_{k'-1})$ if $\sup T_{|k'}\leq \inf T_{|\ell'}$. Assume now $T_{|\ell'}$ is strictly increasing. If
    $k'\in I_2$, we have
    \begin{align*}
      &\mathrel{\phantom{=}}\P\bigl(T_{|k'}(X)\leq T_{|\ell'}(F_X^{-1}(u)), X\in(t'_{k'-1}, t'_{k'}]\bigr)\\
      &=\P\bigl(T_{|\ell'}(X+t'_{k'}-t'_{\ell'})\leq T_{|\ell'}(F_X^{-1}(u)), X\in(t'_{k'-1}, t'_{k'}]\bigr)\\
      &=\frac{(t'_{K'}-t'_0)u+t_0'+t'_{\ell'}-t'_{k'}-t'_{k'-1}}{t'_{K'}-t'_0}=
        u+\frac{t'_{\ell'}-t'_{k'}-t'_{k'-1}+t_0'}{t'_{K'}-t'_0}.
    \end{align*}
    If $k'\in I_3$, we have
    \begin{align*}
      &\mathrel{\phantom{=}}\P\bigl(T_{|k'}(X)\leq T_{|\ell'}(F_X^{-1}(u)),X\in(t'_{k'-1}, t'_{k'}]\bigr)\\
      &=\P\bigl(T_{|\ell'}(t'_{\ell'}-X+t'_{k'-1})\leq T_{|\ell'}(F_X^{-1}(u)), X\in(t'_{k'-1}, t'_{k'}]\bigr)\\
      &=\frac{t'_{k'}-t'_{k'-1}-t'_{\ell'}+t_0'+(t'_{K'}-t'_0)u}{t'_{K'}-t'_0}=
        u+\frac{t'_{k'}-t'_{k'-1}-t'_{\ell'}+t_0'}{t'_{K'}-t'_0}.
    \end{align*}
    Otherwise if $T_{|\ell'}$ is decreasing, then for $k'\in I_2$, one has $\P\bigl(T_{|k'}(X)\leq
    T_{|\ell'}(F_X^{-1}(u)), X\in(t'_{k'-1}, t'_{k'}]\bigr)=-u+\frac{2t'_{k'}-t'_{\ell'}-t_0'}{t'_{K'}-t'_0}$.
    And for $k'\in I_3$, one has $\P\bigl(T_{|k'}(X)\leq T_{|\ell'}(F_X^{-1}(u)),
    X\in(t'_{k'-1}, t'_{k'}]\bigr)=-u+\frac{t'_{\ell'}-t_0'}{t'_{K'}-t'_0}$.
    Combining all cases one sees that $\W(u)$ is a linear function in $u$ with
    absolute slope $\bigl{|}|I_2|+|I_3|\bigl{|}$. Hence we are done. \qedhere
  \end{enumerate}
\end{proof}

\begin{proof}[Proof of Proposition~\ref{prop:linear:period}]
  Assume $\W$ is $p$-periodic. We first prove that $\W$ is bijective almost everywhere.
  For a Lebesgue null set $N$ and $u_1, u_2\in [0, 1]\setminus N$, $\W(u_1)=\W(u_2)$
  implies that $\W^{p}(u_1)=\W^p(u_2)$ which in turn implies that $u_1=u_2$.
  Hence $\W$ is injective on $[0, 1]\setminus N$. On the other hand, for any $v\in [0,
  1]\setminus\W^{p-1}(N)$, let $u=\W^{p-1}(v)$, then $\W
  (u)=\W(\W^{p-1}(v))=\W^{p}(v)=v$. It follows that $\W$ is
  bijective on $[0, 1]\setminus N$. By Lemma~\ref{lem:unif:preserv}, since
  $\W$ is piecewise differentiable and uniformity-preserving, $\sum_{u\in\W^{-1}(v)}
  \frac{1}{|\W^\prime(u)|}=1$ which implies that $|\W^\prime(u)|=1$
  for all but a finite number of points $v\in [0,1]$. Since $\W$ has all its
  pieces defined in non-degenerate intervals, $\W$ is piecewise linear.
\end{proof}

\begin{proof}[Proof of Lemma~\ref{lem:param:derivative}]
  We prove the first case only, the second follows similarly. For any $u\in[0,1]$ and
  for $\ell$ such that $F_X^{-1}(u)\in[t_{\ell-1}, t_\ell)$, differentiate
  \eqref{eq:W:param:fam} with respect to $u$ to get
  \begin{align*}
    \W'_{\bm{t},\bm{1}, F_X}(u) &=\sum_{k=1}^{K}f_X\Bigl(T_{|\ell}
                                  (F^{-1}_X(u))t_k+\bigl(1-T_{|\ell}(F_X^{-1}(u))\bigr)t_{k-1}\Bigr)
                                  \biggl(\dfrac{T'_{\ell}(F_X^{-1}(u))(t_k-t_{k-1})}{f_X(F_X^{-1}(u))}\biggr)\\
                                &=\sum_{k=1}^Kf_X\biggl(\frac{F_X^{-1}(u)-t_{\ell-1}}{t_{\ell}-t_{\ell-1}}(t_k-t_{k-1}) +
                                  t_{k-1}\biggr)\dfrac{t_k-t_{k-1}}{(t_\ell-t_{\ell-1})f_X(F_X^{-1}(u))}.
  \end{align*}
  If $f_X(0+)=\infty$, then
  \begin{align*}
    \W'_{\bm{t},\bm{1}, F_X}(0+) &= \sum_{k=1}^Kf_X\biggl(
                                   \frac{F_X^{-1}(0+)-t_0}{t_{1}-t_{0}}(t_k-t_{k-1}) +
                                   t_{k-1}\biggr)\dfrac{t_k-t_{k-1}}{(t_1-t_{0})f_X(F_X^{-1}(0+))}\\
                                 &=\frac{f_X(F_X^{-1}(0+))}{f_X(F_X^{-1}(0+))} +\sum_{k=2}^K
                                   \dfrac{(t_{k}-t_{k-1})\Bigl(f_X((t_k-t_{k-1})\frac{F_X^{-1}(0+)-t_0}{t_1-t_0}+t_{k-1})\Bigr)}
                                   {(t_1-t_0)f_X(F_X^{-1}(0+))}=1,
  \end{align*}
  where the last equation follows by $f_X(F_X^{-1}(0+))=f_X(t_0+)=\infty$ and $(t_k-t_{k-1})
  \frac{F_X^{-1}(0+)-t_0}{t_1-t_0}+t_{k-1} > t_0$ and therefore $f_X\bigl((t_k-t_{k-1})
  \frac{F_X^{-1}(0+)-t_0}{t_1-t_0}+t_{k-1}\bigr)<\infty$.
\end{proof}

\subsection{Proofs of Section~\ref{sec:general:W:trafo}}
\begin{proof}[Proof of Proposition~\ref{prop:jump:linear}]
  By assumption, $\P(T(X)=T(x_0)) > 0$. Then $F_X(x_0, V) = F_X(x_0-)+(F_X(x_0)-F_X(x_0-))V$
  and
  \begin{align*}
    F_{T(X)}(T(x_0), V) = F_{T(X)}(T(x_0)-)+(F_{T(X)}(T(x_0))-F_{T(X)}(T(x_0)-))V. %
  \end{align*}
  Therefore, the generalised W-transform from $F_X(x_0, V)$ to $F_{T(X)}(T(x_0),$ $W)$ is
  \begin{align*}
    \Wg(u)=\dfrac{F_{T(X)}(T(x_0))-F_{T(X)}(T(x_0)-)}{F_X(x_0)-F_X(x_0-)}(u-
    F_X(x_0-))+F_{T(X)}(T(x_0)-),
  \end{align*}
  where $u \in (F_X(x_0-), F_X(x_0))$. It follows that $\Wg$ is linear on
  $(F_X(x_0-), F_X(x_0))$ with slope
  \begin{align*}
    \dfrac{F_{T(X)}(T(x_0))-F_{T(X)}(T(x_0)-)}{F_X(x_0)-F_X(x_0-)}
    &=\dfrac{F_{T(X)}(s)-F_{T(X)}(s-)}{F_X(x_0)-F_X(x_0-)}=\dfrac{\sum_{\ell=0}^L\P(X=x_\ell)}{\P(X=x_\ell)}. \qedhere
  \end{align*}
\end{proof}

\subsection{Proofs of Section~\ref{sec:cop}}

\begin{proof}[Proof of Theorem~\ref{theorem:U:V:cop}]
  \mbox{} %
  \begin{enumerate}
  \item Suppose $u\in(\delta_{\ell-1}, \delta_\ell]$ for some $\ell\in\{1, \dots, K\}$.
    Consider the construction of $S_k$ in Proposition~\ref{prop:W:property} \ref{prop:W:property:partition:sq}
    and let $R_\ell\coloneqq\{u: \W_{|\ell}(u)\leq v\}$. Then $R_\ell=\bigl(\delta_{\ell-1},
    \min\{u, \W^{-1}_{|\ell}(v)\}\bigr]$ if $\ell\in I$ and $R_\ell=\bigl
    (\W^{-1}_{|\ell}(v), u\bigr]$ otherwise.
    By construction, the event $\{U\leq u, V\leq v\}=\{U\leq u, \W(U)\leq v\}$ is thus equivalent
    to $\bigl(\biguplus_{k=1}^{\ell-1}S_k\bigr)\cup R_\ell$. Therefore,
    \begin{align*}
      &\phantom{{}={}}\P(U\leq u, V\leq v)=\P\biggl(\biggl(\,\biguplus_{k=1}^{\ell-1}S_k\biggr)\cup R_\ell\biggr)\\
      &=\sum_{\substack{k\in I, \\ k<\ell}} (\W^{-1}_{|k}(v)-\delta_{k-1}) + \sum_{\substack{k\in I^C, \\ k<\ell}}(\delta_k-\W^{-1}_{|k}(v))\\
      &\phantom{{}={}}+\I_{\{\ell\in I\}}\bigl(\min\{u, \W^{-1}_{|\ell}(v)\}-\delta_{\ell-1}\bigr)+
        \I_{\{\ell\in I^C\}}\max\{u-\W^{-1}_{|\ell}(v),0\}\\
      &=\sum_{k\in I}\max\bigl\{\min\{u,\W^{-1}_{|k}(v)\}-\delta_{k-1}, 0\bigr\} +\sum_{k\in I^C}\max
        \bigl\{\min\{\delta_k,u\}-\W^{-1}_{|k}(v), 0\bigr\}.
    \end{align*}
  \item Since $\P(U\leq u, V=v)= \frac{\rd}{\rd v}C(u, v)$, \eqref{eq:joint:U:V} follows upon
    differentiating \eqref{eq:U:V:cop} with respect to $v$. One recognises \eqref{eq:joint:U:V}
    to be the distribution function of a multinomial distribution with event probabilities $p_1,
    \dots, p_{|N(v)|}$ where $p_k =\bigl{|}\frac{\rd}{\rd v}\W_{|k}^{-1}(v)
    \bigl{|}$ for all $k\in N(v)$. Hence the statement follows. \qedhere
  \end{enumerate}
\end{proof}

\begin{proof}[Proof of Proposition~\ref{prop:sto:id}]
  By definition, $\W^{-1}(v, U')=\W^{-1}_{|k}(v)$ for some $k\in N(v)$, so
  $\W(\W^{-1}(v, U'))=\W_{|k}(\W^{-1}_{|k}(v))
  =v$. Hence $\W(\W^{-1}(V, U'))=V$. Moreover, for any $k\in N(v)$,
  $\P(\W^{-1}(V, U')=\W^{-1}_{|k}(v)\,|\,V=v)=\P(U'\in
  (\sum_{\ell=1}^{k-1}p_\ell, \sum_{\ell=1}^{k}p_\ell]\,|\,V=v)=p_k$.
  It follows that the pair $(\W^{-1}(V, U'), V)$ has the same conditional distribution as the
  one specified in Theorem~\ref{theorem:U:V:cop}~\ref{theorem:U:V:cop:cond} and $(\W^{-1}(V, U'), V)$ is distributed
  according to the copula $C$ in \eqref{eq:U:V:cop}. Hence $U\sim\U(0, 1)$.
\end{proof}

\begin{proof}[Proof of Theorem~\ref{theorem:W:trafo:cop}]
  For any $j$, consider the event $\{\W_j(U_j)\leq u_j\}$. Then by
  Proposition~\ref{prop:W:property}~\ref{prop:W:property:partition:sq}, we have
  $\{\W_j(U_j)\leq u_j\}=\bigl(\bigcup_{k_j\in I_j} (\delta_{j, k_j-1},
  \W^{-1}_{j|k_j}(u_j)]\bigr)\cup\bigl(\bigcup_{k_j\notin I_j}
  (\W^{-1}_{j|k_j}(u_j),\delta_{j, k_j}]\bigr)$. By definition, the
  joint distribution function of
  $(\W_1(U_1),\dots, \W_d(U_d))$ is thus
  \begin{align*}
    C_{\bm{\W}}(u_1,\dots,u_d)&=\P(\W_1(U_1)\leq u_1, \dots,
                                \W_d(U_d)\leq u_d)\\
                              &=\P\biggl(\,\bigcap_{j=1}^d\biggl[\biggl(\,\bigcup_{k_j\in I_j}(\delta_{j, k_j-1}, \W^{-1}_{j|
                                k_j}(u_j)]\biggr)\cup\biggl(\,\bigcup_{k_j\notin I_j}(\W^{-1}_{j|k_j}(u_j),\delta_{j, k_j}]\biggr)
                                \biggr]\biggr)\\
                              &=\P\biggl(\,\bigcap_{j=1}^d\mathcal{I}_j^{(k_j)}\biggr)
                                =\sum_{k_d=1}^{K_d}\dots\sum_{k_1=1}^{K_1}\Delta_{B_{\bm{\delta}_{\bm{k}},\bm{\W^{-1}(\bm{u})},\bm{I}}}C.
  \end{align*}
  The density follows from \eqref{eq:W:trafo:cop} by an application of the chain rule. Note that if
  $\W^{-1}_{j|k_j} (u_j)\in\{\delta_{j,k_j-1}, \delta_{j,k_j}\}$, then the differentiation gives 0.
  Otherwise, for each $B_{\bm{\delta}_{\bm{k}},\bm{\W^{-1}},\bm{I}}$,
  \begin{align*}
    \frac{\partial}{\partial u_1\cdots\partial u_d}\Delta_{B_{\bm{\delta}_{\bm{k}},\bm{\W^{-1}},\bm{I}}}C=
    &(-1)^{d+\sum_{m=1}^d\I_{\{k_m\in I_m\}}}c(\W^{-1}_{1|k_1}(u_1),\dots, \W^{-1}_{d|k_d}(u_d))\\
    &\times\prod_{\ell=1}^d\frac{1}{\W'_{\ell|k_\ell}(\W^{-1}_{\ell|k_\ell}(u_\ell))}.
  \end{align*}
  Hence the result follows.
\end{proof}

\begin{proof}[Proof of Proposition~\ref{prop:W:volume}]
  Let $I_j$ be the index set such that $\W_{j|k_j}$ is
  increasing (decreasing) if and only if $k_j\in I_j$ ($I_j^C$). Then the $C_{\bm{\W}}$-volume
  of $(\bm{a},\bm{b}]$ is
  \begin{align*}
    V_{C_{\bm{\W}}}((\bm{a},\bm{b}])&=\P(a_1<\W_1(U_1)\leq
                                      b_1,\dots, a_d<\W_d(U_d)\leq b_d)\\
                                    &=\P\biggl(\,\bigcap_{j=1}^{d}\biguplus_{k_j=1}^{K_j}\{a_j<\W_{j|k_j}(U_j)\leq b_j, U_j
                                      \in (\delta_{j,k_j}, \delta_{j,k_j+1}]\}\biggl{)}\\
                                    &=\P\biggl(\,\bigcap_{j=1}^d\biguplus_{k_j=1}^{K_j}\biggl\{\W_{j|k_j}^{-1}\bigl(a_j^{\I\{k_j\in
                                      I_j\}}b_j^{\I\{k_j\notin I_j\}}\bigr)< U_j\leq
                                      \W_{j|k_j}^{-1}\bigl(a_j^{\I\{k_j\notin I_j\}}b_j^{\I\{k_j\in I_j\}}\bigr),\\
                                    &\phantom{{}=\P\biggl(\,\bigcap_{j=1}^d\biguplus_{k_j=1}^{K_j}\biggl\{}U_j\in (\delta_{j,k_j}, \delta_{j,k_j+1}]\biggr\}\biggr)\\
                                    &=\sum_{k_d=1}^{K_d}\cdots\sum_{k_1=1}^{K_1}\P\biggl(\,\bigcap_{j=1}^d
                                      \bigl(\W_{j|k_j}^{-1}\bigl(a_j^{\I\{k_j\in I_j\}}b_j^{\I\{k_j\notin I_j\}}\bigr),
                                      \W_{j|k_j}^{-1}\bigl(a_j^{\I\{k_j\notin I_j\}}b_j^{\I\{k_j\in I_j\}}\bigr)\bigr]\biggr)\\
                                    &=\sum_{k_d=1}^{K_d}\cdots\sum_{k_1=1}^{K_1}
                                      \Delta_{B_{\bm{\W^{-1}}(\bm{a}),\bm{\W^{-1}}(\bm{b}),\bm{I}}}C. \qedhere
  \end{align*}
\end{proof}

\begin{proof}[Proof of Proposition~\ref{prop:flip:v:trans:tail}]
  By Example~\ref{eg:w:trans:volume}~\ref{eg:w:trans:volume:v:trans},
  \begin{align*}
    C_{\bm{\mathcal{V}}^*}(pb,\tfrac{p}{b})&=\Delta_{(0,\mathcal{V}^{*\,-1}_{1|1}(pb)]\times(0,\mathcal{V}^{*\,-1}_{2|1}(\frac{p}{b})]}C+\Delta_{(0,\mathcal{V}^{*\,-1}_{1|1}(pb)]\times(\mathcal{V}^{*\,-1}_{2|2}(\frac{p}{b}),1]}C\\
                                           &\phantom{{}=}+\Delta_{(\mathcal{V}^{*\,-1}_{1|2}(pb),1]\times(0,\mathcal{V}^{*\,-1}_{2|1}(\frac{p}{b})]}C+\Delta_{(\mathcal{V}^{*\,-1}_{1|2}(pb),1]\times(\mathcal{V}^{*\,-1}_{2|2}(\frac{p}{b}),1]}C.
  \end{align*}
  Expanding the terms, interchanging the second and the third, and dividing them by $p$ results in
  \begin{align*}
    \frac{C_{\bm{\mathcal{V}}^*}(pb,\frac{p}{b})}{p}&=
                                                      \frac{C(\mathcal{V}^{*\,-1}_{1|1}(pb),\mathcal{V}^{*\,-1}_{2|1}(\frac{p}{b}))}{p}+\frac{\mathcal{V}^{*\,-1}_{2|1}(\frac{p}{b})-C(\mathcal{V}^{*\,-1}_{1|2}(pb), \mathcal{V}^{*\,-1}_{2|1}(\frac{p}{b}))}{p}\\
                                                    &\phantom{{}=}+\frac{\mathcal{V}^{*\,-1}_{1|1}(pb)- C(\mathcal{V}^{*\,-1}_{1|1}(pb), \mathcal{V}^{*\,-1}_{2|2}(\frac{p}{b}))}{p}\\
                                                    &\phantom{{}=}+\frac{1-\mathcal{V}^{*\,-1}_{1|2}(pb)-\mathcal{V}^{*\,-1}_{2|2}(\frac{p}{b})+C(\mathcal{V}^{*\,-1}_{1|2}(pb), \mathcal{V}^{*\,-1}_{2|2}(\frac{p}{b}))}{p}.
  \end{align*}
  By assumption, $C$ is tail independent in the upper-left, upper and lower-right tails,
  so $\Lambda\equiv0$ in these three regions. %
  We thus obtain that
  \begin{align*}
    \Lambda\biggl(b,\frac{1}{b};C_{\bm{\mathcal{V}}^*}\biggr)
    =\lim_{p\to0+}\frac{C_{\bm{\mathcal{V}}^*}(pb,\frac{p}{b})}{p}=\lim_{p\to0+}\frac{C(\mathcal{V}^{*\,-1}_{1|1}(pb),\mathcal{V}^{*\,-1}_{2|1}(\frac{p}{b}))}{p}.
  \end{align*}
  By Taylor expansions,
  $\mathcal{V}^{*\,-1}_{1|1}(pb)=(\mathcal{V}^{*\,-1}_{1|1})'(0+)pb+o(p)$
  and
  $\mathcal{V}^{*\,-1}_{2|1}(p/b)=(\mathcal{V}^{*\,-1}_{2|1})'(0+)p$ $/b+o(p)$.
  Hence, by Lipschitz continuity of copulas,
  \begin{align*}
    &\Bigl|C\Bigl((\mathcal{V}^{*\,-1}_{1|1})'(0+)pb+o(p),\ (\mathcal{V}^{*\,-1}_{2|1})'(0+)\frac{p}{b}+o(p)\Bigr)\\
    &\phantom{\Bigl|}-C\Bigl((\mathcal{V}^{*\,-1}_{1|1})'(0+)pb,\ (\mathcal{V}^{*\,-1}_{2|1})'(0+)\frac{p}{b}\Bigr)\Bigr|\le 2|o(p)|.
  \end{align*}
  Now $\lim_{p\to0+}o(p)/p=0$ implies that
  \begin{align}
    \Lambda\biggl(b,\frac{1}{b};C_{\bm{\mathcal{V}}^*}\biggr)&=\lim_{p\to0+}\frac{C\bigl((\mathcal{V}^{*\,-1}_{1|1})'(0+)pb+o(p),(\mathcal{V}^{*\,-1}_{2|1})'(0+)\frac{p}{b}+o(p)\bigr)}{p}\notag\\
                                                             &=\lim_{p\to0+}\frac{C(p\alpha_1b,\frac{p\alpha_2}{b})}{p}=\Lambda\biggl(\alpha_1b,\frac{\alpha_2}{b};C\biggr), \label{eq:proof:MTCM:Lam}
  \end{align}
  where $\alpha_1=(\mathcal{V}^{*\,-1}_{1|1})'(0+)$ and $\alpha_2=(\mathcal{V}^{*\,-1}_{2|1})'(0+)$.
  By \cite[Equation~(2.2)]{koikekatohofert2023}, the supremum of $\Lambda(b,\frac{1}{b};
  C_{\bm{\mathcal{V}}^*})$ is attained when $(\alpha_1b,\alpha_2/b)$ is proportional to
  $(b^*,1/b^*)$, that is
  \begin{align*}
    (\alpha_1b,\alpha_2/b)=t(b^*,1/b^*)
  \end{align*}
  for some $t>0$. Solving this with respect to $b$ gives
  $b^*_{C_{\mathcal{V}^*}}=\sqrt{\frac{\alpha_2}{\alpha_1}}b^*$ (and
  $t=\sqrt{\alpha_1\alpha_2}$). Plugging $b=b^*_{C_{\mathcal{V}^*}}$ into~\eqref{eq:proof:MTCM:Lam}
  and using that $\Lambda$ is homogeneous of order $1$, we obtain
  \begin{align*}
    \Lambda\biggl(b^*_{C_{\mathcal{V}^*}},\frac{1}{b^*_{C_{\mathcal{V}^*}}}; C_{\mathcal{V}^*}\biggr)&=\Lambda\biggl(\alpha_1b^*_{C_{\mathcal{V}^*}},\frac{\alpha_2}{b^*_{C_{\mathcal{V}^*}}};C\biggr)=\Lambda\biggl(\alpha_1\sqrt{\frac{\alpha_2}{\alpha_1}}b^*,\frac{\alpha_2}{\sqrt{\frac{\alpha_2}{\alpha_1}}b^*};C\biggr)\\
                                                                                                     &=\Lambda\biggl(\sqrt{\alpha_1\alpha_2} b^*, \frac{\sqrt{\alpha_1\alpha_2}}{b^*};C\biggr)=\sqrt{\alpha_1\alpha_2}\Lambda\biggl(b^*,\frac{1}{b^*}; C\biggr).\qedhere
  \end{align*}
\end{proof}

\begin{proof}[Proof of Proposition~\ref{prop:tail:coeff}]
  Using
  \begin{align*}
    C_{\mathcal{V}}(t,t)&=C(\mathcal{V}^{-1}_{|1}(t)+t, \mathcal{V}^{-1}_{|1}(t)+t)-C(\mathcal{V}^{-1}_{|1}(t)+t, \mathcal{V}^{-1}_{|1}(t))\\
                        &\phantom{{}=}-C(\mathcal{V}^{-1}_{|1}(t), \mathcal{V}^{-1}_{|1}(t)+t)+C(\mathcal{V}^{-1}_{|1}(t), \mathcal{V}^{-1}_{|1}(t)),
  \end{align*}
  dividing by $1-t$, subtracting
  $\frac{2\mathcal{V}^{-1}_{|1}(t)}{1-t}$ from the first summand and adding it thereafter, we obtain that
  \begin{align*}
    &\phantom{{}={}}\frac{1-2t+C_{\mathcal{V}}(t,t)}{1-t}\\
    &=\frac{1-2(\mathcal{V}^{-1}_{|1}(t)+t)+C(\mathcal{V}^{-1}_{|1}(t)+t, \mathcal{V}^{-1}_{|1}(t)+t)}{1-t}+\frac{2\mathcal{V}^{-1}_{|1}(t)}{1-t}\\
                                         &\phantom{{}={}}-\frac{C(\mathcal{V}^{-1}_{|1}(t)+t, \mathcal{V}^{-1}_{|1}(t))}{1-t}-\frac{C(\mathcal{V}^{-1}_{|1}(t), \mathcal{V}^{-1}_{|1}(t)+t)}{1-t}+\frac{C(\mathcal{V}^{-1}_{|1}(t), \mathcal{V}^{-1}_{|1}(t))}{1-t}.
  \end{align*}
  Multiplying and dividing the first summand by $1-(\mathcal{V}^{-1}_{|1}(t)+t)$ and the last one by $\mathcal{V}^{-1}_{|1}(t)$ leads to
  \begin{align*}
    &\frac{1-2(\mathcal{V}^{-1}_{|1}(t)+t)+C(\mathcal{V}^{-1}_{|1}(t)+t, \mathcal{V}^{-1}_{|1}(t)+t)}{1-(\mathcal{V}^{-1}_{|1}(t)+t)}\frac{1-(\mathcal{V}^{-1}_{|1}(t)+t)}{1-t}+\frac{2\mathcal{V}^{-1}_{|1}(t)}{1-t}\\
    &-\frac{C(\mathcal{V}^{-1}_{|1}(t)+t, \mathcal{V}^{-1}_{|1}(t))}{1-t}-\frac{C(\mathcal{V}^{-1}_{|1}(t), \mathcal{V}^{-1}_{|1}(t)+t)}{1-t}+\frac{C(\mathcal{V}^{-1}_{|1}(t), \mathcal{V}^{-1}_{|1}(t))}{\mathcal{V}^{-1}_{|1}(t)}\frac{\mathcal{V}^{-1}_{|1}(t)}{1-t}.
  \end{align*}
  Now take the limit for $t\to1-$ and use L'H\^opital's rule (leading to $(\mathcal{V}^{-1}_{|1})'(1-)=\frac{1}{\mathcal{V}_{|1}'(\mathcal{V}^{-1}_{|1}(1-))}$ $=\frac{1}{\mathcal{V}'_{|1}(0+)}$) to see that
  \begin{align*}
    \lambda_{\text{u}}^{C_{\mathcal{V}}}&=\lambda_{\text{u}}^C\biggl(1-\frac{1}{-\mathcal{V}_{|1}'(0+)}\biggr)+\frac{2}{-\mathcal{V}_{|1}'(0+)}\\
                                        &\phantom{{}=}-\lim_{t\to1+}\frac{C(\mathcal{V}^{-1}_{|1}(t)+t, \mathcal{V}^{-1}_{|1}(t))+C(\mathcal{V}^{-1}_{|1}(t), \mathcal{V}^{-1}_{|1}(t)+t)}{1-t}+\lambda_{\text{l}}^C\frac{1}{-\mathcal{V}'_{|1}(0+)}.
  \end{align*}
  In particular, if $C$ is tail-independent in the upper-left and lower-right tails, so $\Lambda\equiv0$
  in these regions. We thus obtain that
  \begin{align*}
    \lambda_{\text{u}}^{C_{\mathcal{V}}}&=\lambda_{\text{u}}^C\biggl(1-\frac{1}{-\mathcal{V}_{|1}'(0+)}\biggr)+\frac{2}{-\mathcal{V}_{|1}'(0+)}.\qedhere
  \end{align*}
\end{proof}

\begin{proof}[Proof of Proposition~\ref{prop:os:tail:dependence}]
  Since $C_{S,\W}$ is a copula, $C_{S,\W}(u_1, 1)=\sum_{k=1}^K(\delta_k-\delta_{k-1})G_k(u_1)=u_1$
  and $C_{S,\W}(1, u_2)=\sum_{k=1}^K(\delta_k-\delta_{k-1})G_k(u_2)=u_2$.
  It follows that
  \begin{align}
    \sum_{k=1}^K(\delta_k-\delta_{k-1})g_k(u_j)=1, \quad j=1,2.\label{eq:density:relation}
  \end{align}
  Therefore,
  \begin{align*}
    \lambda_{\text{l}}&=\lim_{t\rightarrow0+}\dfrac{C_{S,\W}(t, t)}{t}
                        =\sum_{k=1}^K(\delta_k-\delta_{k-1})\lim_{t\rightarrow0+}
                        \dfrac{C_k(G_k(t), G_k(t))}{t}\\
                      &=\sum_{k=1}^K(\delta_k-\delta_{k-1})\lim_{t\rightarrow0+}
                        \dfrac{G_k(t)}{t}\dfrac{C_k(G_k(t),G_k(t))}{G_k(t)}=\sum_{k=1}^K(\delta_k-\delta_{k-1})g_k(0+)\lambda_{\text{l},k}
                        =\sum_{k=1}^K\alpha_k\lambda_{\text{l},k},
  \end{align*}
  and
  \begin{align*}
    \lambda_{\text{u}}&=\lim_{t\rightarrow1-}\dfrac{C_{S,\W}(t, t)+1-2t}{t}\\
                      &=\sum_{k=1}^K(\delta_k-\delta_{k-1})\lim_{t\rightarrow1-}\dfrac{1-G_k(t)}
                        {1-t}\dfrac{C_k(G_k(t), G_k(t))+1-2G_k(t)}{1-G_k(t)}\\
                      &\phantom{{}=}+\lim_{t\rightarrow1-}\frac{1-2t-\sum_{k=1}^K
			(\delta_k-\delta_{k-1})(1-2G_k(t))}{1-t}\\
                      &= \sum_{k=1}^K(\delta_k-\delta_{k-1})g_k(1-)\lambda_{\text{u},k}+2-2\sum_{k=1}^K(\delta_k-\delta_{k-1})g_k(1-)= \sum_{k=1}^K\beta_k\lambda_{\text{u},k},
  \end{align*}
  where the last equality follows from~\eqref{eq:density:relation} upon letting $u_1,u_2\to 1-$.
\end{proof}

\begin{proof}[Proof of Proposition~\ref{prop:exchan}]
  Let $C$ be exchangeable, so
  $C(u_{\sigma(1)}, \dots, u_{\sigma(d)})=C(u_1, \dots, u_d)$ for any
  permutation $\sigma$ of $\{1, 2, \dots, d\}$. Fix a permutation
  $\sigma$. Then
  \begin{align*}
    &\phantom{{}={}}C_{\W}(u_{\sigma(1)}, \dots, u_{\sigma(d)})\\
    &=\sum_{k_{\sigma(d)}=1}^K\dots\sum_{k_{\sigma(1)}=1}^K V_C\biggl(\,\prod_{j=1}^d \biggl(\,\biguplus_{k_{\sigma(j)}\in I} (\delta_{k_{\sigma(j)}-1},\W^{-1}_{|k_{\sigma(j)}}(u_{\sigma(j)})]\biggr)\cup\\
    &\phantom{=\sum_{k_{\sigma(1)}=1}^K\cdots \sum_{k_{\sigma(d)}=1}^KV_C\biggl(\,\prod_{j=1}^d \biggl(\,}\biggl(\,\biguplus_{k_{\sigma(j)}\notin I}
      (\W^{-1}_{|k_{\sigma(j)}}(u_{\sigma(1)}),
      \delta_{k_{\sigma(j)}}]\biggr)\biggr).
  \end{align*}
  Since $\sigma$ is a bijection and $\{\delta_k\}$ is identical across dimensions, we re-index
  $k_{\sigma(j)}\leftarrow k_j$, leading to
  \begin{align*}
    &\phantom{{}={}}C_{\W}(u_{\sigma(1)}, \dots, u_{\sigma(d)})\\
    &=\sum_{k_d=1}^K\cdots\sum_{k_1=1}^K V_C\biggl(\,\prod_{j=1}^d \biggl(\,\biguplus_{k_j\in
      I}(\delta_{k_j-1}, \W^{-1}_{|k_j}(u_j)]\biggr)\cup\biggl(\,\biguplus_{k_j\notin I}
      (\W^{-1}_{|k_j}(u_j), \delta_{k_j}]\biggr)\biggr)\\
    & = C_{\W}(u_1, u_2, \dots, u_d). \qedhere
  \end{align*}
\end{proof}

\printbibliography[heading=bibintoc]
\end{document}

%
%
%
%
